\documentclass[onecolumn, draftclsnofoot,11pt]{IEEEtran}

\IEEEoverridecommandlockouts

\usepackage[english]{babel}

\usepackage{bbm}
\usepackage{physics}
\usepackage{amsthm}
\usepackage{amssymb}
\usepackage{amsmath}
\usepackage{graphicx}
\usepackage[colorlinks=true, allcolors=red]{hyperref}
\usepackage{cite}
\newtheorem{definition}{Definition}
\newtheorem{lemma}{Lemma}

\newtheorem{theorem}{Theorem}
\newtheorem{remark}{Remark}
\newtheorem{prop}{Proposition}
\newtheorem{example}{Example}

\global\long\def\RR{\mathbb{R}}

\global\long\def\EE{\mathbb{E}}
\global\long\def\PP{\mathbb{P}}
\global\long\def\FF{\mathbb{F}}
\global\long\def\11{\mathbbm{1}}

\def \calA {\mathcal{A}}
\def \calB {\mathcal{B}}
\def \calC {\mathcal{C}}
\def \calD {\mathcal{D}}
\def \calE {\mathcal{E}}

\def \calH {\mathcal{H}}
\def \calI {\mathcal{I}}

\def \calK {\mathcal{K}}
\def \calL {\mathcal{L}}
\def \calM {\mathcal{M}}
\def \calN {\mathcal{N}}

\def \calP {\mathcal{P}}

\def \calR {\mathcal{R}}
\def \calS {\mathcal{S}}

\def \calU {\mathcal{U}}

\def \calW {\mathcal{W}}
\def \calX {\mathsf{X}}

\def \RR {\mathbb{R}}
\def \I {I}

\newcommand{\curlys}[1]{\{{#1}\}}
\newcommand{\curly}[1]{\left\{{#1}\right\}}
\newcommand{\sq}[1]{\left[{#1}\right]}
\newcommand{\round}[1]{\left({#1}\right)}
\newcommand{\absolute}[1]{\left|{#1}\right|}

\def \bV {W}
\def \bA {{A_R}}
\def \bB {{B_R}}
\def \bE {\bar{E}}

\def \bAn {A_R^n}
\def \bBn {B_R^n}

\def \px {\bar{P}_X}
\def \pxhat {{P}_{\hat{X}}}
\def \prevTC {\bar{P}_{X|\hat{X}}}
\def \Tx {T_\delta({X})}
\def \Txhat {T_\delta(\hat{X})}
\def \Txxhat {T_\delta(X,\hat{X})}
\def \Txcond{T_\delta(X|\hat{x}^n)}
\def \codebook{\calC^{(n)}}

\def \I {\mathbbm{1}}
\def \codeDistribution {\lambda_{\hat{X}^n}(\hat{x}^n)}

\def\tensor{\otimes}
\def\rhohata{\hat{\rho}_a^{\bB}}
\def\rhohatan{\hat{\rho}_{a^n}^{\bB}}
\def\rhotildea{\tilde{\rho}_{a}^{\bB}}
\def\rhotildean{\tilde{\rho}_{a^n}^{\bB}}
\def\lambdaaA{\lambda_a^{A}}
\def\lambdaanA{\lambda_{a^n}^{A}}
\def\lambdaaB{\lambda_a^{\bB}}

\def\lambdaaE{\lambda_a^{E}}
\def\TDelta{\mathcal{T}_{\delta}^{(n)}}

\def\zetahat{|{\hat\zeta}\rangle}
\def \eqand {\mbox{ and }}
\def \phirho {\ket{\psi_\rho}^{B\bB}}

\def \epsFourRoot {\sqrt{\epsilon}}
\def \Imset{\calI_\calE^{(m)}}

\def \Wchannel {\calN_{W}}

\def\IndSP{\mathbbm{1}_{\{\mbox{sP}\}}}

\def\normMK{\frac{1}{\sqrt{(1-\epsFourRoot)|\calM||\calK|}}}
\def\normsqMK{\frac{1}{{(1-\epsFourRoot)|\calM||\calK|}}}

\usepackage{stackengine}
\def\deq{\mathrel{\ensurestackMath{\stackon[1pt]{=}{\scriptstyle\Delta}}}}

\def \calA {\mathcal{A}}
\def \calB {\mathcal{B}}
\def \calC {\mathcal{C}}
\def \calD {\mathcal{D}}
\def \calE {\mathcal{E}}

\def \calH {\mathcal{H}}
\def \calI {\mathcal{I}}

\def \calK {\mathcal{K}}
\def \calL {\mathcal{L}}
\def \calM {\mathcal{M}}
\def \calN {\mathcal{N}}

\def \calP {\mathcal{P}}

\def \calR {\mathcal{R}}
\def \calS {\mathcal{S}}

\def \calU {\mathcal{U}}

\def \calW {\mathcal{W}}
\def \calX {\mathsf{X}}

\def \RR {\mathbb{R}}
\def \PP {\mathbb{P}}

\def \Xhat {\hat{X}}
\def \xhat {\hat{x}}
\def \px {{P}_X}
\def \pxhat {{P}_{\hat{X}}}
\def \prevTC {{W}_{X|\hat{X}}}
\def \Tx {\mathcal{T}_{\hat{\delta}}^{(n)}({X})}
\def \Txhat {\mathcal{T}_\delta^{(n)}(\hat{X})}
\def \Txxhat {\mathcal{T}_{\hat{\delta}}^{(n)}(X,\hat{X})}
\def \Txcond{\mathcal{T}_\delta^{(n)}(X|\hat{x}^n)}

\def \codebook{\calC}
\def \encodern {\mathcal{E}^{(n)}}
\def \decodern {\mathcal{D}^{(n)}}
\def \Ipmf {\I_{\curly{\mbox{\normalfont sPMF}}}}
\def \EE {\mathbb{E}}
\def \I {\mathbbm{1}} 
\def \codeDistribution {\mathbb{P}}

\def \Txqc {\mathcal{T}_{{\delta}}^{(n)}({X})}
\def \xn{x^n}
\def \Xn{X^n}
\def \sourcedo{{\rho^{B}}}
\def \refstate{B_{{R}}}
\def\tensor{\otimes}

\def\rhotilde{\tilde{\rho}}

\def \targetpx{P_X}

\def\TDelta{\mathcal{T}_{\delta}}
\def\zetahat{|{\hat\zeta}\rangle}
\def \eqand {\mbox{ and }}

\usepackage{stackengine}
\def\deq{\mathrel{\ensurestackMath{\stackon[1pt]{=}{\scriptstyle\Delta}}}}

\title{Lossy Quantum Source Coding  with  Global Error Criterion using  Posterior Reference Map 
}

\title{\huge 
Lossy Quantum Source Coding with a Global Error Criterion based on  a Posterior Reference Map
}

\author{\IEEEauthorblockN{  Touheed Anwar Atif, Mohammad Aamir Sohail, and S. Sandeep
    Pradhan}\\
\IEEEauthorblockA{Department of Electrical Engineering and Computer Science,\\
University of Michigan, Ann Arbor, MI 48109, USA.\\
Email: \tt touheed@umich.edu, mdaamir@umich.edu, pradhanv@umich.edu}
}

\begin{document}

\maketitle
\vspace{-0.4in}
\begin{abstract}
     We consider the lossy quantum source coding problem where the task is to compress a given quantum source below its von Neumann entropy. Inspired by the duality connections between the rate-distortion and channel coding problems in the classical setting, we propose a new formulation for the lossy quantum source coding problem. 
     This formulation differs from the existing quantum rate-distortion theory in two aspects. 
     Firstly, we require that the reconstruction of the compressed quantum source fulfill a global error constraint as opposed to the sample-wise local error criterion used in the standard rate-distortion setting.
     Secondly, instead of a distortion observable, we employ the notion of a backward quantum  channel, which we refer to as a ``posterior reference map'', to measure the reconstruction error.
     Using these, we characterize the asymptotic performance limit of the lossy quantum source coding problem in terms of single-letter coherent information of the given posterior reference map. 
     We demonstrate a protocol to encode (at the specified rate) and decode, with the reconstruction satisfying the provided global error criterion, and therefore achieving the asymptotic performance limit. 
     The protocol is constructed by decomposing coherent information as a difference of two Holevo information quantities, inspired from prior works in quantum communication problems.
    To further support the findings, we develop analogous formulations for the quantum-classical and classical variants and express the asymptotic performance limit in terms of single-letter mutual information quantities with respect to appropriately defined channels analogous to posterior reference maps. 
    We also provide various examples for the three formulations, and shed light on their connection to the standard rate-distortion formulation wherever possible.
\end{abstract}

\section{Introduction}
\label{QLSC:sec:intro}


A fundamental problem from an information theoretic perspective is the asymptotic characterization of the rate required to compress a source that can be recovered to a certain measurable degree. Such a problem in quantum information theory is referred to as quantum source coding or a quantum data compression problem.
In the lossless regime, Schumacher \cite{schumacher1995quantum,jozsa1994new} proved that a quantum source could be compressed at a rate given by von Neumann entropy while incurring a very small error between the reconstruction and the source state. The error in this model is defined for the entire block, also called as block error or \emph{global error}.
Considering the block error, a strong converse was also proved in the lossless regime \cite{winter1999coding,baghali2022strong}, which states that it is impossible to achieve any rate below von Neumann entropy even when the asymptotic probability of block error is relaxed from being (almost) zero.

As for the lossy regime, where the objective is to further reduce the rate at the expense of increased but bounded error, Barnum  \cite{barnum2000quantum} conjectured  minimal coherent information as a candidate 
in characterizing the asymptotic performance limit. Generalizing the formulation from the classical rate-distortion theory \cite{shannon1959coding}, Barnum in  \cite{barnum2000quantum} introduced a \emph{local} distortion criterion as averaged symbol-wise entanglement fidelity based on  marginal operations (partial trace) between the reconstruction and the reference of the original source. 
In \cite{datta2012quantum}, Datta \emph{et. al} obtained a regularized expression for the quantum rate-distortion
distortion function in terms of the entanglement of purification. Further, the authors also formulated the entanglement-assisted quantum rate-distortion problem and characterized its asymptotic performance limit using a single-letter expression. Wilde \emph{et. al} further refined the characterization of the quantum rate-distortion function in terms of regularized entanglement of formation, and also generalized the problem setup to various scenarios, including side information in \cite{wilde2013auxiliary}.  Works toward the asymptotic simulation of a memoryless quantum
channel in \cite{berta2011quantum,bennett2014quantum}  have shown to be useful in achieving the above results, in particular, the entanglement-assisted formulations. Authors in \cite{datta2013quantum} formulated a quantum-to-classical rate-distortion problem and provided a single-letter formula. 
A rate-distortion version of the quantum state redistribution task \cite{devetak2008exact,luo2009channel} was considered in \cite{khanian2021rate}. Investigations on a rate-distortion framework of generic mixed quantum sources have been the focus of \cite{khanian2022general,baghali2022rate}.  Other works that addressed related problems include \cite{koashi2001compressibility,devetak2002quantum,winter2002compression,datta2013one,hsieh2016channel,salek2018quantum,anshu2019convex}.

 In this work, we  consider a new formulation of the problem  
 of lossy quantum source coding, and  characterize a rate function, no larger than von Neumann entropy,  while allowing for bounded error in the reconstruction. We use a global error criterion as opposed to the  approach of local symbol-wise error studied in the literature.
The problem we consider is without any shared entanglement resources between the encoder and the decoder. We motivate this formulation with the following observations.

The local error criterion in the quantum rate-distortion framework is inspired by the corresponding additive local single-letter distortion criterion in the classical source coding formulation of Shannon \cite{shannon1959coding}, where a  single-letter characterization is available.
The motivation for considering a local criterion is the strong converse of the lossless source coding theorem which states that the entropy bound cannot be breached even when the asymptotic probability of block error is relaxed to any number in $(0,1)$ \cite[Theorem 1.1]{csiszar2011information}.

In \cite{shannon1959coding,csiszar2011information}, a duality connection between the source coding problem and the channel coding problem was observed. These problems were 
interpreted in terms of a covering versus packing perspective. 
In both problems, the same information measure, namely the mutual information, captures the asymptotic performance limits. 
A similar duality connection exists between the classical-quantum communication problem \cite{holevo1998capacity,schumacher1997sending} and the quantum-classical source coding problem \cite{winter, datta2013quantum}, with the performance limits of both problems characterized in terms of single-letter Holevo information quantities \cite{holevo2019quantum}. This has been further explored in \cite{cheng2019duality}.
In the fully quantum setting, 
from this standpoint, its well known that the quantum  channel coding problem 
has 
 an asymptotic performance limit characterized using regularized coherent information \cite{lloyd1997capacity,shor2002quantum,devetak2005private,hayden2008decoupling}. 
Among others, Devetak developed a proof of  this result by employing a coherent approach to covering and packing, and  combined them cohesively, inspired by his work on the private channel capacity problem \cite{devetak2005private}. 
Coherent information can be interpreted in terms of packing of subspaces as elucidated in \cite{lloyd1997capacity}. 
Quantum error-correcting codes have been extensively studied along these lines in the coding theory literature, e.g., quantum Hamming bound \cite{nielsen2002quantum}.
This leads us to the question: why is such a limit based on coherent information absent for the lossy quantum source compression problem?

Toward answering this question, we take a closer look at the classical discrete memoryless setting.  We find that in addition to Shannon's pioneering work of characterizing the rate-distortion problem \cite{shannon1959coding,berger1975rate}, there have been several works discussing the lossy source compression problem. A concept that has received particular attention is the notion of a backward channel \cite[Problem 8.3]{csiszar2011information}, which characterizes the posterior distribution of the source given the reconstruction. The structure of this channel has been studied in \cite{gallager1968information,berger1971rate,gerrish1963estimation}. Although the forward channel, relating the reconstruction to the source, achieving the rate-distortion function need not be unique, the resulting backward channel is indeed unique. 
Moreover, the rate-distortion achievability result in \cite[Theorem 2.3]{csiszar2011information} is shown by constructing a channel code for a backward channel with a large probability of error and by using the encoder of the latter as a decoder of the former and vice versa.
Highlighting this duality further, inspired by results on the output statistics of good channel codes \cite{shamai1997empirical},
the following was shown in \cite{pradhan2004approximation}.
The $n$-letter actual posterior conditional distribution of the source vector given the reconstruction vector of any rate-distortion achieving code converges in normalized divergence to the $n$-product of the unique minimum-mutual-information backward channel conditional distribution. In other words, although the encoder and decoder are block operations, the induced  posterior $n$-letter channel becomes discrete memoryless in the asymptotic limit for a rate-distortion achieving code.
For further developments on this concept see \cite{weissman2005empirical,cuff2010coordination,schieler2013connection,kostina2015output}. This channel also plays a fundamental role in Bayesian estimation and detection theory \cite{poor1998introduction}, e.g., maximum a posteriori (MAP) estimation.
Therefore, we ask the question, can we use such a channel to formulate a lossy source coding problem?




\noindent \textbf{Contributions of this work:}
In light of this, in this work, we explore a new formulation of the source compression problem in the memoryless setting. This formulation is based
on the notion of a posterior channel that produces the reference of the source from that of the reconstruction. 
Instead of a single-letter distortion function, now, we are given a single-letter posterior channel that characterizes the nature  of the loss incurred in the encoding and decoding operations. 
More precisely, we want to construct an encoder and a decoder 
such that the joint effect of encoding and decoding -- to produce a reconstruction sequence 
from the source sequence -- 
is close to the effect of the $n$-product posterior channel acting on the non-product reconstruction sequence.
The closeness is measured using the trace distance in the quantum case and the total variation in the classical case, manifesting as a global error constraint.
A related concept is the 
Petz recovery map
which has found significant relevance in  information-theoretic problems \cite{petz1986sufficient,barnum2002reversing,hayden2004structure}. However, we take a different approach and consider a quantum channel, i.e., a CPTP map, acting on the reference of the reconstruction to produce the reference of the source, whose existence is guaranteed using Uhlmann's theorem.  
We refer to this as a posterior reference map.

As one of the main contributions of our work, we provide a single-letter characterization of the asymptotic performance limit of this source coding problem using the minimal coherent information of the posterior reference map, where the minimization is over all reconstructions (see Theorem \ref{thm:mainResult}).
Furthermore, our work establishes a duality connection between quantum lossy compression and the quantum channel coding problem. Our proof is based on the coherent application of two fundamental tools of quantum information theory, namely, packing and covering, implying a duality relationship with Devetak's proof for the channel coding problem \cite{devetak2005private} (also see \cite{shor2002quantum,hayden2008decoupling}).

We also provide a correspondingly new formulation for the quantum-classical (QC) and classical lossy source coding problems. In the quantum-classical setup, we provide a single-letter characterization of the asymptotic performance limit using 
the minimal Holevo information (or the corresponding quantum mutual information) of the posterior classical-quantum (CQ) channel, where the minimization is over all reconstruction distributions (see Theorem \ref{thm:qclossysourcecoding}). In the classical setup, the minimal mutual information of the posterior channel determines the single letter characterization of the asymptotic performance limit of classical source coding problem (see Theorem \ref{thm:clsrate_distortion}). The posterior CQ channel and the posterior channel are defined analogous to the posterior reference map for the QC and classical settings, respectively.


 At one end of the spectrum, when the posterior reference map is specified as the identity transformation, our rate expression in the quantum case reduces to the von Neumann entropy of the given quantum source, demonstrating the connection with the Schumacher's lossless compression \cite{schumacher1995quantum}. In fact, the two formulations can be shown to be equivalent to one  another. 
 The same follows in the classical and quantum-classical formulations where the rate equals Shannon's entropy and von Neumann's entropy of the sources, respectively. On the other end, when the specified posterior reference map is such that coherent information is negative for some reference of the reconstruction, we characterize the asymptotic performance limit of the lossy quantum source coding problem to be zero. 

 The techniques employed to prove our results can be summarized as follows. For the achievability of the Theorem \ref{thm:mainResult}, we first construct a posterior reference isometry $V$ (as in Definition \ref{def:VBar}) and decompose it as a coherent measurement. We then make use of Winter's measurement compression protocol \cite{winter}, and apply it in a coherent fashion to compress the output of the above isometry. This involves using the Uhlmann's Theorem \cite{uhlmann1976transition} (or \cite[Theorem 9.2.1]{wilde_arxivBook}) followed by incorporating additional phases to achieve a coherent faithful simulation of the posterior reference map. To further decrease the compression rate, we exploit the fact that a noiseless quantum channel can preserve arbitrary superpositions. Therefore, we perform additional encoding to embed the information at the output of $V$ as superpositions within itself. This requires availing the HSW classical communication result \cite{holevo1998capacity,schumacher1997sending} to construct information decoding POVMs, and Naimark's extension theorem to construct a unitary from POVM elements. The method used for expurgation is another interesting feature of the proof. The protocol as it stands only permits operations that are unitary or isometric, followed by partial tracing. It can be challenging to guarantee this when there are repeated codewords in a code. A similar phenomenon was observed in the Devetak's proof \cite{devetak2005private}. 
 
 As for the achievability of Theorem \ref{thm:qclossysourcecoding}, we make use of Winter's measurement compression protocol \cite{winter} to construct the encoding POVM. For Theorem   \ref{thm:clsrate_distortion}, we use the likelihood encoder as discussed in \cite{cuff2013distributed,atif2022source} to prove the achievability of lossy classical source coding.

 For the converse of Theorem \ref{thm:mainResult}, we use the quantum data processing inequality for coherent information,
 the Fannes-Audenart inequality, and monotonicity results.  In the case of the quantum-classical setup,  proof of the converse of Theorem \ref{thm:qclossysourcecoding} uses inequalities 
 such as the quantum data processing inequality, 
 the concavity of conditional quantum entropy, and the continuity of quantum mutual information (AFW inequality). In the classical setup, similar tools are used to prove a converse to Theorem \ref{thm:clsrate_distortion}.

The paper is organized as follows. We provide some necessary definitions and useful lemmas in Section \ref{QLSC:sec:prelim}. In Section \ref{QLSC:sec:mainResults}, we formulate the problems  and provide the main results  pertaining to quantum lossy compression (Theorem \ref{thm:mainResult}), QC lossy compression (Theorem \ref{thm:qclossysourcecoding}), and classical lossy compression (Theorem \ref{thm:clsrate_distortion}). We provide examples corresponding to these three results in section \ref{QLSC:sec:examples}. 
In Sections \ref{sec:q_proof}, \ref{sec:qcproof}, and \ref{sec:clsproof}, we provide proofs of the main results. Within each of these sections, we provide the achievability proof followed by proof of the converse. Finally, Section \ref{sec:conclusion} concludes the paper.


\section{Preliminaries and Notations}
\label{QLSC:sec:prelim}
We supplement the notations in \cite{wilde_arxivBook} with the following. Let $I_A$ denote the identity operator acting on a Hilbert space $\calH_A$. The set of density operators on $\calH_A$ are denoted by $\calD(\calH_A)$, and linear operators by $\calL(\calH_A)$.
We denote $\calH_{\bA}$ as the Hilbert space associated with the reference space of $\calH_A$, with $\dim{\calH_{\bA}} = \dim{\calH_{A}}$.
In this work, we focus exclusively on references obtained from canonical purifications of quantum states \cite[Lemma 14 (Pretty Good Purifications)]{winter}, and define canonical purification $\ket{\psi_\rho}^{\bA A}$ of $\rho^A$ as $\ket{\psi_\rho}^{\bA A} \deq (I_{\bA}\tensor \sqrt{\rho^A})\Gamma_{\bA A}$, where $\Gamma_{\bA A}$ is defined as the unnormalized maximally entangled state. We use $\Psi_\rho^{\bA A}$ to denote the density operator corresponding to $\ket{\psi_\rho}^{\bA A}$. As {is the convention}, for two states acting on the same Hilbert space, we use the same $\Gamma$ when defining their canonical purifications.
We denote the finite alphabet of a source as $\mathsf{X}$, and  
the set of probability distributions on the finite alphabet $\calX$ as $\calP(\calX)$. Let $[\Theta] \deq \curly{1,2,\cdots,\Theta}$.
For a CPTP map $\calN: \calH_A \rightarrow \calH_B$, and an input density operator $\rho^A \in \calD(\calH_A)$, we use $I_c(\calN,\rho^A)$ to denote the coherent information of $\calN$ with respect to $\rho^A$.

\begin{figure}
    \centering
    \includegraphics[trim={0 0 0 0},clip,scale=0.8]{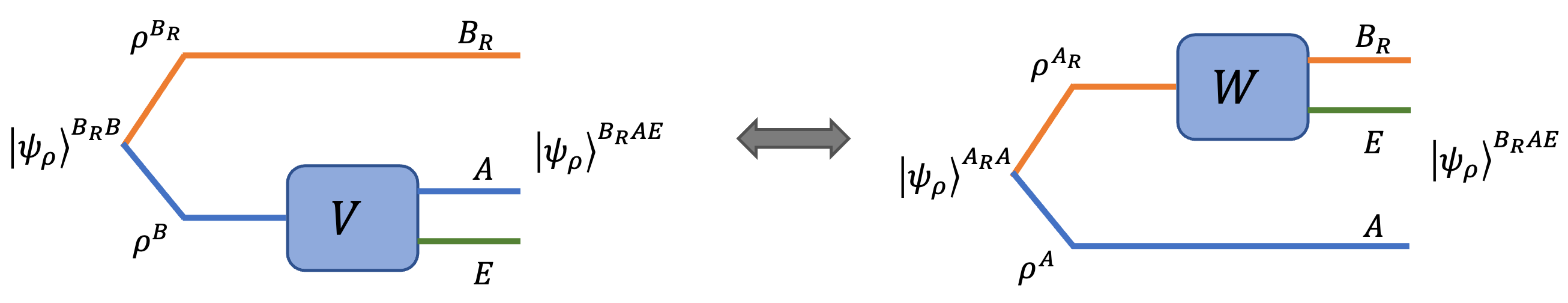}
    \caption{Figure demonstrating the construction of the posterior reference map $W$ from the isometry $V$ (the Stinespring's dilation of $\calN_V$) and the source state $\rho^B$.}
    \label{fig:reverseChannel}
    \vspace{-0.2in}
\end{figure}

\begin{definition}[Posterior Reference Map]\label{def:VBar}
    Given a source $\rho^B \in \calD(\calH_B)$ and a channel $\calN_V: \calH_{B} \rightarrow \calH_A$, let $\rho^A \deq \calN_V(\rho^B)$.  Let $V: \calH_{B} \rightarrow \calH_{A}\otimes \calH_E$ be a Stinespring's isometry corresponding to the CPTP map $\calN_V$ with $\dim(\calH_E) \geq \dim(\calH_A)$, such that $\calN_V(\cdot) = \Tr_{E}\{V(\cdot)V^\dagger\}$. As shown in Figure \ref{fig:reverseChannel}, define the {``posterior reference map''} 
    of $V$ with respect to  $\rho^A$ as the  CPTP map $\calN_{W}:\calH_{\bA} \rightarrow \calH_{\bB}$ corresponding to the isometry $W:\calH_{\bA} \rightarrow \calH_{\bB} \otimes \calH_E$ satisfying $(W \otimes I_{A})\ket{\psi_\rho}^{\bA A} = (I_{\bB}\tensor V)\ket{\psi_\rho}^{\bB B}$ where $\ket{\psi_\rho}^{\bA A}$ and $\ket{\psi_\rho}^{\bB B}$ are the canonical purifications of $\rho^A$ and $\rho^B$, respectively.  
\end{definition}

\begin{remark}[Existence of a Posterior Reference Map]
    Using the equivalence of purifications, one can guarantee the existence of such a posterior reference isometry $W:\calH_{\bA} \rightarrow \calH_{\bB} \otimes \calH_E$. Since $V$ is an isometry
     with $\dim(\calH_E) \geq \dim(\calH_A)$, and since $\ket{\psi_\rho}^{\bA A}$ and $\ket{\psi_\rho}^{\bB AE}\deq (I_{\bB} \tensor V)\ket{\psi_\rho}^{\bB B}$ are purifications of $\rho^A$ (as $\Tr_{E}(V\rho^B V^\dagger) = \rho^A$), from \cite[Theorem 5.1.1]{wilde_arxivBook}, there exists an isometry $W:\calH_{\bA} \rightarrow \calH_{\bB} \otimes \calH_E$ such that $(W\tensor I_A)\ket{\psi_\rho}^{\bA A}  = \ket{\psi_\rho}^{\bB A E}$.
\end{remark}

\subsection{Useful Lemmas}
\label{sec:usefulLemmas}

\begin{lemma}[\cite{fuchs1999cryptographic}, Theorem 9.3.1 \cite{wilde_arxivBook}]
\label{lem:relationshipTraceFidelity}
    Given two states $\rho,\sigma \in \calD(\calH)$, we have
    \[1-\sqrt{F(\rho,\sigma)} \leq \frac{1}{2} \norm{\rho-\sigma}_1 \leq \sqrt{1-F(\rho,\sigma)}.\]
\end{lemma}
\begin{lemma} 
\label{lem:closenessofPurification}
For $\rho^B, \sigma^B \in \calD(\calH_B)$, the following inequality holds:
\begin{equation}F(\ket{\psi_{\rho}}, \ket{\psi_{\sigma}}) \geq \left(1-\frac{1}{2}\norm{{\rho^B}-{\sigma^B}}_1\right)^2,
\label{eq:fidelityPurification1}
\end{equation}
    where $\ket{\psi_{\rho}}$ and $\ket{\psi_{\sigma}}$ are the canonical purifications of $\rho^B$ and $\sigma^B$, respectively. 
\end{lemma}
\begin{proof}
    We provide a proof in Appendix \ref{proof:lem:closenessofPurification}.
\end{proof}
\noindent The above lemma is a slight tightening of the Lemma 14 (“Pretty good purifications”) of \cite{winter}.
\begin{lemma}[Naimark's  extension theorem \cite{neumark1940spectral}, {\cite[Theorem 2.1]{wilde2013sequential}}]
   \label{lem:naimark}
Given a POVM $\{\Gamma_x\}_{x\in \calX}$ acting on the system {$\calH_A$}, there exists a unitary $U_{AA'}$ acting on the system $\calH_A$ and auxiliary system $\calH_{A'}$ and an orthonormal basis $\{\ket{x}^{A'}\}_{x\in \calX}$ such that 
\[\Tr\left\{\overline{\Gamma}_x(\rho^A \otimes \ketbra{0}_{A'})\right\} = \Tr(\Gamma_x \rho^A),\]
where $\{\overline{\Gamma}_x \deq U_{AA'}^{\dagger}(\I_A \otimes \ketbra{x}^{A'})U_{AA'}\}$ are orthogonal projectors acting on system $\calH_A\tensor \calH_{A'}$. Also, $\ket{0}^{A'}$ is some fixed state in $\calH_{A'}$, and independent of $\Gamma_x$ and $\rho^A$.

\end{lemma}



\section{Main Results}
\label{QLSC:sec:mainResults}
\subsection{Lossy Quantum Source Coding}
We first formulate a quantum source coding problem as follows. For any memoryless quantum information source, characterized by $\rho^{B} \in \mathcal{D}(\mathcal{H}_B)$, denote its canonical purification by $\phirho$. Let $\rho^{B_R} \deq \Tr_B [\Psi_\rho^{B_R B}]$.

\begin{definition}[{Quantum Source Coding Setup}] \label{def:qc source coding setup}
A quantum source coding setup is characterized by a triple $(\rho^B,\mathcal{H}_{A},\Wchannel)$, where $\rho^B \in \mathcal{D}(\mathcal{H}_B)$ is a density operator,  $\mathcal{H}_{{A}}$ is a reconstruction Hilbert space, and $\Wchannel$ is a single-letter CPTP map from $\calH_{\bA}$ to $\calH_{\bB}$, where $\calH_{\bA}$ and $\calH_{\bB}$ are reference spaces corresponding to $\calH_A$ and $\calH_B$, respectively.

\end{definition}



\begin{definition}[Lossy Quantum Compression Protocol] \label{def:protocolcompression}
For a given input and reconstruction Hilbert spaces $(\calH_B,\calH_A)$, 
an $(n,\Theta)$ lossy quantum compression protocol consists of a encoding CPTP map $\calN^{(n)}_\calE:\calH_{B^n} \rightarrow \calH_{M}$  and a decoding CPTP map $\calN^{(n)}_\calD:\calH_M \rightarrow \calH_{A^n}$, such that $\dim(\calH_M) = \Theta$, as shown in Figure \ref{fig:q_protocol}. 
\end{definition}

\begin{figure}
    \centering
    \includegraphics[scale=0.9]{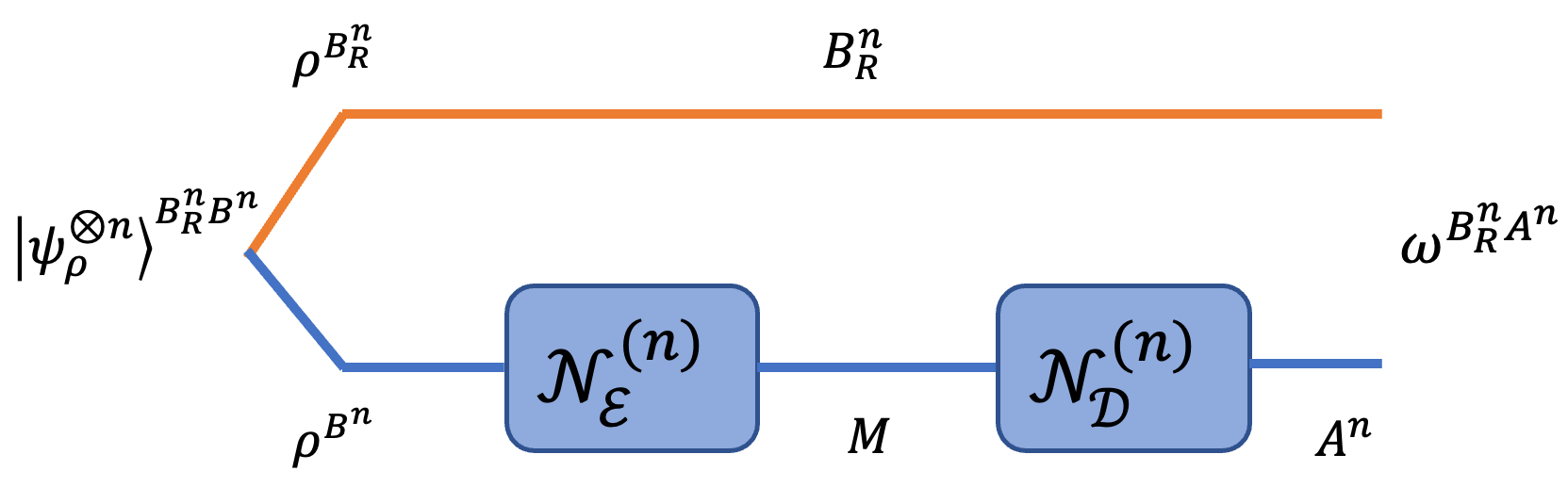}
    \caption{Illustration of Lossy Quantum Compression protocol}
    \label{fig:q_protocol}
    \vspace{-0.2in}
\end{figure}

\begin{definition}[Achievability]
\label{def:achievability}
For a quantum source coding setup $(\rho^B,\mathcal{H}_{A},\Wchannel)$, a rate $R$ is said to be achievable, if for all $ \epsilon > 0 $ and all sufficiently large $n$, there exists an $ (n, \Theta) $ lossy quantum compression protocol satisfying
\begin{align} \label{def:protocolError}
\left\|{\omega^{\bBn A^n} - (  \Wchannel^{\otimes n} \tensor I_{A^n} ) \Psi_{\omega}^{\bAn A^n }}\right\|_1 \leq \epsilon,   
\end{align}
and $ \frac{1}{n}\log{\Theta} \leq R + \epsilon$,
where $\omega^{\bBn A^n} \deq (I \otimes \calN_\calD^{(n)})(I\otimes \calN_\calE^{(n)}) (\Psi_{\rho}^{\bBn B^n })$, and $\Psi_{\rho}^{\bBn B^n}$ and $\Psi_{\omega}^{\bAn A^n } $ are the canonical purifications of ${\rho^B}^{\otimes n}$ and $\omega^{A^n}$, respectively.
    
\end{definition}
In other words, the protocol ensures that the joint state of the reconstruction on $\calH_A^{\otimes n}$ and the original reference  $ \calH_{\bB}^{\otimes n}$ 
is close to the effect of the $n$-product posterior channel acting on the reference of the non-product reconstruction sequence.
Our objective is to characterize the set of all achievable rates using single-letter quantum information quantities. 

\begin{theorem}[Lossy Quantum Compression Theorem]
    \label{thm:mainResult}
    For a $(\rho^B,\mathcal{H}_{A},\Wchannel)$ quantum source coding setup, a rate $R$ is achievable  if and only if $S(\rho^B,\Wchannel)$ is non empty, and
    \[R \geq \min_{\rho^{\bA} \in \calS(\rho^B, \Wchannel)}I^+_c(\Wchannel,\rho^{\bA}),
    \]
    where for any real $x$, $x^+ \deq \max(x,0)$ and
    \[\mathcal{S}(\rho^B,\Wchannel) \deq \{\rho^{\bA} \in \mathcal{D}(\calH_{\bA}): \Wchannel(\rho^{\bA})=\rho^{\bB}\}.\]
\end{theorem}
\begin{proof}
    A proof of the achievability is provided in Sections \ref{sec:proof_overview} and \ref{sec:q_Achievability}, and a  proof of converse is provided in Section \ref{QLSC:sec:proofOuterBound}.
\end{proof}

\begin{remark}[Covering of Subspaces]
The asymptotic rate obtained in the statement of Theorem \ref{thm:mainResult} can be interpreted using a subspace covering argument. Let us assume we are given a source $\rho^B$ and a  CPTP map $\calN_W$ whose coherent information is positive for all $\rho^{\bA} \in \calS(\rho^B,\calN_W).$ Let $W: \calH_{\bA} \rightarrow \calH_{\bB}\tensor \calH_E$ be a Stinespring's dilation of $\calN_W$. This implies $I_c(\calN_W,\rho^{\bA}) = S(\bB)_{\sigma} - S(E)_{\sigma}$, where $\sigma^{\bB E} \deq W \rho^{\bA} W^\dagger,$ for $\rho^{\bA} \in \calS(\rho^B,\calN_W)$. We know that the $n$-product source state ${\rho^{B}}^{\tensor n}$ can be compressed using Schumacher compression to a subspace of normalized logarithmic dimension $S(B_R)_\sigma$ with high probability. In order to further reduce the rate, we use the posterior reference map of $W$ with respect to $\rho^{\bB}$ such that its action on the source produces the state $\rho^A$. 
Each basis vector in the reconstruction space can be thought of as covering a subspace of normalized logarithmic dimension of $S(E)_\sigma$ in the reference space. Therefore, one needs a rate of coherent information (which is the difference of the two entropies)  to cover the entire source space with high probability.
A similar observation was made for the quantum channel coding problem in \cite{lloyd1997capacity}.

\end{remark}

\begin{remark}[Comparison with  Schumacher's lossless compression]
    Schumuacher's compression \cite{schumacher1995quantum} requires $\lim_{n\rightarrow \infty}\|\omega^{\bBn A^n} - \Psi_\rho^{\bBn B^n}\| =0$. In the current formulation, if one chooses the identity map as the posterior reference map, i.e., $\calN_W = I_{\bA \rightarrow \bB}$, we require the condition $ \lim_{n\rightarrow \infty}\|\omega^{\bBn A^n} - \Psi_\omega^{\bAn A^n}\| =0 $. Using Lemma \ref{lem:closenessofPurification}, monotonicity of the trace norm, and the triangle inequality, one can show that the two conditions are equivalent. Subsequently, both formulations yield the same asymptotic performance limit of von Neumann entropy. 
    {Observe that} the standard source coding formulation using the average single-letter distortion criterion at zero distortion level is not equivalent to Schumacher's compression.
\end{remark}

\begin{remark}[Comparison with average single-letter rate distortion]
Given any sequence of $(n,\Theta)$ lossy quantum compression protocol for a quantum source coding setup $(\rho^B,\mathcal{H}_{A},\Wchannel)$ that achieves the optimality in Theorem \ref{thm:mainResult}, we observe that the following is true.
Let  $\omega^{\bBn A^n} \deq (I \otimes \calN_\calD^{(n)})(I\otimes \calN_\calE^{(n)}) (\Psi_{\rho}^{\bBn B^n })$ be the induced state of the $n$-letter reference and the reconstruction by the protocol. Since the protocol satisfies \eqref{def:protocolError}, by monotonicity of trace distance, we obtain 
\[
\lim_{n\rightarrow \infty}\|\omega^{\bB_i A_i} -(\calN_W\tensor I_{A})(\Psi_{\omega}^{\bA_i A_i})\|_1=0, \ \  \forall \; 1 \leq i \leq n,
\]
where $\Psi_{\omega}^{\bA_i A_i}\deq \Tr_{A^{n\backslash i} A^{n\backslash i}_R} [\Psi_{\omega}^{\bAn A^n} ].$ It is worth noting that $\Psi_{\omega}^{\bA_i A_i}$ is not necessarily a pure state.
Moreover, this  does not necessarily provide any guarantee on the average single-letter distortion between the reference and the reconstruction as considered in the  standard formulation of the problem \cite[Lemma 1]{datta2012quantum}, where a single-letter purification of the source is taken into account.
From this perspective, the current formulation is more ``optimistic'' in terms of measuring the quality of the reconstruction. 
\end{remark}

\begin{remark}[Comparison with Entanglement Assistance]
    We note that $$I_c(\calN_W,\rho^{\bA}) = \frac{1}{2}\sq{I(\bB;A)_{\sigma} - I(A;E)_{\sigma}} \leq \frac{1}{2}I(\bB;A)_{\sigma},$$
    where $\sigma^{\bB A E} \deq (I\tensor V) \Psi_{\rho}^{\bB B} (I\tensor V)^\dagger$, and $V:\calH_B \rightarrow \calH_A \tensor \calH_E$ is a posterior reference map of $W$ with respect of $\rho^{\bB}$. It was shown in \cite{datta2012quantum} that $\frac{1}{2}I(\bB;A)_{\sigma}$ characterizes the asymptotic performance limit for the rate-distortion problem (with a local single-letter distortion function) with unlimited entanglement assistance.  Hence, this also provides a lower bound on the asymptotic performance limit for the corresponding problem in the unassisted case. 
    Fortunately, this does not lead to any contradiction, as
    the current formulation  differs from the former by being  more optimistic.

\end{remark}


\subsection{Lossy Quantum-Classical Source Coding}
This section provides the main results
 regarding the quantum-to-classical (QC) setup. 
A memoryless quantum information source is characterized by $\sourcedo \in \calD(\calH_B)$.

\begin{definition}[QC Source Coding Setup] A QC source coding setup is characterized by a triple $(\sourcedo,\calX,\calW)$ where $\sourcedo$ is the source density operator acting on $\calH_B$, $\calX$ is the reconstruction alphabet,
and $\calW: \calX \rightarrow \calD(\calH_B)$ is a single-letter posterior classical-quantum (CQ) channel. 
\end{definition}

\begin{definition}[Lossy QC Compression Protocol] For a given source density operator $\sourcedo$ and the reconstruction alphabet $\calX$,
an $(n,\Theta)$ lossy QC compression protocol is characterized by 
$(i)$ a POVM $\Gamma^{(n)}\deq \{A_m\}_{m=1}^{\Theta}$ and  $(ii)$ a decoding map $f:\curly{1,2,\cdots,\Theta} \rightarrow \calX^n$, as shown in Figure \ref{fig:qc_lossyprotocol}. 

\end{definition}
\vspace{-0.2in}
\begin{figure}[!htb]
    \centering
    \includegraphics[scale=0.8]{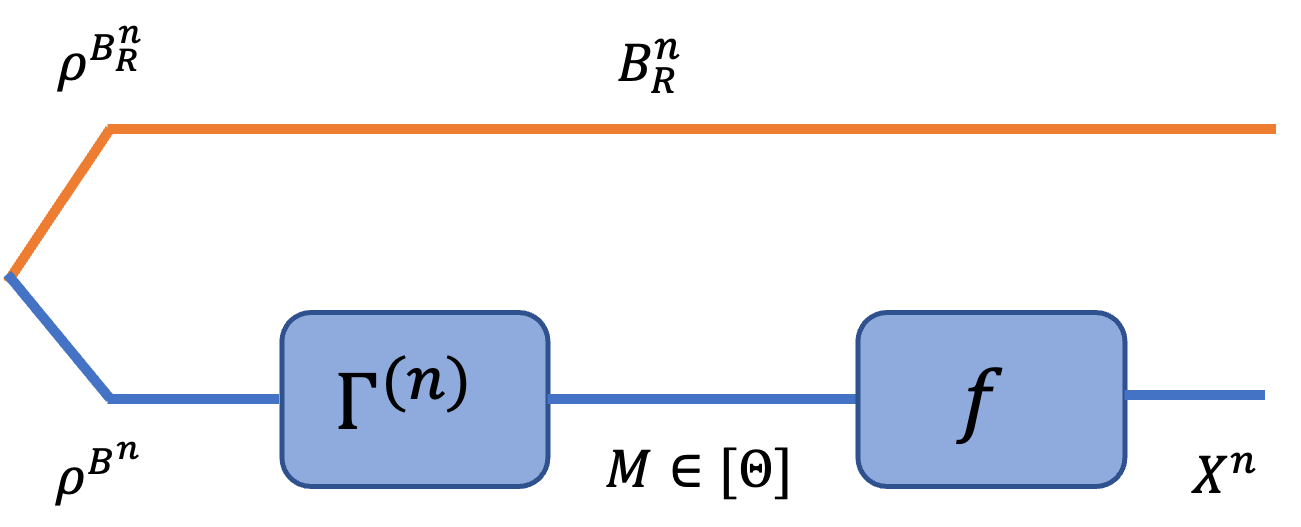}
    \caption{Illustration of Quantum-Classical Lossy Source Compression Protocol.}
    \label{fig:qc_lossyprotocol}
    \vspace{-0.2in}
\end{figure}

\begin{definition}[Achievability]\label{def:qc_achievability}
For a given QC source coding setup $(\sourcedo,\calX,\calW)$, a rate $R$ is said to be achievable if for all $\epsilon > 0$ and all sufficiently large $n$, there exists 
an $(n,\Theta)$ QC lossy compression protocol such that 
$\frac{1}{n}\log \Theta \leq R + \epsilon$, and $\Xi(\Gamma^{(n)},f) \leq \epsilon$, where
\begin{equation}
    \Xi(\Gamma^{(n)},f) \deq \sum_{x^n} \norm{ \sqrt{\sourcedo^{\tensor n}} A_{f^{-1}(x^n)}
\sqrt{\sourcedo^{\tensor n}} -  \Tr(A_{f^{-1}(x^n)} \sourcedo^{\tensor n} ) \bigotimes_{i=1}^{n} \calW_{x_i}}_1.
\end{equation}
\end{definition}

In other words, the post-measurement reference state should look like $n$-tensored posterior CQ channel $\calW^{\tensor n}$. Our objective is to characterize the set of all achievable rates using single-letter quantum information quantities. 

\begin{theorem}[Lossy QC Source Compression Theorem]
\label{thm:qclossysourcecoding}
For a $(\sourcedo,\calX,\calW)$ QC source coding setup, 
a rate $R$ is achievable if and only if $\calA(\sourcedo,\calW)$ is non-empty, and
\[R \geq \min_{P_X \in \calA(\sourcedo,\calW)} I(X;\refstate)_{\sigma},\]
where the quantum mutual information is computed with respect to the classical-quantum state, 
    \[\sigma^{X\refstate} \deq \sum_x P_X(x) \ketbra{x}^{X}\tensor \calW_x 
    ,\]
    $\calA$ is the set of reconstruction distributions defined as
    \[
    \calA(\sourcedo,\calW) \deq \{P_X \in \calP(\calX): \sum_{x}P_X(x) \calW_x =\sourcedo\},
    \]
    and $\{\ket{x}\}_{\curly{x\in \calX}}$ is an orthonormal basis 
    for the Hilbert space $\calH_X$ with $\dim{(\calH_X)}=|\calX|$.
    
\end{theorem}

\begin{proof}
    A proof of the achievability is provided in Section \ref{sec:qc_Achievability}, and a converse proof is provided in Section \ref{sec:qc_converse}.
\end{proof}

\subsection{Lossy Classical Source Coding}


Consider a stationary discrete memoryless source (DMS) $X$ characterized by a source distribution $\px$ over a finite alphabet $\calX$.

\begin{definition}[Source Coding Setup] A source coding setup is characterized by a triple {$(\px,\hat{\calX},\prevTC)$} where $\px$ is the source distribution  over a finite alphabet $\calX$, $\hat{\calX}$ is the reconstruction alphabet, and $\prevTC:\hat{\calX} \rightarrow \calX$ is the posterior (backward) channel, i.e., the single-letter conditional distribution of source given the reconstruction. 
\end{definition}
\noindent We use 
$\px^n$ and $\prevTC^n$ to denote IID distributions, i.e., 
\[\px^n(x^n) = \prod_{i=1}^n\px(x_i) \eqand \prevTC^n(x^n|\xhat^n) = \prod_{i=1}^n\prevTC(x_i|\xhat_i).\]

\begin{definition}[Lossy Source Compression Protocol]
For a given source distribution $\px$ and reconstruction alphabet $\hat{\calX}$,
an $(n,\Theta)$ lossy source compression protocol consists of  $(i)$ a randomized encoding map $\encodern:\calX^n \longrightarrow [\Theta]$ and 
$(ii)$ a randomized decoding map $\decodern:[\Theta] \longrightarrow \hat{\calX}^n$, as shown in Figure \ref{fig:clslossyprotocol}. 
\end{definition}
\begin{figure}[!htb]
    \centering
    \includegraphics[scale=0.9]{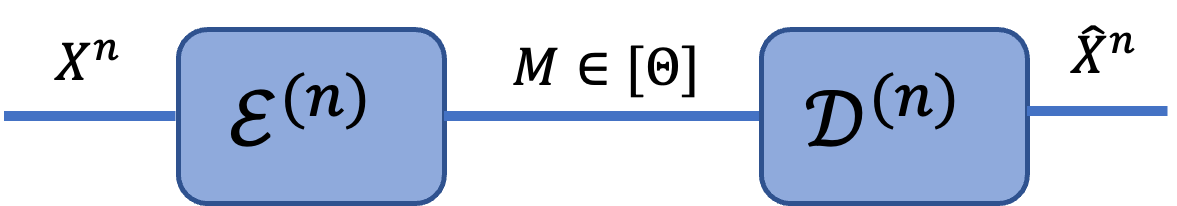}
    \caption{Illustration of Lossy Classical  Source Compression Protocol.}
    \label{fig:clslossyprotocol}
    \vspace{-0.5in}
\end{figure}

\begin{definition} [Achievability] Given a source coding setup $(\px,\hat{\calX},\prevTC)$,
a rate $R$ is said to be achievable if for all $\epsilon >0$ and all sufficiently large $n$, there exists an 
$(n,\Theta)$ lossy source compression protocol such that $\frac{1}{n}\log \Theta \leq R + \epsilon$, and $\Xi(\encodern,\decodern) \leq \epsilon$, where 
\begin{equation}\label{def:error_constraint}
  \Xi(\encodern,\decodern)\deq \frac{1}{2}\sum_{\xn \xhat^n}\absolute{P_{X^n\hat{X}^n}(x^n,\hat{x}^n) -
    P_{\hat{X}^n}(\hat{x}^n)
    \prod_{i=1}^n\prevTC(x_i|\hat{x}_i)},
\end{equation}
and
\[ 
P_{X^n\hat{X}^n}(x^n,\hat{x}^n) =  {\px^n(x^n)} \sum_{m \in [\Theta]} \encodern(m|x^n) \decodern(\hat{x}^n|m), 
\text{  for all  } (x^n, \xhat^n)\in \calX^n \times \hat{\calX}^n,\] is the system-induced distribution,
 and $P_{\hat{X}^n}\prevTC^n$ is the approximating distribution.
\end{definition}
In other words, the posterior distribution of the source given the reconstruction should look like $n$-product posterior channel $\prevTC^n$. Our objective is to characterize the set of all achievable rates using single-letter information quantities. 
\begin{theorem}[Lossy Source Compression Theorem] \label{thm:clsrate_distortion} For a $(\px,\hat{\calX},\prevTC)$ source coding setup, a rate $R$ is said to be achievable if and only if $\calA(\px,\prevTC)$ is non-empty, and \begin{equation}\label{eqn:clsratedistortion}R \geq \min_{\pxhat \in \calA(\px,\prevTC)} I(X;\Xhat), \end{equation}
where $\calA$ is the set of reconstruction distributions defined as \[\calA(\px,\prevTC) \deq \{\pxhat \in \calP(\hat{\calX}):\sum_{\hat{x}} \pxhat(\hat{x}) \prevTC(x|\hat{x}) = \px(x) , \text{ for all } x \in \calX\}.\]
\end{theorem}
\begin{proof}
A proof of the achievability is provided in Section \ref{subsec:clsachievability}, and a converse proof is provided in \ref{subsec:clsconverse}.
\end{proof}
\begin{remark}[Comparison with Shannon's noiseless source compression]
    Noiseless source compression requires $\lim_{n\rightarrow \infty}P(\Xn\neq \Xhat^n) =0$. In the current formulation,  if one chooses the identity posterior channel, i.e., $\prevTC(x|\xhat) = \11_{\{x = \xhat \}}$, for all $x \in \calX$, $\xhat \in \hat{\calX}$, we require $ \lim_{n\rightarrow \infty}\norm{P_{\Xn\Xhat^n} - P_{\Xhat^n}\prevTC^n}_{\normalfont \text{TV}}=0$. One can easily see that the two conditions are equivalent, and both formulations yield the same asymptotic performance limit of Shannon's entropy. 
    However, the standard source coding formulation using the average single-letter distortion criterion at zero distortion level is not equivalent to noiseless source compression.
\end{remark}

\section{Illustrative Examples}
\label{QLSC:sec:examples}

\begin{example}[Quantum Source Coding using Bit-Flip Channel]\label{example2}
    In this example, we analyze the performance of a lossy quantum compression protocol corresponding to a quantum source coding setup $(\rho^B, \calH_A, \calN_W),$ where $\rho^B$ is chosen as the maximally mixed state ($\rho^B = I_B/2$), and $\calN_W{:\calH_{\bA}\rightarrow\calH_{\bB}}$ is specified as a bit-flip channel.
    An isometry $W{:\calH_{\bA}\rightarrow\calH_{\bB}\tensor \calH_E}$ for $\calN_W$ can be specified as $$W = \sqrt{1-p}I\tensor \ket{0}^{E}+\sqrt{p} X \tensor \ket{1}^E,$$ where $\calN_W(\rho^{\bA}) = \Tr_E(W\rho^{\bA} W^\dagger)$ for all $p \in (0,1/2)$. Note that the canonical purification $\ket{\psi_\rho}^{\bB B}$  of $\rho^B$ is given by
    \begin{align}
        \ket{\psi_\rho}^{\bB B} = \frac{1}{\sqrt{2}}\round{\ket{0}^{\bB}\ket{0}^{B} + \ket{1}^{\bB}\ket{1}^{B}},
    \end{align}
    where $\ket{0}^{\bB} \deq (I\tensor \bra{0}^B)\ket{\Gamma}^{\bB B}$. This implies, $\rho^{\bB} = I_{\bB}/2$.
    To compute the asymptotic performance of the protocol for this source coding setup, as characterized by Theorem \ref{thm:mainResult}, we first need to identify a $\rho^{\bA}$ such that $\calN_W(\rho^{\bA}) = \rho^{\bB}$. 
    A simple computation reveals 
    $\calS(\rho^B,\calN_W) = \{{I_{\bA}/2}\}$. 
     This gives
    \[\min_{\rho^{\bA} \in \calS(\rho^B, \Wchannel)}{I^+_c(\calN_W,\rho^{\bA})} = I_c(\calN_W,{I_{\bA}}/{2} )  = S(\bB)_\sigma - S(E)_\sigma,\]
    where
    $\sigma^{\bB E} = W\rho^{\bA}W^\dagger$. Note that $\sigma^{\bB} = I_{\bB}/2$ and $\sigma^{E} = (1-p)\ketbra{0}^{E} + p\ketbra{1}^E,$  which gives $I_c(\calN_W,{I_{\bA}}/{2} ) = 1 -h_b(p)$, where  $h_b(p)\deq -p\log(p) - (1-p)\log(1-p).$ Therefore, a maximally mixed source can be compressed at a rate $1-h_b(p)$ while satisfying the error criterion as defined in \eqref{def:protocolError}. 
\end{example} 

\begin{example}[Quantum Source Coding using Depolarizing Channel]
    In this example, we study the performance of another candidate channel, namely a depolarising channel. We again proceed with the objective of compressing a maximally mixed state $\rho^{\bB} = \frac{I_{\bB}}{2}$, with $\calN_W$ defined as
    \begin{align*}
        \calN_W(\rho^{\bA}) = \left(1-\frac{3p}{4}\right)\rho^{\bA} + \frac{p}{4} (X\rho^{\bA}X^\dagger + Y\rho^{\bA}Y^\dagger + Z\rho^{\bA}Z^\dagger).
    \end{align*}
    for some $p \in [0,1].$
    A simple calculation to satisfy $ \calN_W(\rho^{\bA}) = \rho^{\bB} = \frac{I_{\bB}}{2}$ reveals  $\calS(\rho^B,\calN_W) = \{{I_{\bA}/2}\}$, for all $p \in (0,1)$. Analogous to the above example, finding an isometric extension of $\calN_W$ gives
    \begin{align*}
       \min_{\rho^{\bA} \in \calS(\rho^B, \Wchannel)}I^+_c(\calN_W,\rho^{\bA}) = I^+_c(\calN_W,{I_{\bA}}/{2} )  = \max\big\{0,1- h_b({3p}/{4}) - \frac{3p}{4}\log(3)\big\}.
    \end{align*}
\end{example}


\begin{example}[Hamming codes for quantum source compression] \label{example1}
In this example, we look at how  Hamming codes perform when evaluated using the standard single-letter (local) entanglement fidelity criterion. 
Hamming codes are perfect codes, and achieve the Delsarte upper bound on the covering radius \cite{mattson2012upper}. 
Again, let $\rho^{B} = \frac{I_B}{2}$. 
Let a maximally entangled bipartite state $\ket{\psi_m}^{\bB B}$,  defined as
\begin{align}
    \ket{\psi_m}^{\bB B} = \frac{1}{\sqrt{2}}\left(\ket{00}^{\bB B} +\ket{11}^{\bB B} \right),
\end{align}
be the purification of $\rho^{B}$.
Let $\FF_2$ denote a binary finite field, and let $G \in \FF_2^{k\times n}$ be the generator matrix of a Hamming code. To encode $\rho^B$, we appeal to the duality perspective, and use the decoder of a Hamming code. Then the encoding is defined as ${\calE(x^n)} \deq \mbox{argmin}_{u^k} \{w_{H}({u^k \; G \oplus x^n})\}, $ for all $x^n \in \FF_2^n, $ where $w_H$ denotes the Hamming weight. Similarly, the decoder can be described as mapping  $\calD\circ \calE((x^n))= \calE(x^n) G$.
 We describe this encoding as an isometric action $V_{H}:\calH_{B}^{\tensor n}\rightarrow \calH_{A}^{\tensor n}\tensor \calH_{E}^{\tensor n}$ taking the basis $\ket{x^n}^{B^n}$ to a vector $\ket{\calE(x^n)}^{A^n}\tensor\ket{x^n\oplus \calE(x^n)}^{E^n} \in \calH_{A}^{\tensor n}\tensor \calH_{E}^{\tensor n}$, where the subsystem $\calH_{A}^{\tensor n}$ stores the reconstruction and $\calH_{E}^{\tensor n} $ is eventually traced out, and $\calH_A$ is assumed to be an isomorphic copy of $\calH_B$. This implies that the encoded state can be characterized as
\begin{align}
    \rho^{\bBn A^n} &= \Tr_{E^n} \curly{V_H \ketbra{\psi_m^{\tensor n}}^{\bBn B^n}V_H^\dagger}. \nonumber 
\end{align}
Using
\begin{align*}
    V_H \ket{\psi_m^{\tensor n}}^{\bBn B^n} &= \frac{1}{\sqrt{2^n}}\sum_{x^n}\ket{x^n}_{\bBn} \ket{\calE(x^n)}_{A^n}\ket{x^n\oplus \calE(x^n)}_{E^n} \nonumber \\
    & = \frac{1}{\sqrt{2^n}} \sum_{c^n\in \calC}\;\sum_{e^n \in \FF_2^n: w_H(e^n) \leq 1 }\ket{c^n\oplus e^n}_{\bBn}\ket{c^n}_{A^n}\ket{e^n}_{E^n},
\end{align*} we can simplify $\rho^{\bBn A^n}$ as 
\begin{align}
    \rho^{\bBn A^n} = \frac{1}{2^n}\sum_{c^n, c'^n, e^n}\ket{c^n\oplus e^n}\langle{c'^n\oplus e^n}|\tensor \ket{c^n}\langle{c'^n}|,
\end{align}
where $\calC$ denotes the set of codewords of the Hamming code.
To compute the single-letter entanglement fidelity, we compute
\begin{align}
    \rho^{B_{R_i}A_i} = \Tr_{B_R^{n\backslash i} A^{n\backslash i}}\curly{\rho^{\bBn A^n}}
    =\frac{1}{2^n}\sum_{c^n, e^n}\ket{c_i\oplus e_i}\bra{c_i\oplus e_i}\tensor \ket{c_i}\bra{c_i} ,
\end{align}
where tracing is performed on all the subsystems except corresponding to ${B_R^{n\backslash i} A^{n\backslash i}}$, and the second equality follows from using the fact that 
minimum Hamming distance of any Hamming code is three.
This gives,
\begin{align}
    \langle{\psi}_m^{\bB B}|\rho^{B_{R_i}A_i}|{\psi}_m^{\bB B}\rangle = 
    \frac{1}{2}\frac{1}{2^n}\sum_{c^n e^n} \sq{\11_{\curly{c_i\oplus e_i = 0, c_i = 0}}+ \11_{\curly{c_i\oplus e_i = 1, c_i = 1}}} = \frac{1}{2^{n+1}}\sum_{c^n e^n}\11_{\curly{e_i = 0}}.
\end{align}
Therefore,
\begin{align}
    \frac{1}{n}\sum_{i=1}^n \langle{\psi_m}^{\bB B}|\rho^{B_{R_i}A_i}|{\psi_m}^{\bB B}\rangle &= \frac{1}{2^{n+1}n}\sum_{c^n e^n}\sum_{i=1}^n\11_{\curly{e_i = 0}}  = \frac{|\mathcal{C}|n^2}{2^{n+1}n} = \frac{n}{2\cdot 2^{n-k}}.
\end{align}
We know that for Hamming codes $k = 2^{r}-r-1$ and $n = 2^r - 1$, which simplifies as
\begin{align}
    \frac{1}{n}\sum_{i=1}^n \langle{\psi_m}^{\bB B}|\rho^{B_{R_i}A_i}|{\psi_m}^{\bB B}\rangle & = \frac{2^r-1}{2\cdot 2^r}, 
\end{align}
and goes to half as $r$ goes to infinity.
Note that $r \rightarrow \infty$ serves as both a demonstration of the code's asymptotic performance and the condition for the rate $k/n$ to reach unity. 
This results in a discontinuous asymptotic performance, since at rate exactly one, trivial identity encoding can be used to achieve the average single-letter fidelity of unity.
Further, note that $S(E^n)=\log (n+1)=r$. Hence the normalized amount of qubits that is dissipated, given by $\frac{S(E^n)}{n}$, approaches zero as $r \rightarrow \infty$, indicating that there is significant entanglement between the reconstruction and the reference. 

As was demonstrated in Example \ref{example2}, it is possible to compress a maximally mixed source in a continuous fashion, when the error is measured in accordance with the suggested definition in \eqref{def:protocolError},

\end{example}


\begin{example}(Lossy QC Source Coding for Binary Quantum Source with Binary Symmetric Posterior CQ Channel)
We develop an example similar to that studied in \cite{datta2013quantum}. Here we analyze the performance of the lossy QC source compression protocol corresponding to a lossy source coding setup $(\sourcedo,{\calX},\calW)$. The quantum source $\sourcedo$ generates the state $\ket{+}$ and $\ket{0}$ with probability $p$ and $(1-p)$, respectively, where $p \in [0,1/2]$, so the source density operator can be written as 
\[\sourcedo = p \ketbra{+} + (1-p) \ketbra{0},\]
 the reconstruction set ${\calX} = \curly{0,1}$, and the posterior CQ channel $\calW_x = (1-q)\ \omega_{x} + q \ \omega_{\Bar{x}}$, where  $$\omega_0 = \frac{1}{4}\ketbra{+} + \frac{3}{4}\ketbra{0} , \quad \omega_1 = \frac{3}{4}\ketbra{+} + \frac{1}{4}\ketbra{0},$$  $q\in [0,1/2]$, and $\Bar{x}\deq x \oplus 1$.
 Toward identifying the set $\calA$, we assume $\px(0) = r$, which characterizes the set $\calA$, and solve the following
 \begin{equation}\label{eqn:fig:qccond}\sourcedo = r\calW_0 + (1-r)\calW_1, \quad 0 \leq r \leq 1.\end{equation}
This gives, 
\[
\calA(\sourcedo,\calW) =
   \begin{cases}
        \curly{\frac{1}{2} + \frac{1-2p}{1-2q}} & \mbox{ if } 0 \leq q \leq 2\min\curly{\round{\frac{3}{4}-p},\round{p - \frac{1}{4}}}, q < \frac{1}{2}\\
        {[0,1]} & \mbox{ if } q=p=0.5 \\
        \phi & \mbox{ otherwise,}
   \end{cases} 
\]
where $\phi$ denotes the empty set.
We now compute the asymptotic performance described in Theorem \ref{thm:qclossysourcecoding}. For the above source coding setup, we have
\begin{align}\label{eqn:qcexp1}
    I(X;\refstate)_{\sigma} = S(\sourcedo) - r S(\calW_0) - (1-r) S(\calW_1).
\end{align}
where \[\sigma^{X\refstate} \deq r\ketbra{0} \tensor \calW_0 + (1-r)\ketbra{1} \tensor \calW_1.\] 
 Figure \ref{fig:QC_CLSeg} shows the QC lossy source compression rate curve for the range of values of the parameter $q$ and source $\sourcedo$ with $p = 0.4$ and $0.5$. Note that the curve decreases monotonically with $q$, as expected. 



\end{example}

\begin{figure}
    \centering
    \includegraphics[scale=0.22]{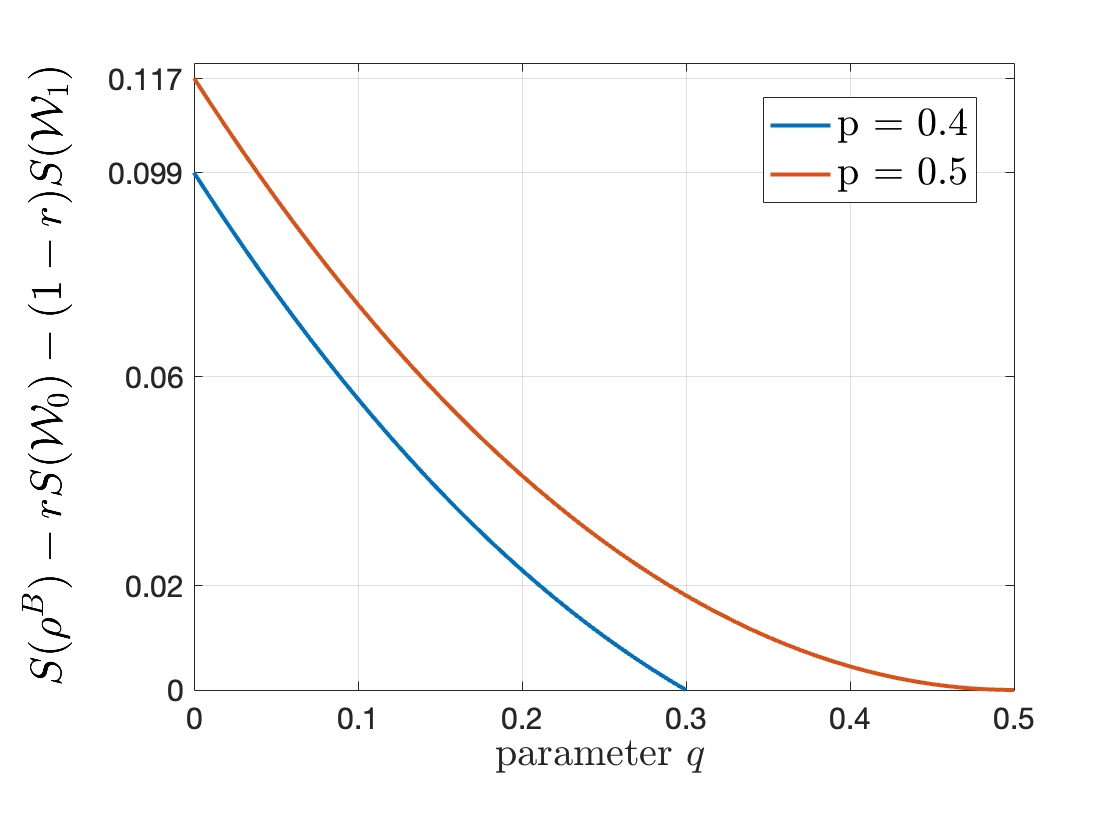}
    \vspace{-0.4in}
    \caption{Example for Lossy QC Source Coding for Binary Quantum Source with Binary Symmetric Posterior CQ Channel.}
    \vspace{-0.2in}
    \label{fig:QC_CLSeg}
\end{figure}
\begin{example}(Lossy Classical Source Coding for Binary Source with Binary Symmetric Channel (BSC) as Posterior Channel)
In this example, we analyze the performance of the lossy source compression protocol corresponding to a lossy source coding setup $(\px,\hat{\calX},\prevTC)$, where $\px \sim \text{\normalfont Bernoulli}(p)$, $\hat{\calX} = \curly{0,1}$, $\prevTC \sim \text{\normalfont BSC}(q)$, and  $p,q \in [0,1/2]$.
 Toward identifying the set $\calA$, we assume $\pxhat(0) = r$, which characterizes the set $\calA$, and solve the following system of linear equations:
\begin{equation}\label{eqn:fig:clsegcond}
p = r(1-q) + (1-r)q\quad \eqand \quad (1-p) = rq + (1-r)(1-q).\end{equation}
This gives, 
\[
\calA(\px,\prevTC) = \begin{cases}
    \curly{\frac{p-q}{1-2q}} &  \mbox{ if } 0 \leq q \leq \min\{p,(1-p)\}, q < \frac{1}{2}\\
   [0,1] & \mbox{ if } q = p = 0.5\\
    \phi & \mbox{ otherwise,}
\end{cases}
\]
where $\phi$ denotes the empty set. We now compute the asymptotic performance described in Theorem \ref{thm:clsrate_distortion}. For the above source coding setup, we have
\begin{align}\label{eqn:clsexp1}
    I(X;\Xhat) = H(X) - H(X|\Xhat) = h_b(p) - h_b(q).
\end{align}

\noindent Moreover, observe that the rate in \eqref{eqn:clsexp1} is identical to the rate-distortion function for a Bernoulli($p$) source
with Hamming distortion criterion  for $D < p$ \cite[Theorem 10.3.1]{cover2006elements}.  
\end{example}


\section{Proof of Theorem \ref{thm:mainResult}}
\label{sec:q_proof}


\subsection{Achievability Proof Overview}
\label{sec:proof_overview}
We provide a brief overview of the achievability proof before formally presenting one. The proof we present here is inspired by Devetak's work in \cite{devetak2005private} for the quantum channel communication problem (also detailed in \cite[Chapter 24]{wilde_arxivBook}). An integral component of that work is the decomposition of coherent information as the difference of two Holevo information quantities. 
We intend to perform a similar decomposition, but from the perspective of the given map $\calN_W$. Toward this, for the given source $\rho^B$, we first search for a $\rho^{\bA} \in \calD(\calH_{\bA})$, satisfying $\calN_W(\rho^{\bA}) = \rho^{\bB}.$ 
Once found, using the spectral decomposition, we expand $\rho^{\bA}$ as $\rho^{\bA} = \sum_{a\in \calA}\lambdaaA \ketbra{a}^{\bA}$, for some finite set $\calA$. Observe that since $\ketbra{a}^{\bA}$ is pure, $S(\calN_W(\ketbra{a}^{\bA})) = S(\calN^c_W(\ketbra{a}^{\bA}))$, where $\calN^c_W:\calH_{\bA} \rightarrow \calH_{E}$ is a complementary CPTP map of $\calN_W$, defined using the Stinespring's dilation $W:\calH_{\bA}\rightarrow\calH_{\bB}\tensor \calH_E$ corresponding to $\calN_W$. This also means that
\[\sum_{a\in\calA}\lambdaaA S(\calN_W(\ketbra{a}^{\bA})) = \sum_{a\in\calA}\lambdaaA S(\calN^c_W(\ketbra{a}^{\bA})).\]
Furthermore, from the linearity of CPTP maps, we see 
\[ \sum_{a\in\calA}\lambdaaA\calN_W(\ketbra{a}^{\bA}) = \calN_W(\rho^{\bA}) \qand \sum_{a\in\calA}\lambdaaA \calN^c_W(\ketbra{a}^{\bA}) = \calN^c_W(\rho^{\bA}).\]
This implies, we can rewrite $I_c(\calN_W,\rho^{\bA})$ as
\begin{align}
    I_c&(\calN,\rho^{\bA})  = S(\calN_W(\rho^{\bA})) - S(\calN^c_W(\rho^{\bA})) \nonumber \\
    & = \sq{S(\calN_W(\rho^{\bA})) - \sum_{a\in\calA}S(\lambdaaA\calN_W(\ketbra{a}^{\bA}))} - 
 \sq{S(\calN^c_W(\rho^{\bA})) - \sum_{a\in\calA}\lambdaaA S(\calN^c_W(\ketbra{a}^{\bA}))} \nonumber \\
 & = \chi\round{\curly{\lambdaaA,\calN_W(\ketbra{a}^{\bA})}} -  \chi\round{\curly{\lambdaaA,\calN^c_W(\ketbra{a}^{\bA})}}. \label{eq:diff_coherent}
\end{align}

Now our aim is to show the achievability of a rate equal to the above difference. After obtaining a similar decomposition,
Devetak achieved the performance limit by applying a coherent version of the CQ packing lemma \cite[Chapter 16]{wilde_arxivBook} followed by an application of the QC covering lemma \cite[Chapter 17]{wilde_arxivBook}. 
 Inspired by this, and the duality connections between the two problems, we achieve the difference obtained in \eqref{eq:diff_coherent}.
 In particular, we start with the objective of applying a coherent version of the QC covering lemma (or the measurement compression result \cite{winter}). Toward this, as shown in Figure \ref{fig:reverseChannel}, we first obtain a posterior reference map $V$ corresponding to the isometry $W$. Then we identify the action of $V$ on the state $\rho^B$ as a coherent quantum measurement. Now, using the approximating POVMs constructed in \cite{winter}, we  perform a coherent covering that allows us to compress the obtained measurement, and in turn the output of $V$, at rate given by the first Holevo information. The compression is performed while \emph{faithfully simulating} the action of $V$, giving a reconstruction satisfying the error criterion (as in \eqref{def:protocolError}). This procedure is delineated  in Step 1.1 where an encoder is constructed to perform coherent covering and in Step 2.1 where the effect of covering is analyzed, and a rate corresponding to the first Holevo information is achieved.

To get the needed coherent information, the rate corresponding to the second Holevo information must be further decreased. This entails diffusing more data or qubits into the environment (partial tracing). However, as will be demonstrated in the proof below,
such an action would destroy quantum correlations present in the source, possibly turning it into a classical mixture.
 Therefore, before such partial tracing operation, in Step 1.2 (Section \ref{sec:q_Achievability}) we construct a unitary operation that can condense the information into fewer qubits in the form of entanglement, and thus allowing for further decrease in the rate. This includes using the coherent post-measurement state of the subsystem $E$ as side information available at the encoder.
 The Step 2.2 of Section \ref{sec:q_Achievability} details this procedure and achieves the desired rate. Finally, an additional step (Step 2.3) is required to show the intended closeness as required in \eqref{def:protocolError}.

 Another intriguing aspect of the proof is the technique used for expurgation. As clear from the definition of the protocol, it only allows unitary or isometric operations, followed by partial tracing. When a code contains repeated codewords, it can be difficult to guarantee this.
An approach to removing all repetitions is to perform expurgations. This is achieved by finding a good code (satisfying all its constraints) while allowing a small fraction of repeats
and then expurgating just this fraction of the code. However, if there are exponentially many constraints, it becomes challenging to finding a good code.
The exponentially many covering constraints in Devetak's problem have a doubly exponential decreasing probability of error, which Devetak was able to take advantage of.
In the current problem we instead have exponentially many packing constraints which only have an exponential decay. In order to combat this,
we construct our proof to just require one packing constraint: the average of all exponentially many packing constraints.
This enables us to find a good code and successfully expurgate it. We now formally construct the arguments toward proving the statement of the theorem.




\subsection{Proof of Achievability}
\label{sec:q_Achievability}

The proof is mainly composed of four parts. In the first part, we develop the necessary single-letter ensembles required in the proof. In the next part, we provide the random coding setup and the distributions on the ensembles with which the codewords are generated. We also state here the constraints that a good code must satisfy and argue the existence of one code with non-zero probability. We further use an expurgation strategy to make all the codewords distinct. In the third part, we construct a protocol by developing all the actions of the encoder and the decoder and describing them as unitary  (or isometry) evolutions. Note that the only actions allowed by the protocol (Definition \ref{def:protocolcompression}) are quantum channels which can be described as unitary or isometric evolutions followed by partial trace operations. 
In parallel, we also provide the necessary lemmas needed for the next part. The last part deals with analyzing the action of encoding and decoding operations
on the source $\rho^B$, and then bounding the trace distance as in Definition \ref{def:protocolcompression}.

Toward this, fix two positive integers $M$ and $K$, and $\epsilon \in (0,1)$. Let 
$\calM$ and $\calK$ denote the sets $[0,M-1]$ and $[0,K-1]$, respectively.  Given a quantum source coding setup $(\rho^B, \calH_A, \calN_W)$, let  $\ket{\psi_\rho}^{\bB B}$ be the canonical purification of $\rho^B$ and $\rho^{\bB} \deq \Tr_{B}\{{\Psi_\rho}^{\bB B}\}$. Moreover, let $\calH_{\bA}$ be the reference space associated with $\calH_A$. Now choose 
  $\rho^{\bA} \in S(\rho^B,\calN_W)$. Let $\calH_E$ denote the Hilbert space such that $W:\calH_{\bA} \rightarrow \calH_{\bB} \tensor \calH_E$ forms an isometric extension (or Stinespring's dilation) of $\Wchannel$ according to \cite[Definition 5.2.1]{wilde_arxivBook} with $\dim(\calH_E) \geq \dim(\calH_{\bB})$. As shown in Figure \ref{fig:reverseChannel}, define a {posterior reference isometry} of $W$ with respect to $\rho^{\bB}$ (according to Definition \ref{def:VBar}) as the isometry  $V:\calH_{B} \rightarrow \calH_{A} \otimes \calH_E$ satisfying $(W \otimes I_{A})\ket{\psi_\rho}^{\bA A} = (I_{\bB} \tensor V )\ket{\psi_\rho}^{\bB B}$ where $\ket{\psi}^{\bA A}$ is the canonical purification of $\rho^{\bA}$. Let $\rho^A \deq \Tr_{\bB E}\{(I\tensor V){\Psi_\rho}^{\bB B}(I\tensor V)^\dagger\}$. 

{\subsubsection{Defining the ensembles}
In this section, we construct the single-letter ensembles corresponding to two Holevo information quantities used in the decomposition of coherent information discussed in Section \ref{sec:proof_overview}. }
We begin by using the definition of $\bV$ to obtain,
\begin{align} \label{eq:relationVtoVbar}
    (I_{\bB}\tensor V)\ket{\psi_\rho}^{\bB B} = (\bV \tensor I_{A}) \ket{\psi}^{\bA A} = \sum_{a\in\calA} \sqrt{\lambda_a^A}\bV\ket{a}^{\bA}\tensor \ket{a}^A,
\end{align}
where we use $\rho^A = \sum_{a\in \calA}\lambda_a^A\ketbra{a}^A$ as its spectral decomposition, and define $\ket{a}^{\bA} \deq (I_{\bA}\tensor \bra{a}^A)\ket\Gamma^{\bA A}$ for $a\in \calA$, for some finite set $\calA$. This also gives,
\begin{align}
    \bV\ket{a}^{\bA} = \frac{(\bra{a}^A\tensor I_{\bB E})(I_{\bB}\tensor V )\ket{\psi_\rho}^{\bB B}}{\sqrt{\lambda_a^A}} \label{eq:actionVbarOnA_0}.
\end{align}
Using the spectral decomposition of $\rho^B$ as $\rho^B = \sum_{b\in \calB}\lambda_b^{B}\ketbra{b}^B$, for $b\in\calB$ for some finite set $\calB$, we can rewrite the action of $V$ on $\rho^B$ as
\begin{align}
    (I_{\bB}\tensor V )\ket{\psi_\rho}^{\bB B} & = \sum_{b\in \calB}\sqrt{\lambda_b^B} \ket{b}^{\bB} \tensor V\ket{b}^B \nonumber \\
    & = \left(I_{\bB}\tensor I_{E} \tensor \sum_{a\in \calA}\ketbra{a}^A\right) \sum_{b\in \calB}\sqrt{\lambda_b^B} \ket{b}^{\bB} \tensor V\ket{b}^B
    \nonumber \\ 
    & = \sum_{a\in\calA}\sum_{b\in\calB} \sqrt{\lambda_b^B} \ket{b}^{\bB} M_a\ket{b}^B \tensor \ket{a}^A ,
\end{align}
where we define $\ket{b}^{\bB} \deq (I_{\bB}\tensor \bra{b}^B)\ket\Gamma^{\bB B}$, and $M_a: \calH_B \rightarrow \calH_E$ as 
\begin{align}\label{eq:def_Ma}
  M_a \deq  \big(I_E \tensor \bra{a}^A\big) V.
\end{align}
By defining a  POVM $\Lambda \deq \{M_a^\dagger M_a\}_{a\in\calA}$, we can identify a coherent measurement (isometry) $U_\Lambda$ corresponding to $\Lambda$ with $U_\Lambda \deq \sum_{a \in \calA}M_a\tensor \ket{a}^A$, and therefore express the action of $V$ as 
\begin{align}
    (I_{\bB}\tensor V )\ket{\psi_\rho}^{\bB B} =  (I_{\bB}\tensor U_\Lambda)\ket{\psi_\rho}^{\bB B}.
\end{align}
Now our objective is to \emph{faithfully simulate} the action of the isometry (or the coherent measurement) $U_\Lambda$ while using an exponentially smaller subspace in $\calH_{A^n}$. Equivalently, we intend to minimize the amount of qubits needed to represent the quantum state in the Hilbert space $\calH_{A^n}$. Employing Schumacher's compression \cite{schumacher1995quantum}, one can only achieve a rate of Von-Neumann entropy while faithfully simulating $U_\Lambda$. However, 
since $U_\Lambda$ is a coherent measurement, we employ a coherent version of the measurement compression protocol \cite{winter} and demonstrate a faithful simulation of the isometry while further decreasing the resource requirement. 
In particular, an approximating coherent measurement (henceforth referred to as the covering isometry) $U_\calM$ is constructed to faithfully simulate the action of $U_\Lambda$ while requiring the rate equal to Holevo quantity corresponding to the canonical ensemble  $\{\lambdaaA, \rhohata\}$, where 
\begin{align}
    \rhohata \deq \frac{\sqrt{\rho^{\bB}}(M_a^\dagger M_a)^{T}\sqrt{\rho^{\bB}}}{\lambdaaA} \qand (M_a^\dagger M_a)^{T} \deq \sum_{b,b'} \ket{b}\langle{b'}|^{\bB} \langle{b'}|(M_a^\dagger M_a)\ket{b}^{B}. 
\end{align}


\noindent Observe that using the definition of $M_a$ from  \eqref{eq:def_Ma}, it follows
\begin{align}
    \label{eq:Ma_orthogonal}
    \Tr{M^\dagger_{a'}M_a\rho^B} &= 
\Tr{(I_{\bB E} \tensor \bra{a})V\ketbra{\psi_\rho}^{\bB B}V^\dagger (|{a'}\rangle\tensor I_{\bB E})} \nonumber \\
& = \sum_{b}\sqrt{\lambda_a^A \lambda_{a'}^A}\Tr{\bra{b}\bV |{a'}\rangle\bra{a} \bV^\dagger \ket{b}} \nonumber \\
& = \sum_{b}\sqrt{\lambda_a^A \lambda_{a'}^A} \bra{a}\bV^\dagger \ketbra{b} \bV |{a'}\rangle = \lambdaaA\cdot \11_{\{a=a'\}}, 
\end{align}
for all $a,a' \in \calA$, where the first equality uses the definition of $M_a$, and the second follows from using the relation \eqref{eq:relationVtoVbar}.
{Using the simplification from \eqref{eq:actionVbarOnA_0}, it is useful to note 
\begin{align}
    \bV\ket{a}^{\bA} = \frac{(I_{\bB}\tensor M_a) \ket{\psi_\rho}^{\bB B}}{\sqrt{\lambda_a^A}} \label{eq:actionVbarOnA}.
\end{align}}
For the second Holevo information, we define the packing ensemble $\{\lambda_a^{E}, \tau_a^E\}$ as
\begin{align}
    \tau_a^E \deq \frac{\Tr_{\bB}{(I_{\bB}\tensor M_a)\Psi_\rho^{\bB B}}(I_{\bB}\tensor M_a)^{\dagger}}{\lambda_a^{E}} = \frac{M_a \rho^B M_a^\dagger}{\lambda_a^{E}}, \;\; \eqand \;\; 
    \lambdaaE \deq \lambdaaA. \label{eq:def_tau}
\end{align}
The discussion on how this ensemble is employed to reduce the rate follows in the sequel.

  {\subsubsection{Random Coding and Expurgation}\label{sec:findingCode}
  In this section, we construct the random coding argument, and simultaneously, define all the conditions that pertain to the construction of a good random code. Subsequently, we randomly generate one code that satisfies these constraints. We then expurgate this code to ensure no repetitions are present.}
  Toward constructing an approximating coherent measurement $U_{\calM}$, randomly and independently select $|\calM|\times|\calK|$ sequences $A^n(m,k)$ according to the following pruned distribution
 \begin{align}\label{def:distribution}
     &\PP\left(A^n(m,k) = a^n\right) = \left\{\begin{array}{cc}
          \dfrac{\lambdaanA}{(1-\varepsilon)} \quad & \mbox{for} \quad a^n \in \mathcal{T}_{\delta}^{(n)}(A)\\
           0 & \quad \mbox{otherwise},
     \end{array} \right. \!\!
 \end{align} 
 where $\varepsilon = \sum_{a^n \notin \mathcal{T}_{\delta}^{(n)}(A)}\lambdaanA$, $\mathcal{T}_{\delta}^{(n)}(A)$ is the $\delta$-typical set corresponding to the distribution $\lambdaaA$ on the set $\calA$, and $\lambdaanA \deq \Pi_{i=1}^n \lambda_{a_i}^A$.
 Let $\mathcal{C}^{(m)}$ denote the codebook $\{A^n(m,k)\}_{k\in\calK}$ for a given $m$, and $\mathcal{C}$ denote the collection of all codebooks $\{\mathcal{C}^{(m)}\}_{m\in\calM}$.
Further, for each $a^n\in \TDelta^{(n)}(A)$ define 
\begin{align}
    \rhotildean \deq \hat{\pi}\pi_{\rho^{\bB}} \pi_{a^n} \rhohata \pi_{a^n} \pi_{\rho^{\bB}} \hat{\pi},
\end{align}
and $\rhotildean  = 0,  $ for $a^n \notin \TDelta^{(n)}(A)$, where $\rhohata \deq \bigotimes_{i} \hat{\rho}^{\bB}_{a_i}$, 
 ${\pi}_{\rho^{\bB}}$ and $ \pi_{a^n}$ are the $\delta-$typical and conditionally typical projectors defined as in \cite[Def. 15.1.3]{wilde_arxivBook} and \cite[Def. 15.2.4]{wilde_arxivBook}, with respect to $\rho^{\bB} = \sum_{a\in\calA}\lambdaaA \rhohata$ and $\rhohata$, respectively, and $ \hat{\pi}$ is the cut-off projector as defined in \cite{winter}.
Using the Average Gentle Measurement Lemma \cite[Lemma 9.4.3]{wilde_arxivBook}, for any given $\epsilon \in (0,1)$, and all sufficiently large $n$ and all sufficiently small $\delta$, we have 
\begin{align} \label{eq:closeness_reference}
    \sum_{a^n\in \mathcal{A}^n}\lambdaanA\|\rhohatan - \rhotildean\|_1 \leq \epsilon.
\end{align}
A detailed proof of the above statement can be found in \cite[Eq. 35]{wilde_e}. 
Using these definitions, construct operators
\begin{align}\label{eq:A_uB_v}
A_{a^n}^{\bBn} &\deq  \gamma_{a^n} \bigg(\sqrt{{\rho^{\bB}}^{\otimes n}}^{-1}\rhotildea\sqrt{{\rho^{\bB}}^{\otimes n}}^{-1}\bigg),\;\;  \gamma_{a^n}\deq  \frac{1-\varepsilon}{1+\eta}\frac{1}{|\calM||\calK|}|\{(m,k):A^{n}(m,k)=a^n\}|,
\end{align}
and $\eta \in (0,1)$ is a parameter that determines the probability of not obtaining a sub-POVM. Note that in the definition of $A_{a^n}^{\bBn}$ the right hand side operates on $\calH_{\bBn}$, however, we define $A_{a^n}$ belonging to $\calL(\calH_{A^n})$. To obtain this, we transform  $A_{a^n}^{\bBn}$ as 
\[A_{a^n} = \sum_{b^n, \bar{b}^n}  \langle b^n|A_{a^n}^{\bBn}|\bar{b}^n\rangle_{\bB} 
 \ket{b^n}\langle\Bar{b}^n|_{B}.\]
 Then construct a sub-POVM $\Gamma^{(n)}$ as
\begin{align}
   \Gamma^{(n)}& \deq \{A_{a^n} \colon a^n \in  \mathcal{T}_{\delta}^{(n)}(A)\}.
\end{align} 
Let $\mathbbm{1}_{\{\mbox{sP}\}}$ denote the indicator random variable corresponding to the event that  $\Gamma^{(n)}$ forms a  sub-POVM. We have the following result.
\begin{prop}
\label{lem:Lemma for not sPOVM}
For any $\epsilon \in (0,1)$, any $\eta \in (0,1)$, any $\delta \in (0,1)$ sufficiently small, and any $n$ sufficiently large, we have 
$\mathbb{E}\left[\mathbbm{1}_{\normalfont \{\mbox{sP}\}}\right]>1-\epsilon,$
if $\frac{1}{n}\left(\log M + \log K \right)> \chi({\lambdaaB,\rhohata})$.
\end{prop}
\begin{proof}
The result follows from \cite{winter}.
\end{proof}

Define the code dependent random variables $\mathit{E}_1$ and $\mathit{E}_2$ as 
\begin{align*}
    \mathit{E}_1 \deq \sum_{m\in \calM}\sum_{k\in\calK} (|\calM||\calK|)^{-{1}} \Tr{\Tilde{\rho}^{\bB}_{m,k}}, \qand \mathit{E}_2 \deq \sum_{m\in \calM}\sum_{k\in\calK} (|\calM||\calK|)^{-{1}} \left\| \Tilde{\rho}^{\bB}_{m,k} - \hat{\rho}^{\bB}_{m,k} \right\|_1,
\end{align*}
where $\hat{\rho}^{\bB}_{m,k}$, and $\Tilde{\rho}^{\bB}_{m,k}$ are used as shorthand notations to denote $\Tilde{\rho}^{\bB}_{a^n(m,k)}$ and $\Tilde{\rho}^{\bB}_{a^n(m,k)}$, respectively.
Further, using the results \cite[Eq. (28) and Eq. (35)]{wilde_e}, for all $\epsilon \in (0,1)$, we have 
$\EE[\mathit{E}_1] \geq 1-\epsilon, $ and $ \EE[\mathit{E}_2] \leq \epsilon, $
for all sufficiently large $n$ and all sufficiently small $\delta > 0$.

Now, considering the ensemble $\{\lambdaaE, \tau_a^E\}$, we construct the operators $\{\tau_{a^n(m,k)}^E\}$ using the codebook $\mathcal{C}$ and the distribution defined in \eqref{def:distribution}, where  $\tau_{a^n}^E \deq \bigotimes_{i} \tau_{a_i}^E$.
For this ensemble, we construct a collection of $n$-letter POVMs, one for each $m\in \calM$, capable of decoding the message $k \in \calK$. In particular, we employ the Holevo POVMs \cite{holevo2019quantum} defined as
\begin{align}
    \xi_k^{(m)} \deq \pi^\tau\pi^{(m)}_{k} \pi^\tau \quad \eqand \quad \Xi^{(m)}_k \deq \left(\sum_{k'\in\calK}\xi^{(m)}_{k'}\right)^{-1/2}\xi^{(m)}_k \left(\sum_{k'\in\calK}\xi^{(m)}_{k'}\right)^{-1/2}, \label{eq:povm_packing}
\end{align}
where $\pi^\tau$ is the $\delta-$typical projector (as in \cite[Def. 15.1.3]{wilde_arxivBook}) defined for the density operator $\tau \deq \sum_{a\in \calA}\lambdaaE \tau_a^E$, and $\pi^{(m)}_{k} $ denotes the strong conditional typical projectors (as in \cite[Def. 15.2.4]{wilde_arxivBook}) for the operators $\tau_{a^n(m,k)}$. For these POVMs, we know the average  probability of error can be made arbitrarily small.
More formally, we have the following.

\begin{prop}\label{lem:packingResult}
    Given the ensemble  $\{\lambda_a^{E}, \tau_a^E\}$ and the collection of POVMs $\{\Xi^{(m)}_k\}_{k}$, for any $\epsilon\in(0,1)$, 
    \begin{align}
        \EE\left[\frac{1}{|\calK|}\sum_{k\in \calK} \Tr{\Xi^{(m)}_k \tau_{k}^{(m)}}\right] \geq 1-\epsilon,
    \end{align}
    for sufficiently small $\delta > 0$ and for all sufficiently large $n$, and for all $m \in \calM$, if $\frac{1}{n}\log{K} < \chi(\{\lambda_a^{E}, \tau_a^E\})$, where $\tau_{k}^{(m)}$ is used as a shorthand for $\tau_{a^n(m,k)}$.
\end{prop}
\begin{proof}
    The proof follows from the result of classical communication over quantum channels \cite{holevo2019quantum} or the packing lemma of \cite[Lemma 16.3.1]{wilde_arxivBook} while making the following identification. For each $m\in \calM$, identify $\calM$ with $\calK$, $\calX$ with $\TDelta^{(n)}(\calA)$, $\{\sigma_{C_m}\}_{m}$ with $\{\tau_{k}^{(m)}\}_k$, $\Pi$ with $\pi^{\tau}$, $\Pi_x$ with $\pi_k^{(m)}$, $d$ with $2^{n(S(E|A)_{\Bar{\tau}} + \Bar{\delta})}$,  $D$ with $2^{n(S(E)_{\Bar{\tau}} - \Bar{\delta})}$, and $\Lambda_m$ with $\Xi^{(m)}_k$, where $\Bar{\tau}^{AE} \deq \sum_a \lambda_a^E \ketbra{a}_A \tensor \tau_a^E$ and $\Bar{\delta}(\delta) \searrow 0$ as $\delta \searrow 0$. 
\end{proof}

\noindent The above result also implies a weaker average result which suffices here. This can be stated as $\EE[\mathit{E}_3] \geq 1-\epsilon$,
for sufficiently small $\delta > 0$ and for all sufficiently large $n$, if $\frac{1}{n}\log{K} < \chi(\{\lambda_a^{E}, \tau_a^E\})$, where
\begin{align}
    \mathit{E}_3 \deq \frac{1}{MK}\sum_{m\in\calM}\sum_{k\in \calK} \Tr{\Xi^{(m)}_k \tau_{k}^{(m)}} \label{eq:packingConstraint}.
\end{align}

Finally, toward finding a good code, we need one last property which is that all its codewords are distinct. In the dual, the quantum channel communication problem \cite{devetak2005private}, Devetak used the double exponential decay of the covering error to argue the existence of an expurgated code for exponentially many covering constraints. However, in the current problem, we have exponentially many packing constraints, with each having only an exponential decay in the error. To resolve this issue, we develop a proof that only requires the average of the packing constraints. However, in such a case, it becomes unclear as to what should be the expurgation strategy. For this, we introduce another event that captures the non-distinctness of the codebook, and expurgate with respect to this event. 
Precisely, we define a codeword $A^n(m,k)$ is bad if there exists $(m',k')\neq (m,k)$ such that $A^n(m,k) = A^n(m',k').$ Let 
\[
E_4 \deq \frac{1}{MK}\sum_{m\in\calM}\sum_{k\in\calK}\11_{\{A^n(m,k) \mbox{ is bad}\}}.
\]
Computing its expectation, we get
\begin{align}
    & \EE[E_4]  = \EE\sq{\frac{1}{MK}\sum_{m\in\calM}\sum_{k\in\calK} \11_{\curlys{\exists (m',k')\neq (m,k) \mbox{ such that } A^n(m,k) = A^n(m',k')}}} \nonumber \\
    & \overset{a}{\leq} \frac{1}{MK}\!\!\!\!\!\sum_{\substack{m, m'\in\calM\\k,k'\in\calK\\(m,k)\neq (m',k')} }
    \sum_{a^n \in \TDelta^{(n)}(A)}\EE\sq{\11_{\curlys{A^n(m,k)=a^n}}}\EE\sq{\11_{\curlys{A^n(m',k')={a}^n}}} 
    \overset{b}{\leq} MK 2^{-n(S(\lambdaaA) -\delta_1)} 
    \overset{}{\leq}\epsilon,
\end{align}
for all sufficiently large $n$ and sufficiently small $\delta>0$ if $\frac{1}{n}\left(\log M + \log K \right)< S(\lambdaaA),$ where (a)  uses the mutual independence of the codewords, and (b) define $\delta_1$ as $\delta_1(\delta,\varepsilon) \searrow 0 $ as $\delta,\varepsilon \searrow 0$.
 Using the Markov inequality and the union bound, we have 
\begin{align*}
    \PP\bigg(\{\IndSP = 1\} &\cap \curly{E_1 \geq 1-\sqrt\epsilon} \cap \curly{E_2 \leq \sqrt\epsilon}\cap \curly{E_3 \geq 1-\sqrt\epsilon}\cap \curly{E_4 \leq \sqrt{\epsilon}}\bigg) \geq 1-5\sqrt{\epsilon}.
\end{align*}
Therefore, for all $\epsilon \in (0,1/25),$ and for all sufficiently small $\delta>0$, for all sufficiently large $n$ there exists a code $\mathcal{C}$ that satisfies the conditions $\{\IndSP = 1\}$, $\curly{E_1 \geq 1-\sqrt\epsilon}$, $\curly{E_2 \leq \sqrt\epsilon}$, $ \curly{E_3 \geq 1-\sqrt\epsilon}$, and $\curly{E_4 \leq {\epsilon}}$, simultaneously if 
\begin{align}\tag{S-0}
    \frac{1}{n}\left(\log M + \log K \right)> \chi({\lambdaaB,\rhohata}),\;\;\; \frac{1}{n}\log{K} < \chi(\{\lambda_a^{E}, \tau_a^E\}),\;\;\;
\frac{1}{n}\left(\log M + \log K \right)< S(\lambdaaA).\label{eq:rateConditions}
\end{align} 
At this point, we choose one such code $\calC$ satisfying all the above conditions, and fix it for the rest of the analysis.

Toward showing that this chosen code achieves the asymptotic performance stated in the theorem statement, we expurgate the code $\calC$ with respect to the random variable $E_4$, ensuring that the code has all distinct codewords. The assumption of codebook being distinct becomes crucial at multiple places in the proof and will be highlighted as necessary.
Since $\curly{E_4 \leq \sqrt{\epsilon}}$ ensures at most $\sqrt{\epsilon}MK$ codewords in $\calC$ are not distinct, we remove $\sqrt{\epsilon}MK$ codewords from $\calC$. This is performed by first removing all the non-distinct codewords, and then further removing some more from the distinct ones arbitrarily (if needed) until we remain with a total of  $(1-\sqrt{\epsilon})MK$ codewords.
Let the expurgated set (the remainder of the codewords) be denoted by $\calC_\calE$, and define the sets $\calC_\calE^{(m)}$ as $\calC_\calE^{(m)} \deq \calC_\calE \cap \calC^{(m)}$. Observe that, all the codewords in $\calC_\calE$ are distinct.
However, as opposed to $\calC$ which was consistent with regards to the size of $\calC^{(m)}$ (equal to $K$ for all $m\in \calM$), $\calC_\calE$ has varying sizes. Therefore, we define $K'_m$ to denote the size of $\calC_\calE^{(m)}$ and $M'$ to denote number of non-empty sets in the collection $\{\calC_\calE^{(m)}\}_{m\in\calM}$. Note that for some $m \in \calM$, $K'_m$ may be zero. 
Let $\calM'$ denote the subset of $\calM$ for which $K'_m > 0$, and let $\calH_{M}'$ denote the corresponding Hilbert space with $\dim(\calH_{M}') = M'+1$.
As is evident, $\sum_{m\in \calM'} K'_m = (1-\sqrt{\epsilon})MK$. In addition, define the set of indices corresponding to the expurgated codebook as $\calI^{(m)}_E \deq \{k: a^n(m,k) \in \calC_\calE^{(m)}\}$ and $\calI_E \deq \{(m,k): a^n(m,k) \in \calC_\calE\}$. 
Further, for the expurgated code, we have
\begin{align}
    {E}'_1 &\deq \frac{1}{(1-\sqrt{\epsilon})|\calM||\calK|}\sum_{m\in \calM'}\sum_{k\in\calI^{(m)}_E}  \Tr{\Tilde{\rho}^{\bB}_{m,k}} \geq 1-2\sqrt{\epsilon}, \label{eq:E1_prime} \\
    {E}'_2 &\deq \frac{1}{(1-\sqrt{\epsilon})|\calM||\calK|}\sum_{m\in \calM'}\sum_{k\in\calI^{(m)}_E}  \left\| \Tilde{\rho}^{\bB}_{m,k} - \hat{\rho}^{\bB}_{m,k} \right\|_1 \leq \frac{\sqrt{\epsilon}}{1-\sqrt{\epsilon}} \leq 2\sqrt{\epsilon} \label{eq:E2_prime}, \\
    E'_3 &\deq \frac{1}{(1-\sqrt{\epsilon})|\calM||\calK|} \sum_{m\in\calM'}\sum_{k\in\calI^{(m)}_E} \Tr{\Xi^{(m)}_k \tau_{k}^{(m)}} \geq 1-2\sqrt{\epsilon}, \label{eq:E3_prime} 
\end{align}
where the inequalities above follow from the fact that codebook $\calC$ satisfies $ \curly{E_1 \geq 1-\sqrt\epsilon}, \curly{E_2 \leq \sqrt\epsilon}$ and $ \curly{E_3 \geq 1-\sqrt\epsilon}$ and that only $\sqrt{{\epsilon}}$ fraction of the code is expurgated. Observe that the event $\curlys{\IndSP=1}$ remains true for the expurgated $\calC_\calE$.  Define the collection 
\[ \Gamma_\calE^{(n)} \deq \{A_{a^n(m,k)}\}_{ m\in \calM', k\in \Imset
} .\]
The collection $ \Gamma_\calE^{(n)}$ is completed using the operator $I - \sum_{m\in\calM'}\sum_{k\in \calI_\calE^{(m)}}A_{a^n(m,k)}$,  and the operator is associated with sequence $a^n_0$ chosen arbitrarily from $\mathcal{A}^n\backslash\mathcal{T}_{\delta}^{(n)}(A) $, i.e., 
\[ A_{a^n_0} \deq I - \sum_{m\in\calM'}\sum_{k\in \calI_\calE^{(m)}}A_{a^n(m,k)}.\]
Corresponding to this expurgated code, we now construct our encoding and decoding operations.

\subsubsection{Encoding and Decoding Isometries}
The encoding isometry $U_{\calE}$ is constructed by concatenating three isometries: (i) the covering isometry $U_{\calM}: \calH_{B^n} \rightarrow \calH_{B^n}\tensor \calH_{M}' \tensor \calH_K $, (ii) the rotation isometry $U_{\calR}: \calH_{B^n}\tensor \calH_M' \tensor \calH_K \rightarrow \calH_{E^n}\tensor \calH_M' \tensor \calH_K$, and (iii) the packing isometry $U_{\calP}: \calH_{E^n}\tensor\calH_{{\bE}}\tensor \calH_{M}' \tensor \calH_K  \rightarrow \calH_{E^n}\tensor\calH_{{\bE}}\tensor \calH_{M}' \tensor \calH_K $, where $\calH_M', \calH_K $ and $\calH_{{\bE}} $ are auxiliary Hilbert spaces with dimensions $M'+1, K+1$, and $K+1$, respectively.

\noindent \textbf{Step 1.1: Covering Isometry\\} 
To define the covering isometry $U_\calM$, we use the completion $[\Gamma_{\calE}^{(n)}]$  as
\begin{align}
    U_\calM &\deq \sum_{m \in \calM'}\sum_{k\in\calI_\calE^{(m)}}\sqrt{A_{a^n(m,k)}}\tensor \ket{m}\tensor\ket{k} + \sqrt{A_{a_0^n}}\tensor \ket{M'}_M\tensor \ket{K}_K.
    \label{def:M_unitary}
\end{align}
Note that, for the chosen code, the event $\{\IndSP = 1\}$  
makes $U_{\calM}$ a valid isometry.
From now on, for the ease of notation, we use $M_{m,k}, \lambda_{m,k}^{\bB},  $ and $A_{m,k}$ to denote the corresponding  $n-$letter objects constructed for the codewords $A^n(m,k)$.

\noindent\textbf{Step 1.2: Rotation Isometry\\}
Although the above covering unitary aims to cover the source, it only does so for the reference system. To be able to apply the next step of packing, we wish to use the post-measured state as side information. This could  be possible if the post-measured state also looks close to being product. For this, we employ a rotation unitary.
A similar operation is discussed in \cite[Fact 6]{anshu2017measurement} with regards to classical-quantum states obtained post measurement. The  construction below generalizes this to a coherent application of a measurement. More formally, for the expurgated code $\calC_\calE$, we construct the states
\begin{align}\label{eq:sigma_hat_tilde}
    \ket{\hat{\sigma}}^{{\bBn}E^nMK} &\deq \sum_{m\in \calM'}\sum_{k\in \calI^{(m)}_\calE}\frac{1}{\sqrt{(1-\sqrt{\epsilon})|\calM||\calK|}}
    \frac{(I\tensor {M_{m,k})}}{\sqrt{\lambda_{m,k}}}\ket{\psi_{\rho}^{\tensor n}}^{\bBn B^n}\tensor \ket{m,k} \eqand \nonumber \\
    \ket{\tilde{\sigma}}^{\bBn B^nMK} &\deq \sum_{m\in \calM'}\sum_{k\in \calI^{(m)}_\calE}\frac{1}{\sqrt{(1-\sqrt{\epsilon})|\calM||\calK|}}\frac{(I\tensor \sqrt{A_{m,k})}}{\sqrt{\delta_{m,k}}}\ket{\psi_{\rho}^{\tensor n}}^{\bBn B^n}\tensor \ket{m,k},
\end{align}
where $\delta_{m,k} \deq \Tr{A_{m,k}{\rho^B}^{\tensor n}} = \gamma \Tr{\Tilde{\rho}^{\bB}_{m,k}}$, and  $\gamma  \deq \frac{1-\varepsilon}{1+\eta}\frac{1}{|\calM||\calK|}.$ {For brevity in notation, we skip the sets in the summations over $m$ or $k$ when summations are performed over the codewords belonging to the set $\calI_\calE$ corresponding to the expurgated codebook $\calC_\calE$.}
Clearly, 
$\ket{\hat{\sigma}}^{{\bBn}E^nMK}$ and $\ket{\tilde{\sigma}}^{\bBn B^nMK}$ are valid states. Now to construct $U_{\calR}$,  consider the following lemma which upper bounds the fidelity.
\begin{lemma}\label{lem:RotationUnitary} For any $\epsilon, \eta \in (0,1)$, 
there exists a collection of isometries $\{U_r(m,k):\calH_{B^n}\rightarrow\calH_{E^n}\}$ 
and a collection of phases $\theta_{m,k}$ such that $F(\ket{\hat{\sigma}}^{{\bBn}E^nMK}, ( I_{\bB}\tensor U_{\calR}) \ket{\tilde{\sigma}}^{\bBn B^nMK})\geq 1-4\sqrt\epsilon$, for all sufficiently large $n$ and all sufficiently small $\delta > 0$, where 
    \begin{align}
        U_{\calR} \deq 
        \sum_{m\in \calM'\cup \curly{M'}}\sum_{k\in \calK\cup \curly{K}} 
        e^{-i\theta_{m,k}}U_r(m,k)\tensor \ketbra{m}_M \tensor \ketbra{k}_K. \label{eq:defUR}
    \end{align}
    and $U_r(m,k) = I$ and $\theta_{m,k} = 0$ for all $(m,k) \in (\calM \times \calK)$ such that $a^n(m,k) \notin \calC_\calE$.
\end{lemma}
\begin{proof}
    The proof of the lemma follows using \eqref{eq:E1_prime}, \eqref{eq:E2_prime} and from the Lemma \ref{lem:coveringSuperposition}. For completeness, we detail the proof in Appendix \ref{proof:lem:RotationUnitary}.
\end{proof}

\noindent\textbf{Step 1.3: Packing Isometry\\}
Observe that by coherently performing the covering and rotation operations, one can show that the source $\rho_B$ can be successfully recovered by the quantum registers $\{\ket{m,k}\}$. This implies that the quantum states $\{\ket{m,k}\}$ can be used by the decoder to faithfully reconstruct the source as per Definition \ref{def:protocolcompression}. As a result, we would require a rate of $\frac{1}{n}\log M + \frac{1}{n}\log K$ which has to be greater than the Holevo information $\chi({\lambdaaB,\rhohata})$, as constrained by Proposition \ref{lem:Lemma for not sPOVM}. However, we intend to further reduce the rate from this Holevo information to the coherent information provided in the statement of the theorem.

One can perhaps argue why we cannot simply release the information in the $\calH_K$ system into the environment (partial tracing)? But as expected for a purely quantum setup, this would lead to the protocol becoming incoherent. More precisely, the subsystem $\calH_{E^n}$ that the encoder has in its possession is entangled with the subsystem $\calH_K$, and tracing out the latter without decoupling the two systems would render the former in a mixed state. Once this entanglement is lost, the decoder would not be able to faithfully reconstruct the source by using such a (mixed) state. 

Therefore, a major task here is to successfully decouple the system $\calH_K$ before releasing it to the environment. To achieve this, we introduce the notion of coherent packing or coherent binning. This notion is built on the idea that the post-measured system present in $\calH_{E^n}$ contains information about the quantum state $\ket{k}_{K}$, and hence, conditioned on the state $\ket{m}$, a copy of the state $\ket{k}$ can be recovered from the state present in subsystem $\calH_{E^n}$, albeit with a small probability of error. Using this copy, we intend to decouple the existing copy of $\ket{k}$ from the latter. However, this new copy can erase (decouple) the original, but will itself still remain.
Therefore, as will become evident in the sequel, we perform the process of erasing the information in $\calH_K$ intrinsically without producing any additional copies.

{Toward this, we employ the packing code consisting of the sub-POVMs $\{\Xi_k^{(m)}\}_{k \in \Imset }$, generated for the ensemble $\{\lambdaaE,\tau_a^E\}$}. We complete this sub-POVM for each $m \in \calM'$ as 
\[ \Xi_K^{(m)} \deq I - \sum_{k\in \Imset}\Xi_k^{(m)}. \]
In addition,
we also make use of Naimark's extension theorem (also provided in Lemma \ref{lem:naimark}). This lemma gives us a collection of orthogonal projectors $\{\Pi_k^{(m)}\}$ each acting on $\calH_{E^n}\tensor \calH_{\bE}$, corresponding to the collection $\{\Xi^{(m)}_k\}_{k}$, such that
\begin{align}\label{eq:equivalence_naimark}
    \Tr{\Pi_k^{(m)} (\tau_{k}^{(m)}\tensor \ketbra{0}_{\bE})} = \Tr{\Xi^{(m)}_k \tau_{k}^{(m)}},
\end{align}
for all $m\in \calM'$ and $k \in \Imset \cup\{K\}$, and $\dim{\calH_{\bE}} = K+1$. 
Finally, we define the packing unitary $U_{\calP}$ as 
\begin{align}
    U_{\calP} \deq \sum_{m\in \calM'}&\left[\sum_{k \in \Imset \cup \{K\}}  \!\!\!\!\!\!\!\Pi_k^{(m)}\!\tensor \!\bigg(\sum_{k' \in \calK\cup \{K\}}e^{i\alpha_{k'}^{(m)}}\ket{(k-k')\!\!\!\!\mod (K+1)}\bra{k'} \bigg)\right] \tensor \!\ketbra{m} \nonumber \\
    & \hspace{3.7in}+ I_{E^n \bE K}\tensor\ketbra{M'}, \label{eq:UP_def}
\end{align}
where the phases $\{\alpha_k^{(m)}\}$ are introduced for later convenience,  and will be specified in the sequel\footnote{Moving forward, we implicitly assume the modulus operation rather than explicitly mentioning it for the purpose of brevity.}.
Note that by using $\Pi_k^{(m)}$ in the above definition, instead of $\Xi_k^{(m)}$ ensures that $U_{\calP}U_{\calP}^{\dagger} = I$, implying $U_{\calP}$ is a valid unitary.


As a result,
 we can express the encoding CPTP map $\calN_{\calE}^{(n)}$ as
\begin{align}
\left(I_{\bBn}\tensor\calN_{\calE}^{(n)}\right)&\left(\ketbra{\psi_{\rho}^{\tensor n}}^{\bBn B^n}\right)\nonumber \\
& \deq \Tr_{\bE E^n K}\left( (I\tensor U_\calP U_{\calR} U_{\calM})\ketbra{\psi_{\rho}^{\tensor n}}^{\bBn B^n} \tensor \ketbra{0}_{\bE} (I\tensor U_\calP U_{\calR} U_{\calM})^\dagger\right).
\end{align}
The quantum state in $\calH_M'$ is now sent to the decoder.

\noindent\textbf{Step 1.4: Decoding Isometry:\\}
The following decoding isometry is applied on the state in $\calH_M'$:
\begin{align}
    U_{\calD} \deq \sum_{m\in \calM' }\bigg(\frac{1}{\sqrt{K'_m}}\sum_{k\in \calI_\calE^{(m)}}e^{-i\beta_{k}^{(m)}}\ket{a^n(m,k)}\bigg)\bra{m}  + \ket{a^n_0}\langle{M'}|,
\end{align}
where the phases $\{\beta_k^{(m)}\}$ will be identified in the continuation. Observe that, to argue $U_{\calD}$ is a valid isometric operation, we need the vectors $\{\ket{a^n(m,k)}\}$ to be distinct. 
By expurgating the codebook to generate $\calC_\calE$, and only using the codewords from $\calC_\calE$ ensures this distinctness.
With the definitions of encoder and decoder, we move on to bounding the error incurred by the protocol (as defined in Definition \ref{def:protocolError}).

\subsubsection{Trace Distance} 
We begin by defining the following terms
\begin{align}
    \ket{\omega}^{\bBn E^n \bE A^n K} &\deq (I\tensor U_{\calD})(I\tensor U_{\calP})(I\tensor U_{\calR}U_{\calM})\ket{\psi_\rho^{\tensor n}}^{\bBn B^n} \ket{0}_{\bE}, \nonumber \\
    \ket{\zeta}^{\bBn E^n A^n} &\deq  ( \bV^{\otimes n} \tensor I) \ket{\psi_{\omega}}^{ \bAn A^n},  
\end{align}
where\footnote{For conciseness, we drop the $\tensor n$ from $\ket{\psi_\rho^{\tensor n}}^{\bBn B^n}$ when understood from context.} $\ket{\psi_{\omega}}^{ \bAn A^n}$ is the canonical purification of $\omega^{A^n}$.
Let \[ G  \deq \|\omega^{\bBn A^n} - \zeta^{\bBn A^n}\|_1.
\]
Following Definition \ref{def:protocolError}, our objective now is to show $G$ can be made arbitrarily small for all sufficiently large $n$ for the code $\calC_\calE$. 


\noindent\textbf{Step 2.1: Closeness of $\ket{\omega} \eqand  (I \tensor U_{\calD})(I\tensor U_{\calP}) \ket{\hat{\sigma}}:$}\\
Recall the definitions of $\ket{\hat{\sigma}} \eqand \ket{\tilde{\sigma}}$ from \eqref{eq:sigma_hat_tilde}, and let $\ket{\omega_1} \deq (I \tensor U_{\calR}U_{\calM})\ket{\psi_\rho}^{\bBn B^n}$ and $\varepsilon_1 \deq {(1-\varepsilon)/(1+\eta)}$.
Consider
\begin{align}
&\sqrt{F\left(\ket{\omega_1}^{\bBn E^n MK},(I\tensor U_{\calR})\ket{\tilde{\sigma}}^{\bBn E^n MK} \right)} \nonumber \\
    & = \sum_{m,k} \frac{1}{\sqrt{(1-\epsFourRoot)|\calM||\calK|}} \frac{\bra{\psi_\rho}(I\tensor U_r(m,k)\sqrt{A_{m,k}})^\dagger ((I\tensor U_r(m,k)\sqrt{A_{m,k}})\ket{\psi_\rho}}{\sqrt{\delta_{m,k}}} \nonumber \\
    & = \frac{1}{\sqrt{1-\epsFourRoot}}\sum_{m,k} \frac{\sqrt{\varepsilon_1}}{{|\calM||\calK|}}\sqrt{ \Tr{\Tilde{\rho}_{m,k}}} 
    \geq  \frac{1}{\sqrt{1-\epsFourRoot}}{\sum_{m,k} \frac{\sqrt{\varepsilon_1}}{{|\calM||\calK|}} \Tr{\Tilde{\rho}_{m,k}} } \geq \sqrt{\varepsilon_1}\sqrt{1-\epsFourRoot} (1-2\epsFourRoot),\label{eq:Fid_omg1_sigmahat}
\end{align}
where we note that there is no overlap between the term corresponding to $\sqrt{A_{a^n_0}}\tensor \ket{M'}_M\tensor\ket{K}_K$ 
of $\ket{\omega_1}^{\bBn E^n MK}$ and the state $(U_{\calR}\tensor I)\ket{\tilde{\sigma}}^{\bBn E^n MK}$, and the last inequality follows from \eqref{eq:E1_prime}.

\noindent Using Lemma \ref{lem:relationshipTraceFidelity}, and the inequality \eqref{eq:Fid_omg1_sigmahat}
, we get \footnote{At times, the subspace notation is omitted when it is clear from the context.}   
\begin{align}
    &\left\| {\omega_1}^{\bBn E^n MK} - (I \otimes U_{\calR} ){\tilde{\sigma}}^{\bBn E^n MK}(I \otimes U_{\calR})^\dagger \right\|_1 \nonumber \\
    &\hspace{1in}\leq 2\sqrt{1 - (1-\epsFourRoot)(1-2\epsFourRoot)^2\left(1-\frac{\eta+\varepsilon}{1+\eta}\right)} \leq 2\sqrt{\frac{\eta+\varepsilon}{1+\eta} + 5\epsFourRoot} \leq 6\sqrt[4]{\epsilon},
\end{align}
for all sufficiently large $n$ and sufficiently small $\eta,\delta>0$.
Further, using the unitary invariance of trace distance, we get the closeness of the states:
\begin{align}
    &\left\| (I \tensor U_{\calD}U_{\calP}){\omega_1}^{\bBn E^n MK}\tensor\ketbra{0}_{\bE}(I \tensor U_{\calD}U_{\calP})^\dagger \right.
    \nonumber \\
    & \hspace{1in} \left.- (I \tensor U_{\calD}U_{\calP}U_{\calR}){\tilde{\sigma}}^{\bBn E^n MK}\tensor\ketbra{0}_{\bE}(I \tensor U_{\calD}U_{\calP}U_{\calR})^\dagger \right\|_1 \leq 6\sqrt[4]{\epsilon}. \label{eq:trace_omg1_sigmatilde}
\end{align}
Using Lemma \ref{lem:RotationUnitary} and the fact that trace norm is invariant under isometric transformations,  we have 
\begin{align}\label{eq:trace_sigmatilde_sigmahat}
    & \left\|(I \tensor U_{\calD} U_{\calP}){\hat{\sigma}^{\bBn E^n MK}}\tensor\ketbra{0}_{\bE}(I \tensor U_{\calD} U_{\calP})^\dagger \right. \nonumber \\
    & \hspace{0.5in} \left.- (I \tensor U_{\calD}U_{\calP}U_{\calR}) {\tilde{\sigma}}^{\bBn E^n MK}\tensor\ketbra{0}_{\bE}(I \tensor U_{\calD}U_{\calP}U_{\calR})^\dagger \right\|_1  \nonumber \\
     &\hspace{0.4in}= \left\|\hat{\sigma} - (I \tensor U_{\calR}){\tilde{\sigma}}(I \tensor U_{\calR})^\dagger \right\|_1  \leq 2\sqrt{1-F\left( \ket{\hat{\sigma}} , (I \tensor U_{\calR})\ket{\tilde{\sigma}}\right) } \leq 4\sqrt[4]{\epsilon}.
\end{align}
Using triangle inequality and inequalities \eqref{eq:trace_omg1_sigmatilde} and \eqref{eq:trace_sigmatilde_sigmahat}, we obtain
\begin{align}\tag{S-1}
    \left\|{\omega}^{\bBn E^n \bE A^n K} -(I \tensor U_{\calD}U_{\calP})({\hat{\sigma} }^{\bBn E^n MK} \tensor \ketbra{0}_{\bE}) (I \tensor U_{\calD}U_{\calP})^\dagger\right\|_1 \leq 10\sqrt[4]{\epsilon} , \label{eq:trace_omg_sigmahat}
\end{align}
for all sufficiently large $n$ and sufficiently small $\eta,\delta>0$, 
which concludes Step 2.1.

\noindent For the next step, define $|{\hat\zeta}\rangle^{\bBn E^n\bE MK}$ as
\begin{align}
    \zetahat^{\bBn E^n\bE MK} \deq \sum_{m\in\calM'}\sum_{k\in \calI_\calE^{(m)}}\normMK e^{i\beta_k^{(m)}} \frac{(I\tensor M_{m,k})\ket{\psi_{\rho}}}{\sqrt{\lambda_{m,k}}}\tensor \ket{m}_M \tensor \ket{0}_K \tensor \ket{0}_{\bE}, \label{eq:zetahat}
\end{align}
where the phases $\beta_k^{(m)}$ will be specified shortly. Observe that $\zetahat^{\bBn E^n\bE MK} $ is a valid pure state due to (i) the
distinctness of codewords in $\calC_\calE$ and (ii) the identity \eqref{eq:Ma_orthogonal}. Furthermore, in its definition, the information in the subsystem $\calH_K$ is decoupled from the remaining subsystems. Since $U_\calP$ acts on $\calH_{E^n}\tensor \calH_{\bE}$, an additional pure ancilla is attached for appropriate comparisons. We aim to show that this state is close to the action of $(I \tensor U_{\calP})$ on the state $ \ket{\hat{\sigma}}\ket{0}_{\bE}$.

\noindent\textbf{Step 2.2: Closeness of $ (I\tensor U_{\calP})\ket{\hat{\sigma}}\ket{0}_{\bE}$ and $ \zetahat:$}\\
We begin by simplifying $(I\tensor U_{\calP})\ket{\hat{\sigma}} \ket{0}_{\bE}$ as
\begin{align}
    &(I\tensor U_{\calP})\ket{\hat{\sigma}} \ket{0}_{\bE}
    =  \sum_{m\in\calM'}\sum_{k\in \calI_\calE^{(m)}}\normMK e^{i\alpha_{k}^{(m)}} |{\phi_{k}^{(m)}}\rangle \tensor\ket{m}_M, \nonumber
\end{align}
where
\begin{align}
\label{eq:phimk_defnition}
    |{\phi_{k}^{(m)}}\rangle \deq \sum_{k'\in \Imset\cup \{K\}}\frac{(I\tensor \Pi_{k'}^{(m)} M_{m,k})\ket{\psi_{\rho}} \ket{0}_{\bE}}{\sqrt{\lambda_{m,k}}} \tensor \ket{k'-k}_K, \mbox{ for all } k \in \calI_\calE^{(m)} \mbox{ and } m\in \calM'.
\end{align}
Similarly, let 
\begin{align}
\label{eq:chimk_defnition}
    \zetahat = \sum_{m,k}\normMK e^{i\beta^{(m)}_k} |{\chi_{k}^{(m)}}\rangle\tensor\ket{m}_M, \;\; \;\;|{\chi_{k}^{(m)}}\rangle \deq \frac{(I\tensor M_{m,k})\ket{\psi_{\rho}}\ket{0}_{\bE}}{\sqrt{\lambda_{m,k}}} \tensor \ket{0}_K,
\end{align}
for all $m\in \calM'$ and $k \in \Imset$, and the phases $\{\beta_k^{(m)}\}$ are the same phases incorporated in the construction of the decoding isometry $U_\calD$.
Further, from \eqref{eq:equivalence_naimark}, we know for all $m \in \calM'$,
\begin{align}
   \frac{1}{K'_m}\sum_{k \in \calI_\calE^{(m)}} \langle{\phi_{k}^{(m)}}|{\chi_{k}^{(m)}}\rangle =  
   \frac{1}{K'_m}\sum_{k \in \calI_\calE^{(m)}} \Tr{\Xi^{(m)}_k \tau_{k}^{(m)}}.
   \label{eq:avg_packing_error}
\end{align}
Now the fidelity between  $ \zetahat \eqand (I \tensor U_{\calP})\ket{\hat{\sigma}}\ket{0}_{\bE}$ can be written as
\begin{align}
    \sqrt{F\left( (I \tensor U_{\calP})\ket{\hat{\sigma}} \ket{0}_{\bE} ,\zetahat \right)} = \frac{1}{M'} \left|\sum_{m\in \calM'}\bra{\phi_m}\ket{\chi_m}\right|,
    \label{eq:fid_packing1}
\end{align}
where, for all $m\in \calM'$,
\begin{align}
    \ket{\phi_m} \deq c\sum_{k\in \Imset} e^{i\alpha_k^{(m)}}|{\phi_{k}^{(m)}}\rangle  \qand 
    \ket{\chi_m} \deq c\sum_{k\in \Imset} e^{i\beta_k^{(m)}}|{\chi_{k}^{(m)}}\rangle, \nonumber
\end{align}
and $c \deq \sqrt{\frac{M'}{(1-\epsFourRoot)|\calM||\calK|}}$. Toward a lower bound on the fidelity, we provide the following proposition.
\begin{prop}\label{prop:ExpectedFourier}
    For any $\epsilon \in (0,1)$, there exists phases $ \{\alpha_k^{(m)}\},$ and $ \curlys{\beta_k^{(m)}}$ such that 
    \begin{align}
       \left| \frac{1}{M'} \sum_{m\in \calM'}\bra{\phi_m}\ket{\chi_m}\right| \geq 1-2\sqrt\epsilon,
    \end{align}
    for all sufficiently small $\delta > 0$ and all sufficiently large $n$.
\end{prop}
\begin{proof}
    The proof is provided in Appendix \ref{proof:prop:ExpectedFourier}.
\end{proof}
\noindent Observe that, using the relation in Lemma \ref{lem:relationshipTraceFidelity}, and the result of Proposition \ref{prop:ExpectedFourier} and \eqref{eq:fid_packing1}, we get
\begin{align}\tag{S-2}
    \left\| (I\tensor U_{\calP}){\hat{\sigma}}\tensor \ketbra{0}_{\bE}(I\tensor U_{\calP})^\dagger - \hat\zeta \right\|_1  \leq 2\sqrt{1- F\left( (U_{\calP}\tensor I)\ket{\hat{\sigma}}\ket{0}_{\bE} ,\zetahat \right) } \leq 4\sqrt[4]{\epsilon}, \label{eq:step2_2Result1}
\end{align}
for all sufficiently large $n$, and sufficiently small $\eta,\delta > 0$.
Observe that $\zetahat^{\bBn E^n MK} = \zetahat^{\bBn E^n M} \tensor \ket{0}_{K\bE}$, and hence $\zetahat^{\bBn E^n M}$ remains pure after partial tracing over the subsystem $\calH_K\tensor \calH_{\bE}$.
Finally, we are left with showing the closeness of the state $ (I\tensor U_{\calD})\zetahat^{\bBn E^n M}$ with the state $\ket{\zeta}^{\bBn E^n A^n}$.

\noindent\textbf{Step 2.3: Closeness of $(I \tensor U_{\calD})\zetahat^{\bBn E^n M}$ and $\ket{\zeta}^{\bBn E^n A^n}:$}\\
We begin by defining $\sigma^{A^n}$ as
\begin{align}
    \sigma^{A^n} &\deq \Tr_{\bBn E^n}\{(I\tensor U_{\calD})\hat\zeta^{\bBn E^n M}(I\tensor U_{\calD})^\dagger\}, \nonumber 
\end{align}
and perform the simplification
\begin{align}
 \sigma^{A^n} &= U_{\calD} \bigg( \sum_{m,m'}\sum_{k,k'}\normsqMK e^{-i(\beta_{k}^{(m)} - \beta_{k'}^{(m')})}\frac{\Tr\big(M_{mk}\rho^B M_{m'k'}^{\dagger}\big)}{\sqrt{\lambda_{mk}\lambda_{m'k'}}}|{m}\rangle\langle{m'}|\bigg)U_{\calD}^{\dagger} \nonumber \\
    & = U_{\calD} \left( \sum_{m}\frac{K'_m}{{{(1-\epsFourRoot)|\calM||\calK|}}}\ketbra{m}\right)U_{\calD}^{\dagger} = \sum_{m}\frac{K'_m}{{{(1-\epsFourRoot)|\calM||\calK|}}}\ketbra{b^n(m)}^{A^n} , 
\end{align}
where the second equality uses \eqref{eq:Ma_orthogonal} and the crucial condition that the codebook $\calC_\calE$ obtained after expurgation is distinct, and last equality defines $\ket{b^n(m)}^{A^n}$ as
\begin{equation}\ket{b^n(m)}^{A^n} \deq \frac{1}{\sqrt{K'_m}} \sum_{k\in \calI_\calE^{(m)}}e^{-i\beta_k^{(m)}}\ket{a^n(m,k)}^{A^n}, \label{eq:bm_definition}\end{equation} for all $m\in \calM'$.
This implies, we can write the canonical purification  of $\sigma^{A^n}$ as 
\begin{align}
    \ket{\psi_\sigma}^{\bAn A^n} \deq (I_{\bAn}\tensor \sqrt{\sigma^{A^n}})&\ket{\Gamma^{\tensor n}}^{\bAn A^n } = \sum_{m}\sqrt{\frac{K'_m}{{{(1-\epsFourRoot)|\calM||\calK|}}}}\ket{b^n(m)}^{\bAn}\tensor \ket{b^n(m)}^{A^n} \nonumber \\
    & = (I \tensor U_{\calD}^{A^n} ) \sum_{m}\sqrt{\frac{K'_m}{{{(1-\epsFourRoot)|\calM||\calK|}}}}\ket{b^n(m)}^{\bAn}\tensor\ket{m}^{M}, 
\end{align}
where the first equality follows by defining $ \ket{b^n(m)}^{\bAn} \deq (I_{\bAn}\tensor \bra{b^n(m)}^{A^n})\ket{\Gamma^{\tensor n}}^{\bAn A^n }$.
Using the relation from \eqref{eq:actionVbarOnA} and definition \eqref{eq:bm_definition}, we can write
\begin{align}
   \bV^{\tensor n}\ket{b^n(m)}^{\bAn}  &= \!\frac{1}{\sqrt{K'_m}} \!\sum_{k\in \calI_\calE^{(m)}}\!\!\!e^{i\beta_k^{(m)}}\bV^{\tensor n}\ket{a^n(m,k)}^{\bAn} =   \!\frac{1}{\sqrt{K'_m}}\! \sum_{k\in \calI_\calE^{(m)}}\!\!\!e^{i\beta_k^{(m)}}\frac{(I_{\bB}\!\tensor \!M_{m,k}) \ket{\psi_\rho^{\tensor n}}^{\bBn B^n}}{\sqrt{\lambda_{m,k}^A}},\nonumber
\end{align}
for all $m\in \calM'$, which gives
\begin{align}
    (\bV^{\tensor n} \tensor I_{A^n}) \ket{\psi_\sigma}^{\bAn A^n } 
    & = (I \tensor U_{\calD}^{A^n} ) \sum_{m}\sqrt{\frac{K'_m}{{{(1-\epsFourRoot)|\calM||\calK|}}}}\bV^{\tensor n}\ket{b^n(m)}^{\bAn}\tensor\ket{m}^{M}\nonumber \\
    & = (I \tensor U_{\calD}^{A^n} )\sum_{m,k}\normMK e^{i\beta_k^{(m)}}\frac{(I_{\bB}\tensor M_{m,k}) \ket{\psi_\rho^{\tensor n}}^{\bBn B^n}}{\sqrt{\lambda_{m,k}^A}} \tensor \ket{m}^{M} \nonumber \\
    & = (I \tensor U_{\calD}^{A^n}) \zetahat^{\bBn E^n M},
\end{align}
where the last equality follows from the definition of $\zetahat^{\bBn E^n M}$ in \eqref{eq:zetahat}. 
Observing that \[\ket\zeta^{\bBn E^n A^n} = (\bV^{\tensor n} \tensor I_{A^n})\ket{\psi_\omega}^{\bAn A^n},\]
we are left with showing the closeness of $\ket{\psi_\sigma}^{\bAn A^n }$ and  $\ket{\psi_\omega}^{\bAn A^n }$. 
Using \eqref{eq:trace_omg_sigmahat} and \eqref{eq:step2_2Result1}, the triangle inequality, the monotonicity of trace distance, and identification of appropriate purifications, we obtain for all sufficiently large $n$ and sufficiently small $\eta,\delta> 0$,
\begin{align}
     \| \omega^{A^n} - \sigma^{A^n} \|_1
    & \leq
    \| \omega^{\bBn E^n A^n} - (I \tensor U_{\calD})\hat\zeta^{\bBn E^n M}(I \tensor U_{\calD})^\dagger \|_1 
    \leq 14\sqrt[4]{\epsilon}.\nonumber
\end{align}
This implies,
{\begin{align}
     \left\|(I \tensor U_{\calD}^{A^n}) \hat{\zeta}^{\bBn E^n M}(I \tensor U_{\calD}^{A^n})^\dagger -  \zeta^{\bBn E^n A^n}\right\|_1 
  & \overset{a}{=}   \left\|{\psi_\sigma}^{A^n \bAn} -  {\psi_\omega}^{A^n \bAn}\right\|_1  \nonumber \\
  & \overset{b}{\leq} 2{\sqrt{1-F(|{\psi_\sigma}\rangle^{A^n \bAn},|{\psi_\omega}\rangle^{A^n \bAn} )}}  \nonumber \\
  & {\overset{c}{\leq} 2\sqrt{\| \omega^{A^n} - \sigma^{A^n} \|_1 } \leq 2\sqrt{{14\sqrt[4]{\epsilon}}} \leq 8\sqrt[8]{\epsilon},} \label{eq:step3}\tag{S-3}
\end{align}}
for all sufficiently large $n$ and sufficiently small $\eta,\delta>0$, where (a) follows from the isometric invariance of trace distance, (b) uses Lemma \ref{lem:relationshipTraceFidelity}, and (c) uses Lemma \ref{lem:closenessofPurification}, which concludes this step.

In summary, combining results of \eqref{eq:rateConditions}, \eqref{eq:trace_omg_sigmahat}, \eqref{eq:step2_2Result1} and \eqref{eq:step3}, we have showed that there exist a code $\calC$ satisfying $G \leq {14\sqrt[4]{\epsilon} + 8\sqrt[8]{\epsilon}} $ with the following rate constraints:
\begin{align}
     S(\lambdaaA)>\frac{1}{n}\left(\log M + \log K \right) &> \chi({\lambdaaB,\rhohata}), \;\;
    \frac{1}{n}\log K< \chi(\{\lambda_a^{E}, \tau_a^E\}), \;\;
     \frac{1}{n}\log{M} \geq 0 , \quad  \frac{1}{n}\log K\geq 0,\nonumber
\end{align}
for all sufficiently large $n$ and sufficiently small $\eta,\delta>0$,
where we have also included the necessary non-negativity constraints.
Eliminating $\frac{1}{n}\log{K}$ using Fourier-Motzkin elimination \cite{fourier-motzkin} gives
\[ \frac{1}{n}\log{M} > \chi({\lambdaaB,\rhohata}) - \chi(\{\lambda_a^{E}, \tau_a^E\}), \qand \frac{1}{n}\log{M} \geq 0, \]
where we remove the redundant constraints. This completes the proof. 

\begin{remark}[Zero performance rate]
\label{remark:zero_rate}
The coherent information $I_c(\calN_W,\rho^{A_R})$ is negative when the Holevo information quantities are such that $\chi({\lambdaaB,\rhohata}) < \chi(\{\lambda_a^{E}, \tau_a^E\})$. The asymptotic performance limit for these situations is zero, according to the statement of the theorem. We must therefore demonstrate that a rate of zero is feasible. To put it another way, we must construct a protocol (see Definition \eqref{def:protocolcompression}) that  satisfies \eqref{def:protocolError} with a rate that can be made arbitrarily close to zero. The constraints imposed by the preceding proof are still met if we select $M = 1$ and $\frac{1}{n}\log K = \chi({\lambdaaB,\rhohata}) + \delta_0 < \chi(\{\lambda_a^{E}, \tau_a^E\})$, while achieving a rate of $log(2)/n$,
for a sufficiently small $\delta_0$.
This rate $(1/n)$ can be made arbitrarily close to zero (i.e., smaller than the provided $\epsilon$) for any given $\epsilon$, for all sufficiently large $n$ and sufficiently small $\eta,\delta > 0.$
Similarly, when the coherent information $I_c(\calN_W,\rho^{A_R})$ is exactly zero, i.e., $\chi({\lambdaaB,\rhohata}) = \chi(\{\lambda_a^{E}, \tau_a^E\})$, we choose $M$ and $K$ such that $\frac{1}{n}\log M = 2\delta_0 $, and $\frac{1}{n}\log K =\chi(\{\lambda_a^{E}, \tau_a^E\}) - \delta_0$. This gives a rate of $2\delta_0$ which can be again made arbitrarily close to zero. Therefore, even though the coherent information is not necessarily positive, the rate in the theorem can still be achieved.

\end{remark}


\subsection{Proof of Converse}
\label{QLSC:sec:proofOuterBound}
Let $R$ be an achievable rate. Then  from Definition \ref{def:achievability},
given a triple ($\rho_B$,
$\calH_A$, $\calN_W$), for all
$\epsilon>0$, and all sufficiently large $n$, there exists $(n,\Theta)$ lossy compression protocol with an encoding CPTP map $\calN_\calE^{(n)}$ and a decoding CPTP map $\calN_\calD^{(n)}$ that satisfies the following constraints:
\begin{align}\label{eq:cons_conv1}
c_0: \frac{1}{n}\log{\Theta} \leq R + \epsilon, \qand c_1:    \|\omega^{\bBn A^n} - \upsilon^{\bBn A^n}\|_1 \leq \epsilon,
\end{align}
where $\omega^{\bBn A^n} \deq (I \tensor \calN_\calD^{(n)})(I \tensor \calN_\calE^{(n)} ) (\ket{\psi_{\rho}}^{\bBn B^n})$, 
\[
\upsilon^{\bBn A^n} =\Tr_{E^n} \left\{\upsilon^{{B}^n_R A^n E^n} \right\}\deq \Tr_{E^n}\left\{(\bV^{\otimes n} \otimes I) \Psi_{\omega}^{\bAn A^n}(\bV^{\otimes n} \otimes I)^{\dagger}\right\},
\] and $\ket{\psi_{\omega}}^{A^n \bA^n}$ is the canonical purification of $\omega^{A^n}$, and $\bV$ is the Stinespring's dilation of the CPTP map $\calN_W$. 
Let $\omega^{A_R^n}\deq \Tr_{A^n} (\Psi_\omega^{A^n A_R^n})$.

\noindent \textbf{Step 1: Quantum Data Processing Inequality:} 
Let $M$ denote the quantum state at the output of the encoder.
Let $V_\calE^{(n)}: \calH_{B^n} \rightarrow 
\calH_M \tensor \calH_{\tilde{E}_1}$ and 
$V_\calD^{(n)}:\calH_{M} \rightarrow \calH_{A^n} \tensor \calH_{\tilde{E}_2}$ be Stinesping dilations of encoding and decoding maps $\calN_{\calE}^{(n)}$ and 
$\calN_{\calD}^{(n)}$, respectively, such that $\dim(\calH_{\Tilde{E}_1}) \geq \dim(\calH_M)$ and $\dim(\calH_{\Tilde{E}_2}) \geq \dim(\calH_{A^n})$,  as shown in Figure \ref{fig:q_converse}(a).
Let
\begin{align}
    {\omega_1}^{\bBn M \tilde{E}_1} &\deq (I_{\bBn}\tensor V_\calE^{(n)})(\Psi_{\rho}^{\bBn B^n})(I_{\bBn}\tensor V_\calE^{(n)})^\dagger\nonumber
     \\
     \quad {\omega}^{\bBn \tilde{E}_1 \tilde{E}_2 A^n} &\deq (I_{\bBn \Tilde{E}_1}\tensor V_\calD^{(n)})({\omega_1}^{\bBn \tilde{E}_1 M})(I_{\bBn\Tilde{E}_1}\tensor V_\calD^{(n)})^\dagger.
\end{align}
Let $|\psi_{\omega_1}\rangle ^{M_R M}$ denote the canonical purification of the quantum state $\omega_1^M$. Let 
$W_\calE^{(n)}: \calH_{M_R} \rightarrow \calH_{\bBn}\tensor \calH_{\Tilde{E}_1}$ denote the posterior reference isometry  (see Definition \ref{def:VBar}) of
$V_\calE^{(n)}$ with respect to $\omega_1^M$, as shown in Figure \ref{fig:q_converse}(b). Moreover, let 
$W_\calD^{(n)}: \calH_{\bAn} \rightarrow \calH_{M_R} \tensor \calH_{\Tilde{E}_2}$ denote the posterior reference isometry of
$V_\calD^{(n)}$ with respect to $w^{A^n}$, as shown in Figure \ref{fig:q_converse}(c). Let $\calN_{W_{\calE}}(\cdot)=\Tr_{\tilde{E_1}}(W_{\calE}^{(n)} \cdot (W_{\calE}^{(n)})^{\dagger})$
and $\calN_{W_{\calD}}(\cdot)=\Tr_{\tilde{E_2}}(W_{\calD}^{(n)} \cdot (W_{\calD}^{(n)})^{\dagger})$
be the induced CPTP maps. 
Let $$\tilde{\omega}_1^{M_R \tilde{E}_2 A^n} \deq (W_\calD^{(n)}\tensor I_{A^n})(\Psi_{\omega}^{\bAn A^n})(W_\calD^{(n)}\tensor I_{A^n})^\dagger.$$ 
Using the quantum data processing inequality for coherent information \cite[Theorem 11.3.2]{wilde_arxivBook}, we obtain
\[
I_c(\calN_{W_{\calD}},\omega^{A_R^n}) \geq 
I_c(\calN_{W_{\calE}}\circ 
\calN_{W_{\calD}}, \omega^{A_R^n}).
\]
Expanding the coherent information in terms of  Von Neuman entropy, we get
\[
S(M_R)_{\tilde\omega_1}-S(\tilde{E}_2)_{\tilde\omega_1} \geq S(B_R^n)_\omega-S(\tilde{E}_1\tilde{E}_2)_\omega,
\]
which implies that 
\begin{equation}
    S(M)_{\omega_1} \geq S(B_R^n)_\omega-S(B_R^nA^n)_\omega.
\label{eq:QDPI}
\end{equation}

\noindent \textbf{Step 2: Implication of the constraints $c_0$ and $c_1$:}
Consider the following sequence of inequalities:
\begin{figure}
    \centering
    \includegraphics[scale=0.7]{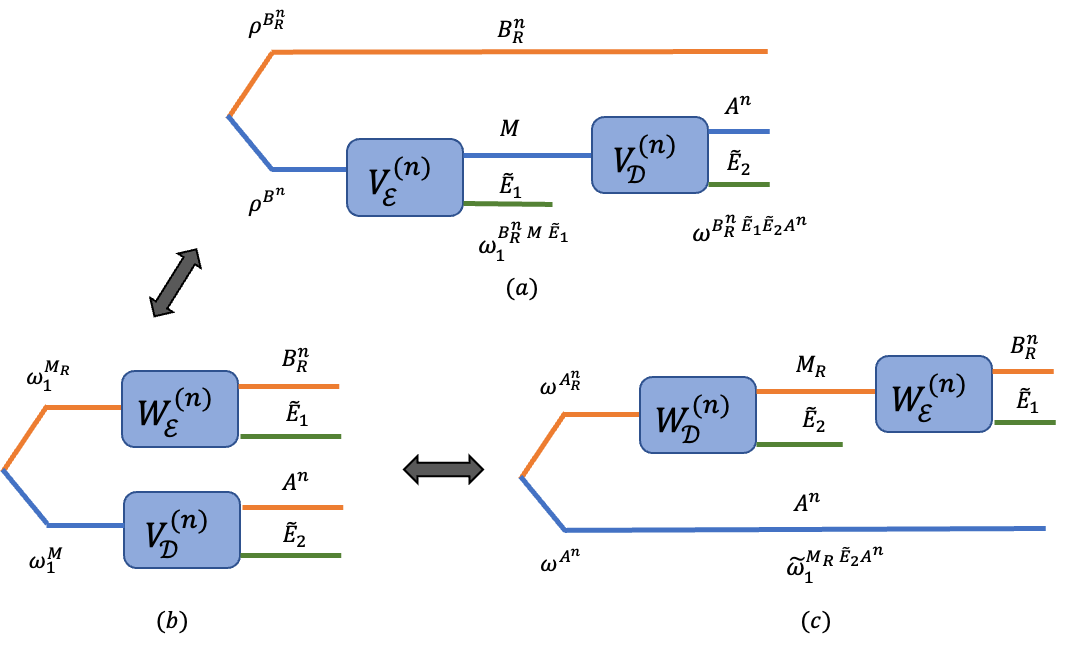}
    \vspace{-0.2in}
    \caption{Lossy quantum source coding protocol and the associated CPTP maps and their Stinespring dilations.}
    \vspace{-0.2in}
    \label{fig:q_converse}
\end{figure}
\begin{align}
nR &\geq \log \Theta -n \epsilon  \geq S(M)_{\omega_1}  -n \epsilon \\
&\overset{a}{\geq} S({B}^n_R)_\omega-S({B}^n_R,A^n)_\omega  -n \epsilon \\
&\overset{b}{\geq} S({B}^n_R)_\omega- S({B}^n_R,A^n)_\upsilon  -n \epsilon-n\tilde{\epsilon}_1 \\
&=S({B}^n_R)_\omega- S(E^n)_\upsilon  -n \epsilon-n\tilde{\epsilon}_1 \\
&\overset{c}{\geq} S({B}^n_R)_\omega- \sum_{i=1}^n S(E_i)_\upsilon  -n \epsilon-n\tilde{\epsilon}_1 \\
&\overset{d}{=}\sum_{i=1}^n S(\bB_i)_\omega-\sum_{i=1}^n S(E_i)_\upsilon  -n \epsilon-n\tilde{\epsilon}_1 \\
&\overset{e}{\geq}  \sum_{i=1}^n S(\bB_i)_\upsilon-\sum_{i=1}^n S(E_i)_\upsilon  -n \epsilon-n\tilde{\epsilon}_1-n\tilde{\epsilon}_2 \\
&\overset{f}{=}\sum_{i=1}^n I_c(\calN_W,\omega^{{A}_{Ri}})   -n \epsilon-n\tilde{\epsilon}_1-n\tilde{\epsilon}_2 \\
&\overset{g}{\geq} n \min_{\rho^{\bA} \in \calD(\calH_{\bA}): \| \rho^{\bB}-\mathcal{N}_W(\rho^{\bA}) \|_1 \leq \epsilon} 
I_c(\calN_W,\rho^{{A}_R})   -n \epsilon-n\tilde{\epsilon}_1-n\tilde{\epsilon}_2 \label{eq:conv_l},
\end{align}
where the inequalities are argued as follows.
(a) follows from \eqref{eq:QDPI}. (b) follows from the condition $c_1$ and the Fannes-Audenaert inequality \cite[Theorem 11.10.1]{wilde_arxivBook}
by defining $\tilde{\epsilon}_1\deq\epsilon \log |\calH_A| |\calH_B|+h_b(\epsilon)$. (c) follows from the subadditivity of entropy. (d) follows from the memoylessness of the quantum source.
(e) follows from condition $c_2$ and the Fannes-Audenaert inequality \cite[Theorem 11.10.1]{wilde_arxivBook}, where 
condition 
\[
c_2: \| \omega^{\bB_i} -\upsilon^{\bB_i}\|_1 \leq \epsilon, \ \  \forall\; 1 \leq i \leq n,
\] is implied by $c_1$ using the monotonicity of trace distance with respect to partial trace. $\tilde{\epsilon}_2$ is defined as 
$\tilde{\epsilon}_2 \deq \epsilon \log|\calH_B|+h_b(\epsilon)$. 
(f) follows from the fact that 
$\upsilon^{\bB_i E_i}=W \omega^{{A}_{Ri}} W^{\dagger}$.
(g) follows from condition $c_2$ which can also be stated as
\[
c_2: \| \rho^{\bB} -\calN_W(\omega^{\bA_i})\|_1 \leq \epsilon, \ \  \forall \;1 \leq i \leq n,
\]
and the fact that coherent information is continuous, and the constraint set is closed and bounded. The continuity follows from the following arguments: for the fixed CPTP map $\calN_W$, let a function $f: \calD(\calH_{\bA}) \rightarrow \RR$ be defined as $f(\rho^{\bA}) = I_c(\calN_W,\rho^{{A}_R}).$
    One can establish the continuity of $f$ for a fixed $\calN_W$ by writing $I_c(\calN_W,\rho^{{A}_R}) = S(\bB)_{W\rho^{\bA}W^\dagger } - S(E)_{W\rho^{\bA}W^\dagger }, $ and using the Fannes–Audenaert Inequality \cite[Theorem 11.10.2]{wilde_arxivBook}, where $W$ is the Stinespring's extension of the given CPTP map $\calN_W$.

\noindent \textbf{Step 3: Continuity Argument:} 
We have shown that 
\[
R \in \bigcap_{\epsilon>0} \calI_{\epsilon},
\]
where we have defined for all $\epsilon \geq 0$, 
\begin{align}
    \calI_\epsilon &\deq \curly{R: \exists \; \rho^{{A}_R} \in \calS_\epsilon(\rho^{B},\calN_W) \mbox{ such that  } R \geq I_c(\calN_W,\rho^{{A}_R}) - g(\epsilon) }, 
\end{align}
and
\begin{align}
   \calS_\epsilon(\rho^{B},\calN_W) &\deq \curly{ \rho^{{A}_R} \in \calD(\calH_{\bA}):  \| \mathcal{N}_W(\rho^{\bA})-\rho^{\bB} \|_1 \leq \epsilon }, 
\end{align}
 $g(\epsilon) \deq \epsilon+\tilde{\epsilon}_1+\tilde{\epsilon}_2.$
Condition $c_2$ ensures that the set $\calS_\epsilon$ is non-empty for all $\epsilon>0$.
Now, by arguing continuity of $\calI_\epsilon$ at $\epsilon = 0$, we obtain the desired result.
\begin{lemma} \label{lem:intersectionSets}
    For the above definitions of $\calS_\epsilon$ and $\calI_{\epsilon}$, we have
    $\calS_0(\rho^{B},\calN_W)$ non-empty, and
    \[\calI_0 = \bigcap_{\epsilon>0}\calI_\epsilon.\]
\end{lemma}
\begin{proof}
This is a standard argument used in the literature \cite{dueck1981strong,ahlswede2006strong,cuff2013distributed}.  
    A proof is provided in Appendix \ref{proof:lem:intersectionSets} for completeness.
\end{proof}
This completes the proof.

\section{Proof of Theorem \ref{thm:qclossysourcecoding}}\label{sec:qcproof}
\subsection{Proof of Achievability}
\label{sec:qc_Achievability}
For a given $(\sourcedo,\calX,\calW)$ QC source coding setup, we choose a reconstruction distribution $\targetpx \in \calA(\sourcedo,\calW)$. 
Toward specifying the POVM $\Gamma^{(n)}$ and a decoding map $f:\curly{1,2,\cdots,\Theta} \rightarrow \calX^n$, we construct a codebook $\codebook$. From now on, we let $\Theta = 2^{nR}$.  



\subsubsection{Codebook Design} We generate a codebook $\codebook$ consisting of $n$-length codewords by randomly and independently selecting $2^{nR}$ sequences $\curly{\Xn(m)}_{m\in [2^{nR}]}$ according to the following pruned distribution:
 \begin{align}\label{def:qc_distribution}
     &\codeDistribution(\Xn(m) = \xn) = \left\{\begin{array}{cc}
          \dfrac{\targetpx^n(\xn)}{(1-\varepsilon)} \quad & \mbox{for} \quad \xn \in \Txqc\\
           0 & \quad \mbox{otherwise}
     \end{array} \right. \!\!,
 \end{align} 
 where $ \targetpx^n(\xn) = \prod_{i=1}^n \targetpx(x_i)$, $\Txqc$ is the $\delta$-typical set corresponding to the distribution $\px$ on the set $\calX$, and $\varepsilon(\delta,n) \triangleq \sum_{\xn \not \in \Txqc} \targetpx^n(\xn)$. Note that $\varepsilon(\delta,n) \searrow 0$ as $n \rightarrow \infty$ and for all sufficiently small $\delta > 0$. The generated codebook $\calC$ is revealed to both the encoder and decoder before the QC lossy source compression protocol begins.
 
\subsubsection{Construction of POVM}
We use Winter's POVM construction \cite{winter}. Let $\pi_{\sourcedo}$ and $\pi_{\xn}$ denote the $\delta$-typical and conditional $\delta$-typical projectors defined as in \cite[Def. 15.1.3]{wilde_arxivBook} and \cite[Def. 15.2.4]{wilde_arxivBook}, with respect to $\sourcedo$ and $\calW$, respectively. Consider the following positive operators with a trace of less than one, and we exploit the random selection of these operators to construct the sub-POVM $\curly{A_{\xn}}$. For all $\xn \in \Txqc$, define:
\begin{align}
    \xi_{\xn} &\deq  \pi_{\sourcedo}\pi_{\xn}\calW_{\xn}\pi_{\xn}\pi_{\sourcedo}, 
\end{align}
 and $\xi_{\xn} = 0$ for $\xn \not \in \Txqc$, where $\calW_{\xn} \deq \bigotimes_i\calW_{x_i}$. We 
now define $\xi$ as the expectation of $\xi_{\xn}$ with respect to the pruned distribution $\codeDistribution$ as defined in \eqref{def:qc_distribution}:
\[\xi \deq \EE_{\codeDistribution}\sq{\xi_{\xn}} = \sum_{\xn \in \Txqc} \frac{\targetpx^n(\xn) }{1-\varepsilon}\xi_{\xn}.\]
Let $\hat{\pi}$ be the cut-off projector onto the subspaces spanned by the eigenstates of $\xi$ with eigenvalues greater than $\epsilon d$, where $d \deq 2^{-n(H(\sourcedo)+\delta_1)}$ and $\delta_1$ will be specified later. With the above notation, we define  
\begin{align}\label{eqn:rhotilde}
    \rhotilde_{\xn} \deq \hat{\pi} \xi_{\xn} \hat{\pi} \quad \eqand \quad \rhotilde \deq \EE_{\PP}\sq{\rhotilde_{\xn}} = \hat{\pi} \xi \hat{\pi}.
\end{align}
Using the Average Gentle Measurement Lemma \cite[Lemma 9.4.3]{wilde_arxivBook}, for any given $\epsilon \in (0,1)$, and all sufficiently large $n$ and all sufficiently small $\delta$, we have 
\begin{align} \label{eq:differencerhotildeandrhohat}
    \sum_{\xn \in \calX^n} \frac{\targetpx^n(\xn)}{1-\varepsilon}\norm{\rhotilde_{\xn} - \calW_{\xn}}_1 \leq \epsilon.
\end{align}
Detailed proof of the above statement can be found in \cite[Eq. 35]{wilde_e}.
Using the above definitions, for all $\xn \in \calX^n$, we construct the operators,
\[A_{\xn} \deq \gamma_{\xn}\ \sqrt{\sourcedo^{\tensor n}}^{-1} \rhotilde_{\xn} \sqrt{\sourcedo^{\tensor n}}^{-1},
\text{  where  }\gamma_{\xn} \deq \frac{1}{2^{nR}}  \frac{(1-\varepsilon)}{(1+\eta)}  \sum_{m=1}^{2^{nR}} \I_{\curly{\Xn(m) = \xn}},\]
and $\eta \in \round{0,1}$ is a parameter that determines the probability of not obtaining a sub-POVM.
Let $\I_{\curly{\mbox{sP}}}$ denote the indicator random variable corresponding to the event that  $\{A_{\xn} \colon \xn \in  \Txqc\}$ forms a  sub-POVM. 
If $\I_{\curly{\mbox{sP}}} = 1$, then construct sub-POVM $\Gamma^{(n)}$ as follows:
\begin{align}
   \Gamma^{(n)}& \deq \{A_{\xn} \colon \xn \in  \Txqc\}.
\label{eq:POVM_r1}
\end{align} 
Since $\Gamma^{(n)}$ is a sub-POVM, we add an extra operator $A_{\xn_0} \deq \round{I-\sum_{\xn\in \Txqc}A_{\xn}}$, associated with an arbitrary sequence $\xn_0 \in \calX^n \backslash\Txqc$, to form a valid POVM $\sq{\Gamma^{(n)}}$ with at most $(2^{nR}+1)$ elements.
If $\I_{\curly{\mbox{sP}}} = 0$, then we define $\Gamma^{(n)} = \curly{I}$ and associate it with $\xn_0$. This defines the POVM and the associated decoder. 
We now provide a proposition from \cite{winter}, which will be helpful later in the analysis. 
\begin{prop} \label{prop:qcsubpovm}For any $\epsilon, \eta \in (0,1)$, for any sufficiently small $\delta > 0$, and sufficiently large $n$, we have
$\EE\sq{\I_{\{\mbox{\normalfont sP}\}}} \geq 1-\epsilon$, if $R > I(X;\refstate)_{\sigma}$, where the quantum mutual information is computed with respect to the CQ state,
\[\sigma^{X\refstate} \deq \sum_x P_X(x) \ketbra{x}^{X}\tensor \calW_x,\]
 and $\{\ket{x}\}_{\curly{x\in \calX}}$ is an orthonormal basis 
    for the Hilbert space $\calH_X$ with $\dim{(\calH_X)}=|\calX|$.
\end{prop}

\subsubsection{Error Analysis} 
We show that for the above-mentioned POVM and decoder, the sum of unnormalized post-measurement reference state  $\sqrt{\sourcedo^{\tensor n}} A_{\xn}
\sqrt{\sourcedo^{\tensor n}}$ is close to the unnormalized n-product posterior reference state  $\Tr(A_{\xn} \sourcedo^{\tensor n} ) \bigotimes_{i=1}^{n} \calW_{x_i}$ 
in the trace distance, averaged over the random codebook. In other words, we would like to bound the following error term:
\[\EE[\Xi(\Gamma^{(n)})] = \EE\sq{\sum_{x^n} \norm{ \sqrt{\sourcedo^{\tensor n}} A_{x^n}
\sqrt{\sourcedo^{\tensor n}} -  \Tr(A_{x^n} \sourcedo^{\tensor n} ) \bigotimes_{i=1}^{n} \calW_{x_i}}_1}.\]
We begin by splitting the error $ \Xi(\Gamma^{(n)}) $ into two terms using the indicator function $\I_{\curly{\mbox{\normalfont sP}}}$ as 
\begin{align}
    \Xi(\Gamma^{(n)})  &= \I_{\curly{\mbox{\normalfont sP}}} 
\Xi(\Gamma^{(n)})  + \round{1- \I_{\curly{\mbox{\normalfont sP}}}} \Xi(\Gamma^{(n)}) ,\nonumber \\
&\leq \I_{\curly{\mbox{\normalfont sP}}} 
\Xi(\Gamma^{(n)})  + 2\round{1- \I_{\curly{\mbox{\normalfont sP}}}} \label{eqn:qcerrorsubpovm},
\end{align}
where \eqref{eqn:qcerrorsubpovm} follows from upper bounding the trace distance between two density operators by two, i.e., its maximum value.

\vspace{10pt}
\noindent\textbf{Step 1: Isolating the error term induced by not covering}

\noindent Using the triangle inequality, we now expand the $\Xi(\Gamma^{(n)})$ under the condition $\I_{\curly{\mbox{\normalfont sP}}} = 1$.

\begin{align}
    \Xi(\Gamma^{(n)})
&\leq \!\!\!\!\!\!\!\!\ \sum_{x^n \in \Txqc} \norm{ \sqrt{\sourcedo^{\tensor n}} A_{\xn}
\sqrt{\sourcedo^{\tensor n}} -  \Tr(A_{\xn} \sourcedo^{\tensor n} ) \calW_{\xn}}_1  \nonumber\\
&\hspace{1.5in}+ \norm{ \sqrt{\sourcedo^{\tensor n}} A_{\xn_0}
\sqrt{\sourcedo^{\tensor n}}}_1 \!\!\!+  \Tr(A_{\xn_0} \sourcedo^{\tensor n} ) \norm{\calW_{\xn_0}}_1, \nonumber \\
&= \!\!\!\!\!\!\!\!\ \sum_{x^n \in \Txqc} \norm{ \sqrt{\sourcedo^{\tensor n}} A_{\xn}
\sqrt{\sourcedo^{\tensor n}} -  \Tr(A_{\xn} \sourcedo^{\tensor n} ) \calW_{\xn}}_1 + 2  \Tr(A_{x_0^n} \sourcedo^{\tensor n})=\zeta + 2\Tilde{\zeta},
\end{align}
where we have defined:
\begin{align}
    \zeta &\deq  \sum_{\xn\in \Txqc}\norm{ {\sqrt{\sourcedo^{\tensor n}} A_{\xn}
\sqrt{\sourcedo^{\tensor n}}} -  {\Tr(A_{\xn} \sourcedo^{\tensor n} )}\bigotimes_{i=1}^{n} \calW_{x_i}}_1, \nonumber \\
\eqand \Tilde{\zeta}&\deq \Tr(A_{\xn_0} \sourcedo^{\tensor n}) =   \Tr{\Bigg(I- \!\!\!\!\! \sum_{x \in \Txqc} \!\!\! A_{\xn}\Bigg)\sourcedo^{\tensor n}}. \nonumber 
\end{align}
The error term $\Tilde{\zeta}$ captures the error induced by not covering the $n$-tensored posterior reference state. 
We provide the following proposition that bounds this term.
\begin{prop}\label{prop:qcerrornotcovering}
For all $\epsilon \in (0,1)$, and for all sufficiently small $\eta, \delta > 0$, and sufficiently large $n$, we have $\EE\sq{ \I_{\curly{\mbox{\normalfont sP}}}\Tilde{\zeta}} \leq \epsilon$.
\end{prop}
\begin{proof}
    The proof is provided in Appendix \ref{app:prop:proof:qcnotcoveringerror}.
\end{proof}
\vspace{10pt}
\noindent\textbf{Step 2: Bounding the error induced by covering}

\noindent We now bound the term $\zeta$, which captures the error induced by covering. 
Under the condition $\I_{\curly{\mbox{\normalfont sP}}} = 1$, we rewrite $\zeta$ as
\begin{align}
    \zeta &
= \sum_{\xn \in \Txqc} \gamma_{\xn} \Tr{\rhotilde_{\xn}}\norm{\frac{\rhotilde_{\xn}}{\Tr{\rhotilde_{\xn}}} - \calW_{\xn}}_1.
\end{align}
We now provide the following proposition that bounds the error term $\zeta$.
\begin{prop}\label{prop:qcerrorcovering}
For all $\epsilon, \eta \in (0,1)$, for all sufficiently small $\delta > 0$, and sufficiently large $n$, we have
    $\EE\sq{\I_{\curly{\mbox{\normalfont sP}}}\zeta} \leq \epsilon$.
\end{prop}
\begin{proof}
    The proof is provided in Appendix \ref{app:prop:proof:qcerrorcovering}.
\end{proof}

\noindent Finally, using Propositions \ref{prop:qcsubpovm}, \ref{prop:qcerrornotcovering}, and \ref{prop:qcerrorcovering}, we bound  $\EE\sq{\Xi(\Gamma^{(n)}}$, for all $\epsilon \in (0,1)$,  
\begin{align}
    \EE_{\codebook}\sq{\Xi(\Gamma^{(n)})} &\overset{}{\leq}\EE_{\codebook}\sq{\I_{\curly{\mbox{\normalfont sP}}} 
\Xi(\Gamma^{(n)}) + 2\round{1- \I_{\curly{\mbox{\normalfont sP}}}}} \overset{}{\leq}\EE_{\codebook}\sq{ \I_{\curly{\mbox{\normalfont sP}}} 
\Xi(\Gamma^{(n)})} + 2\epsilon \overset{}{\leq} 6 \epsilon. \nonumber
\end{align}
Since $\EE_{\codebook}\sq{\Xi(\Gamma^{(n)})} \leq 6\epsilon$, there exists a codebook $\codebook$ and the associated POVM $\Gamma^{(n)}$ 
such that $\Xi(\Gamma^{(n)}) \leq  6\epsilon$. 
This completes the achievability proof.

\subsection{Proof of Converse}
\label{sec:qc_converse}
Let $R$ be an achievable rate. Then from Definition \ref{def:qc_achievability}, given a triple $(\sourcedo,{\calX},\calW)$, for all $\epsilon > 0$, and all sufficiently large $n$, there exists $(n,\Theta)$ QC lossy compression protocol with a POVM $\Gamma^{(n)} = \curly{A_m}_{m\in [\Theta]}$ and a decoding map $f$ that satisfies the following constraint:
\[\sum_{x^n} \norm{ \sqrt{\sourcedo^{\tensor n}} A_{f^{-1}(x^n)}
\sqrt{\sourcedo^{\tensor n}} -  \Tr(A_{f^{-1}(x^n)} \sourcedo^{\tensor n}) \bigotimes_{i=1}^{n} \calW_{x_i}}_1 \leq \epsilon, \eqand \frac{1}{n}\log \Theta \leq R+\epsilon.\]
Let $M$ denote the transmitted message, and define the following classical-quantum state:
\begin{align}
\omega^{\Xn \refstate^n} &\deq \sum_{\xn} 
\ketbra{\xn}\tensor {\sqrt{\sourcedo^{\tensor n}} A_{f^{-1}(x^n)}
\sqrt{\sourcedo^{\tensor n}}}
\eqand \nonumber \\
\tau^{\Xn \refstate^n} &\deq \sum_{\xn} \Tr(A_{f^{-1}(x^n)} \sourcedo^{\tensor n}) \ketbra{\xn} \tensor \calW_{x^n},
\end{align}
where $\omega^{\Xn \refstate^n}$
and $\tau^{\Xn \refstate^n}$ are the resulting CQ-states of the QC lossy compression protocol and the ideal QC lossy compression protocol according to Definition \ref{def:qc_achievability}, respectively. By triangle equality, we have $\norm{\omega^{\Xn \refstate^n} - \tau^{\Xn \refstate^n}}_1 \leq \epsilon$.
We now provide a lower bound on the rate $R$. We have the following inequalities:
\begin{align}
    nR &= \log \Theta - n\epsilon \geq H(M) - n\epsilon  \overset{}{\geq} I(M;\refstate^n)_{\omega} 
    - n\epsilon \nonumber \\
    &\overset{a}{\geq} I(\Xn;\refstate^n)_{\omega} -n\epsilon, \nonumber \\
    &\overset{b}{\geq} nS(\refstate)_\omega - \sum_{i=1}^nS((\refstate)_i|X_i)_{\omega} -n\epsilon\nonumber\\
    &\overset{c}{\geq} nS(\refstate)_{\omega_Q} - nS(\refstate|X)_{\omega_Q} -n\epsilon = nI(X;\refstate)_{\omega_Q} -n\epsilon,\nonumber\\
    &\overset{d}{\geq} nI(X;\refstate)_{\tau_Q} - n\tilde{\epsilon}(\epsilon) -n\epsilon ,
\end{align}
where inequalities are argued as follows:
$(a)$ follows from the quantum data processing inequality \cite[Section 11.9.2]{wilde_arxivBook}, 
$(b)$ follows from the fact that conditioning does not increase quantum entropy,
$(c)$ follows from the concavity of conditional quantum entropy \cite[Ex. 11.7.5]{wilde_arxivBook} and by defining 
\[\omega^{X_Q (\refstate)_Q} \deq \frac{1}{n}\sum_{i=1}^n \Tr_{X^{n \backslash i}(\refstate)^{n\backslash i}}\curly{\omega^{\Xn \refstate^n}}\eqand \text{ noting that } \omega^{(\refstate)_Q} = \sourcedo,\] and 
$(d)$ follows from the continuity of quantum mutual information (AFW inequality) \cite[Ex. 11.10.2]{wilde_arxivBook},  by defining \[\tau^{X_Q (\refstate)_Q} \deq \frac{1}{n}\sum_{i=1}^n  \Tr_{X^{n \backslash i}(\refstate)^{n\backslash i}}\curly{\tau^{\Xn \refstate^n}} = \sum_{x} \round{\frac{1}{n}\sum_{i=1}^n \sum_{x^{n\backslash i}} \Tr{A_{f^{-1}(\xn)}\sourcedo^{\tensor n}}} \ketbra{x} \tensor \calW_{x}, \]
and
$ \Tilde{\epsilon} \deq \frac{3}{2}\epsilon \log(\dim \calH_B) + (2+\epsilon)h_{b}\Big(\frac{\epsilon}{2+\epsilon}\Big)$,
and noting
 \begin{equation}\label{eq:QCaverageSingleletter}
 \norm{\sourcedo - \Tr_{X_Q}\{\tau^{X_Q (\refstate)_Q}\}}_1 \leq \norm{\omega^{X_Q (\refstate)_Q} - \tau^{X_Q (\refstate)_Q}}_1 \leq \norm{\omega^{\Xn \refstate^n} - \tau^{\Xn \refstate^n}}_1 \leq \epsilon,
\end{equation}
where $\Tr_{X_Q}\{\tau^{X_Q (\refstate)_Q}\} =  \sum_x  P_{X_Q}(x)\calW_x$, and 
\[
P_{X_Q}(x) \deq \round{\frac{1}{n}\sum_{i=1}^n \sum_{x^{n\backslash i}} \Tr{A_{f^{-1}(\xn)}\sourcedo^{\tensor n}}}.
\]
We note that $\sum_x P_{X_Q}(x)=1$.
So far, we have shown that 
\[
R \in \bigcap_{\epsilon>0} \calI_{\epsilon},
\]
where we have defined for all $\epsilon \geq 0$, 
\[
\calI_\epsilon(\sourcedo,\calW) \deq \{R: \exists \; \px \in \calA_{\epsilon} \mbox{ such that } R\geq  I(X,\refstate)_{\sigma} -g(\epsilon)\},
\]
\[
\calA_\epsilon(\sourcedo,\calW) \deq \{\px \in \calP(\calX): \|\sum_{x} \px(x) \calW_x - \sourcedo\|_1 \leq \epsilon\}, \eqand \sigma^{X\refstate} \deq \sum_x P_X(x) \ketbra{x}^{X}\tensor \calW_x,\]
$g(\epsilon) \deq \tilde{\epsilon}+\epsilon$. Equation \eqref{eq:QCaverageSingleletter} ensures that the set $\calA_{\epsilon}$ is non-empty for $\epsilon>0$.
Using the continuity of rate regions similar to Lemma \ref{lem:intersectionSets}, we obtain 
$\bigcap_{\epsilon>0} \calI_{\epsilon}=\calI_{0},
$ and $\calA_0$ is non-empty, and hence 
$R \in \calI_0$.
This concludes the converse proof.

    

\section{Proof of Theorem \ref{thm:clsrate_distortion}}\label{sec:clsproof}
We begin the section with the achievability part, i.e., any rate $R$
that satisfies \eqref{eqn:clsratedistortion} is achievable. We then prove the converse, i.e., any achievable lossy source compression protocol must satisfy \eqref{eqn:clsratedistortion}.
\subsection{Proof of Achievability}\label{subsec:clsachievability}
For a given source distribution $\px$, reconstruction alphabet $\hat{\calX}$, and a posterior channel $\prevTC$, we choose a reconstruction distribution $\pxhat \in \calA(\px,\prevTC)$. Toward specifying the encoder $\encodern:\calX^n\longrightarrow[\Theta]$ and  the decoder $\decodern:[\Theta]\longrightarrow\hat{\calX}^n$, we construct a codebook $\codebook$. From now on, we let $\Theta = 2^{nR}+1$. 


\subsubsection{Codebook Construction} We construct a codebook $\codebook \deq \{\hat{X}^n(1),
\hat{X}^n(2),\cdots, \hat{X}^n(2^{nR})\}$, by choosing each codewords randomly and independently according to the following ``\textit{pruned}'' distribution:
\[\codeDistribution(\hat{X}^n(m) = \xhat^n) = \begin{cases}
     \dfrac{\pxhat^n(\hat{x}^n)}{1-\varepsilon}  & \text{if} \ \hat{x}^n\in \Txhat, \\
     0 &\text{otherwise.}
    \end{cases}\]
where $ \pxhat^n(\hat{x}^n) = \Pi_{i=1}^n \pxhat(\hat{x}_i)$, $\Txhat$ is the $\delta$-typical set corresponding to the distribution $\pxhat$ on the set $\hat{\calX}$, and $\varepsilon(\delta,n) \deq \sum_{\hat{x}^n \not \in \Txhat} \pxhat^n(\hat{x}^n)$. 
The codebook $\codebook$ is revealed to both the encoder and the decoder before the lossy source compression protocol begins.
\subsubsection{Encoder Description} 
For an observed source sequence $x^n$, construct a randomized encoder that chooses an index $m \in [2^{nR}]$ according to a sub-PMF $E_{M|X^n}(m|x^n)$\footnote{A non-negative function $q_X(x)$ over a finite alphabet $\calX$ is said to be a sub-PMF if $\sum_{x\in \calX} q_{X}(x) \leq 1$.}, which is analogous to the likelihood encoders used in  \cite{cuff2013distributed, atif2022source}. We now specify $E_{M|X^n}(m|x^n)$ for $x^n\in \Tx$ and $m\in[2^{nR}]$, where $\hat{\delta} = \delta(|\calX| + |\hat{\calX}|)$. For a  $\eta \in (0,1)$ (to be specified later), and $\delta>0$, define
\begin{align}\label{eqn:clssubpmfenc}
 E_{M|X^n}(m|x^n) \deq \sum_{\hat{x}^n}   \frac{1}{2^{nR}}\frac{(1-\varepsilon)}{(1+\eta)}\frac{\prevTC^n(x^n|\hat{x}^n)}{\px^n(x^n)} 
  \I_{\{\hat{x}^n\in \Txhat\}} 
  \I_{\{x^n\in \Txcond\}}
  \I_{\{\hat{X}^n(m) = \hat{x}^n\}}.
\end{align}
Similar to the encoder specification in \cite{atif2022source}, we also have relaxed the constraint that $E_{M|X^n}(\cdot|x^n)$ is strictly a PMF, i.e, $\sum_{m = 1}^{2^{nR}} E_{M|X^n}(m|x^n) = 1$. Let $\I_{\curly{\mbox{sPMF}}}$ denotes the indicator random variable corresponding to the event 
that $\{E_{M|\Xn}(m|\xn)\}_{m \in [\Theta]}$ forms a sub-PMF for all $\xn\in\Tx $.
If $\I_{\curly{\mbox{sPMF}}} = 1$, then construct the sub-PMF as follows:
\[P_{M|\Xn}(m|\xn) \deq {E_{M|\Xn}(m|\xn), \text{ for all }\xn\in\Tx \eqand m\in [\Theta]}.\]
We then add an additional PMF element $P_{M|\Xn}(0|\xn) = E_{M|\Xn}(0|\xn) \deq \round{1-\sum_{m = 1}^{2^{nR}} E_{M|X^n}(m|x^n)}$ for all $\xn \in \Tx$, associated with $m=0$, to form a valid PMF $P_{M|\Xn}(m|\xn)$ for all $\xn \in \Tx$ and $ m\in\curly{0} \cup [2^{nR}]$. If $\xn \not \in \Tx$, then we define $P_{M|\Xn}(m|\xn) = \I_{\curly{m=0}}$.
We provide a proposition that will be helpful later in the analysis.
\begin{prop}\label{prop:clssubPMF}
    For all $\epsilon,\eta \in (0,1)$, for all sufficiently small $\delta > 0$, and sufficiently large $n$, we have $\EE\sq{\I_{\curly{\mbox{\normalfont sPMF}}}} \geq 1-\epsilon$, 
    i.e.,
    \[\Pr\round{\bigcap_{\xn \in \Tx} \round{\sum_{m=1}^{2^{nR}} E_{M|\Xn}(m|\xn) \leq 1}} \geq 1-\epsilon,\]
if $R > I(X;\Xhat)$.
\end{prop}
\begin{proof}
    A proof is provided in Appendix \ref{app:proof:prop:clssubPMF}.
\end{proof}
\noindent We now summarize $P_{M|X^n}$ for $m\in \{0\}\cup [2^{nR}]$ and under the condition that $\I_{\curly{\mbox{\normalfont sPMF}}} = 1$,
\begin{equation}
    P_{M|X^n}(m|x^n) \deq \begin{cases}
    \I_{\{m=0\}} & \text{if } x^n \not \in \Tx,\\
     E_{M|X^n}(m|x^n) & \text{if } x^n \in \Tx.
    \end{cases}
\end{equation}
If $\I_{\curly{\mbox{\normalfont sPMF}}} = 0, \text{ then } P_{M|\Xn}(m|\xn) = \I_{\{m=0\}}$, for all $\xn \in \calX^n$. This concludes the encoder description.
\subsubsection{Decoder Description}
We now describe the decoder. For an observed index $m \in \{0\}\cup [2^{nR}]$ communicated by the encoder, the decoder outputs $\hat{X}^n(m)$ if $m \neq 0$. Otherwise, decoder outputs a fixed $\hat{x}_0^n \in \hat{\calX}^n \backslash \Txhat$, i.e.,
\begin{equation}
    \decodern(m) \deq 
    \begin{cases}
    \hat{X}^n(m) & \text{if } m \neq 0,\\
    \hat{x}^n_0  &\text{otherwise.}
    \end{cases}
\end{equation}

\subsubsection{Error Analysis}

We show that for the above-mentioned encoder and decoder, $P_{X^n\hat{X}^n}$ is close to the approximating distribution $P_{\hat{X}^n}\prevTC^n$ in the total variation, averaged over the random codebook. 
We begin by splitting the error $ \Xi(\encodern,\decodern)$ into two terms using the indicator function $\I_{\curly{\mbox{\normalfont sPMF}}}$ as 
\begin{align}
    \Xi(\encodern,\decodern) &= \I_{\curly{\mbox{\normalfont sPMF}}} 
\Xi(\encodern,\decodern) + \round{1- \I_{\curly{\mbox{\normalfont sPMF}}}} \Xi(\encodern,\decodern),\nonumber \\
&\leq \I_{\curly{\mbox{\normalfont sPMF}}} 
\Xi(\encodern,\decodern) + \round{1- \I_{\curly{\mbox{\normalfont sPMF}}}}. \label{eqn:clserrorsubpmf}
\end{align}

\vspace{10pt}
\noindent\textbf{Step 1: Isolating the error term induced by not covering}

\noindent Using the triangle inequality, we now expand the $\Xi(\encodern,\decodern)$ under the condition $\I_{\curly{\mbox{\normalfont sPMF}}} = 1$.
\begin{align}
    2 \; \Xi(&\encodern,\decodern) \nonumber \\ &= \sum_{\xn \xhat^n}   \Bigg | \px^n(\xn) \!\!\!\!\!\sum_{m \in \curly{0}\cup[2^{nR}]}   \!\!\!\!\! P_{M|X^n}(m|x^n) \I_{\{\hat{X}^n(m) = \hat{x}^n\}} 
    -\!\!\!\!\!\!\!\!\sum_{m \in \curly{0}\cup[2^{nR}]} \!\!\!\!\!  \!\!\!\! P_{M}(m)  
   \I_{\{\hat{X}^n(m) = \hat{x}^n\}}\prevTC^n(\xn|\hat{x}^n)\Bigg|  \nonumber\\
   &\overset{a}{\leq} \!\!\!\!\!\sum_{\substack{\xn  \in \Tx \\ \xhat^n}}   \Bigg | \px^n(\xn) \!\!\!\!\!\sum_{m \in [2^{nR}]}  \!\!\!\!\! E_{M|X^n}(m|x^n) \I_{\{\hat{X}^n(m) = \hat{x}^n\}}
    -\!\!\!\!\sum_{m \in [2^{nR}]} \!\!\!\!P_{M}(m) \I_{\{\hat{X}^n(m) = \hat{x}^n\}} 
    \prevTC^n(\xn|\hat{x}^n)\Bigg| \nonumber\\ 
   &\hspace{1in}+
   \sum_{\substack{\xn  \in \Tx}}\Bigg| \px^n(\xn)E_{M|\Xn}(0|\xn) - P_M(0) \prevTC^n(\xn|\xhat_0^n) \Bigg|
   \nonumber \\
   &\hspace{1in} +
   \sum_{\substack{\xn  \not \in \Tx \\ \hat{x}^n}}\Bigg| \px^n(\xn) \I_{\{\xhat^n = \hat{x}_0^n\}}  - \!\!\!\!\!\!\!\!\sum_{m\in \curly{0}\cup[2^{nR}]}P_M(m) 
 \prevTC^n(\xn|\xhat^n) \I_{\{\hat{X}^n(m) = \hat{x}^n\}}  \Bigg| 
   \nonumber \\
    &\overset{b}{\leq} \zeta + 
   \Tilde{\zeta} + \!\!\!\! \sum_{\xn\in \Tx} P_M(0) 
   \prevTC^n(\xn|\xhat_0^n)
   + \!\!\!\! \sum_{\xn \not \in \Tx} P_M(0) \prevTC^n(\xn|\xhat_0^n) + \!\!\!\! \sum_{\substack{\xn  \not \in \Tx}} \px^n(\xn) \nonumber \\
   &\hspace{1in}+ \sum_{ m\in [2^{nR}]}\sum_{\substack{\xhat^n \in \Txhat \\ \xn \not \in \Tx}} P_{M}(m) \prevTC^n(\xn|\hat{x}^n)\I_{\{\hat{X}^n(m) = \hat{x}^n\}}
   \nonumber \\
   &\overset{c}{\leq} \zeta + 
   \Tilde{\zeta} + \!\!\!\!
   \sum_{\substack{\xn  \not \in \Tx}} \px^n(\xn) + P_M(0)
   + \epsilon
   \nonumber \\
    &= \zeta + 
   2 \Tilde{\zeta} + 
   2 \!\!\!\! \sum_{\substack{\xn  \not \in \Tx}} \px^n(\xn) +\epsilon
   \overset{d}{\leq} \zeta + 2 \Tilde{\zeta} +
   3\epsilon,
\end{align}
for all sufficiently large $n$ and all $\delta>0$, 
where $(a)$ and $(b)$ follow from the triangle inequality, and by defining  
\begin{align}
   \zeta &\deq \!\!\!\!\!\sum_{\substack{\xn  \in \Tx \\ \xhat^n}}   \Bigg | \px^n(\xn) \!\!\!\sum_{m \in [2^{nR}]}  \!\!\!\!\! E_{M|X^n}(m|x^n) \I_{\{\hat{X}^n(m) = \hat{x}^n\}}
    -\!\!\sum_{m \in [2^{nR}]} \!\!\!\! P_{M}(m)  \I_{\{\hat{X}^n(m) = \hat{x}^n\}}
\prevTC^n(\xn|\hat{x}^n)\Bigg|, \nonumber\\
   \eqand \Tilde{\zeta} &\deq \sum_{\xn \in \Tx}
\px^n(\xn) E_{M|X^n}(0|x^n) =  \sum_{\xn \in \Tx} 
\px^n(\xn) \Bigg(1-\sum_{m = 1}^{2^{nR}} E_{M|X^n}(m|x^n)\Bigg),
\end{align}
$(c)$ follows from the conditional typicality argument for all sufficiently large $n$, and finally, $(d)$ follows from the standard typicality argument for all sufficiently large $n$.
The error term $\Tilde{\zeta}$ captures the error induced by not covering the $n$-product posterior test channel. 
We provide the following proposition that bounds this term.
\begin{prop}\label{prop:clserrornotcovering}
For all $\epsilon \in (0,1)$, and for all sufficiently small $\eta, \delta > 0$, and sufficiently large $n$, we have $\EE\sq{ \I_{\curly{\mbox{\normalfont sPMF}}}\Tilde{\zeta}} \leq \epsilon$
if $R > I(X;\Xhat)$.

\end{prop} 
\begin{proof}
    The proof is provided in Appendix \ref{app:proof:prop:clserrornotcovering}.
\end{proof}

\vspace{10pt}
\noindent\textbf{Step 2: Bounding the error induced by covering}

\noindent We now bound the term $\zeta$, which captures the error induced by covering. Using the triangle inequality, we get
\begin{align}
    \zeta 
     &= \sum_{\substack{ \xn \in \Tx \\ \xhat^n }} \Bigg|\;{\sum_{m \in [2^{nR}]} \round{  \px^n(x^n) E_{M|X^n}(m|x^n) -  P_{M}(m)\prevTC^n(x^n|\hat{x}^n)}\I_{\{\hat{X}^n(m) = \hat{x}^n\}}}\Bigg|, \nonumber\\
    &\overset{}{\leq}\!\! \sum_{ m\in [2^{nR}]}\sum_{\substack{ \xn \in \Tx \\ \xhat^n }}\absolute{\px^n(x^n) E_{M|X^n}(m|x^n) -  P_{M}(m)\prevTC^n(x^n|\hat{x}^n)} \I_{\{\hat{X}^n(m) = \hat{x}^n\}}.
    \nonumber 
\end{align}
We now provide the following proposition that bounds the error term $\zeta$.
\begin{prop}\label{prop:clserrorcovering}
    For all $\epsilon,\eta \in (0,1)$, for all sufficiently small $\delta > 0$, and sufficiently large $n$, we have $\EE\sq{ \I_{\curly{\mbox{\normalfont sPMF}}}{\zeta}} \leq \epsilon$.
\end{prop}
\begin{proof}
    A proof is provided in Appendix \ref{app:proof:prop:clserrorcovering}.
\end{proof}
\noindent Finally, using Propositions \ref{prop:clssubPMF}, \ref{prop:clserrornotcovering}, and \ref{prop:clserrorcovering}, we bound  $\EE\sq{\Xi(\encodern,\decodern)}$, for all $\epsilon \in (0,1)$,  
\begin{align}
    \EE_{\codebook}\sq{\Xi(\encodern,\decodern)} &\overset{}{\leq}\EE_{\codebook}\sq{\I_{\curly{\mbox{\normalfont sPMF}}} 
\Xi(\encodern,\decodern) + \round{1- \I_{\curly{\mbox{\normalfont sPMF}}}}}
\overset{}{\leq}9 \epsilon/2.\nonumber
\end{align}
Since $\EE_{\codebook}\sq{\Xi(\encodern,\decodern)} \leq 9\epsilon/2$, there exists a code $\codebook$ 
such that the associated $\Xi(\encodern,\decodern) \leq  9\epsilon/2$. 
This completes the achievability proof. 

\subsection{Proof of Converse}
\label{subsec:clsconverse}
Let $R$ be an achievable rate. Then  from Definition \ref{def:error_constraint}, given a triple $(\px,
\hat{\calX}, \prevTC)$, for all
$\epsilon>0$, and for all sufficiently large $n$, there exists $(n,\Theta)$ lossy compression protocol with an encoding map $\encodern$ and a decoding map $\decodern$ that satisfy the following constraints:
\begin{equation}
  \Xi(\encodern,\decodern)
  = \norm{P_{X^n\hat{X}^n} -
    P_{\hat{X}^n}
    \prevTC^n}_{\text{TV}}
    \leq \epsilon, \eqand \frac{1}{n}\log \Theta \leq R+\epsilon. 
    \nonumber
\end{equation}
Let $M$ denote the transmitted message. We now provide a lower bound on the rate $R$.  We have the following inequalities:
\begin{align}
    nR &= \log \Theta -n \epsilon \geq H(M) -n \epsilon 
    \geq I(X^n,M) -n \epsilon \nonumber \\
    &\overset{a}{\geq} I(X^n,\Xhat^n) -n \epsilon \nonumber \\
    &\overset{}{\geq} \sum_{i} H(X_i) - \sum_i H(X_i|\hat{X}_{i}) -n \epsilon \nonumber \\
    &= \sum_{i} I(X_i;\hat{X}_i) -n \epsilon \nonumber \\
    &\overset{b}{\geq} nI(X_Q;\hat{X}_Q) - n\epsilon  \nonumber\\
    &\overset{c}{=} nI(P_{\hat{X}_Q},P_{X_Q|\Xhat_Q}) - n\epsilon  \nonumber\\
    &\overset{d}{\geq} nI(P_{\hat{X}_Q},\prevTC) - n\Tilde{\epsilon}(\epsilon) -n\epsilon, 
    \nonumber
\end{align}
where the inequalities are argued as follows:
$(a)$ follows from the data processing inequality,
$(b)$ follows from the convexity of mutual information as the function of varying channel for a fixed source, and by defining $$P_{X_Q\hat{X}_Q} = \sum_i\frac{1}{n}P_{X_i\hat{X}_i} \quad \text{ and noting that } \quad P_{X_Q} = \px,$$
$(c)$ follows from the change of notation of mutual information \cite{csiszar2011information}, and $(d)$ follows from the continuity of mutual information \cite[Theorem 17.3.3]{cover2006elements} and from Lemma \ref{lem:averageSingleletter} (see below) and by defining $\tilde{\epsilon} \deq -2\epsilon\log \frac{4\epsilon^2}{|\calX|^2|\hat{\calX}|}$.
\begin{lemma} \label{lem:averageSingleletter}
The distributions $P_{X^n\hat{X}^n}$ and $P_{\hat{X}^n}\prevTC^n$ satisfy
\[\|P_X - \sum_{\xhat} P_{\hat{X}_Q}(\xhat)\prevTC(\cdot|\xhat) \|_{\normalfont \text{TV}} \leq \|P_{X_Q\hat{X}_Q} - P_{\hat{X}_Q}\prevTC \|_{\normalfont \text{TV}}  \leq 
\|P_{X^n\hat{X}^n} -
    P_{\hat{X}^n}\prevTC^n\|_{\normalfont \text{TV}}.\] 
\end{lemma}
\begin{proof}
The proof is provided in Appendix \ref{app:lem:proof:nlettergeqavg}.
\end{proof}


\noindent So far, we have shown that 
\[
R \in \bigcap_{\epsilon>0} \calI_{\epsilon},
\]
where we have defined for all $\epsilon \geq 0$, 
\[
\calI_\epsilon(\px,\prevTC) \deq \{R: \exists \pxhat \in \calA_{\epsilon}(P_X,W_{X|\hat{X}}) \mbox{ such that } R\geq  I(\pxhat,W_{X|\hat{X}}) -g(\epsilon)\},
\]
\[
\calA_\epsilon(\px,\prevTC) \deq \{\pxhat \in \calP(\hat{\calX}): \|\sum_{\hat{x}} \pxhat(\hat{x}) \prevTC(\cdot|\hat{x}) - \px\|_{\text{TV}} \leq \epsilon\},
\]
$g(\epsilon) \deq \tilde{\epsilon}+\epsilon$.
Lemma \ref{lem:averageSingleletter} ensures that the set $\calA_{\epsilon}$ is non-empty for $\epsilon>0$. Using the continuity of rate regions similar to Lemma \ref{lem:intersectionSets}, we obtain 
$\bigcap_{\epsilon>0} \calI_{\epsilon}=\calI_{0}
$ and $\calA_0$ is non-empty, and hence 
$R \in \calI_0$.
This concludes the converse proof.

\section{Conclusion}
\label{sec:conclusion}


In this work, we explored a new formulation of the lossy quantum source coding problem. The two ingredients that make our formulation different from the standard rate-distortion problem are (i) the usage of a global error criterion to measure the quality of reconstruction, and (ii) the notion of a posterior reference channel defined as a CPTP map acting on the reference of the reconstruction to produce the reference of the source. 
Instead of a single-letter distortion function, a global error criterion measures the error incurred by using the given single-letter posterior channel. The given channel characterizes the nature  of the loss incurred in the encoding and decoding operations. 

As a first main result, we provide a single-letter characterization of the asymptotic performance limit of this source coding problem using the minimal coherent information of the posterior reference map, where the minimization is over all reconstructions. 
Even though the formulation uses a global error criterion, it sheds light on an ``optimistic'' perspective of the lossy source coding theory.
In this regard, our results provide the missing duality pair of the quantum channel coding problem, and also
broadens the framework of performing lossy quantum source compression. 
Investigation of this formulation to other variants of lossy source coding problem can be an interesting research avenue to pursue. Similarly, it would be interesting to explore other techniques of establishing the achievability and converse of this limit.

Subsequently, we considered the quantum-classical (QC) setting and formulated a corresponding lossy QC source coding problem. We provided a single-letter characterization of the asymptotic performance limit of this problem using 
the minimal Holevo information (or the corresponding quantum mutual information) of the posterior classical-quantum (CQ) channel, where the minimization is over all reconstruction distributions (see Theorem \ref{thm:qclossysourcecoding}).
Finally, we performed a correspondingly new formulation for the classical setup, and established the minimal mutual information of the posterior channel as the single-letter characterization of the asymptotic performance limit of the classical source coding problem.

\appendix
\label{QLSC:sec:appdx}
\section{Proof of Lemmas}

\subsection{Proof of Lemma \ref{lem:closenessofPurification}}
\label{proof:lem:closenessofPurification}
Note that
    \begin{align*}
        \ket{\psi_{\rho}} &\deq (I_R\otimes \sqrt{\rho^B}) \ket{\Gamma}^{RB}, \quad
        \ket{\psi_{\sigma}} \deq (I_R\otimes \sqrt{\sigma^B}) \ket{\Gamma}^{RB},
    \end{align*}
    where $\ket{\Gamma}_{RB}$ is the unnormalized maximally entangled pure state:
    $\ket{\Gamma}_{RB} = \sum_{i} \ket{i}_R\ket{i}_B$.
    Consider the fidelity between the canonical purification states $\ket{\psi_{\rho}}$ and $\ket{\psi_{\sigma}}$:
    \begin{align*}
        F(\ket{\psi_{\rho}},\ket{\psi_{\sigma}}) \overset{a}{=}|\bra{\psi_{\rho}}\ket{\psi_{\sigma}}|^2
        &\overset{}{=} |\bra{\Gamma}^{RB}(I_R\otimes \sqrt{\rho^B}\sqrt{\sigma^B}) \ket{\Gamma}^{RB}|^2,\\
        &\overset{b}{=} |\Tr(\sqrt{\rho^B}\sqrt{\sigma^B})|^2
        \overset{c}{\geq} \left(1-\frac{1}{2}\norm{{\rho^B}- {\sigma^B}}_{1}\right)^2 \geq 1-\norm{{\rho^B}- {\sigma^B}}_{1},
    \end{align*}
    where $(a)$ follows from the definition of fidelity for a pure state, $(b)$ follows from the definition of trace, 
    $(c)$ follows from 
    the Power-Størmer inequality \cite[Lemma 4.1]{powers1970free}, i.e., for any positive semi-definite matrices $A$ and $B$, we have
    \[\Tr(A)+\Tr(B) - \norm{A-B}_1 \leq 2 \Tr(\sqrt{A}\sqrt{B}).\]

\subsection{Proof of Lemma \ref{lem:RotationUnitary}}
\label{proof:lem:RotationUnitary}
We first provide the following lemma.
\begin{lemma}[Covering superposition states] \label{lem:coveringSuperposition}
    Consider a finite set $\calU$, and a 
    pair of collections $\{\rho_u\}_{u \in \calU}$ and $\{\sigma_u\}_{u \in \calU}$ where  $\rho_u,\sigma_u \in \calD(\calH_A)$ for all $u \in \calU$.
    Let $\{\Psi_u^\rho\}_{u \in \calU}$ and $ \{\Psi_u^\sigma\}_{u \in \calU} $ acting on $\calD(\calH_{R_1}\tensor\calH_A)$ and $\calD(\calH_{R_2}\tensor\calH_A)$ be some purifications of $\{\rho_u\}_{u \in \calU}$ and $\{\sigma_u\}_{u \in \calU}$, respectively, with $\dim(\calH_{R_1}) \leq \dim(\calH_{R_2})$.
    Then there exists a collection of isometric operators $\{U_r(u)\}_{u\in\calU}$ acting on $\calH_{R_1} \rightarrow \calH_{R_2}$ and phases $\{\delta_u\}$ such that 
    \begin{align}
         F((U_R\tensor I_A)\ket{\tau_\rho}, \ket{\tau_\sigma})  =  F(\ket{\tau_\rho}, (U_R\tensor I_A)^\dagger\ket{\tau_\sigma}) \geq 1-\sum_{u\in \calU} \frac{1}{|\calU|}\|\rho_u - \sigma_u\|_1,
    \end{align}
    where
    \begin{align}
        U_R \deq \sum_{u\in \calU} e^{-i\delta_u} U_r(u) \tensor \ketbra{u}, \quad \ket{\tau_\rho} &\deq \sum_{u\in\calU} \frac{1}{\sqrt{|\calU|}}\ket{\psi_u^\rho}\tensor \ket{u}, \quad
        \ket{\tau_\sigma} \deq \sum_{u\in\calU} \frac{1}{\sqrt{|\calU|}}\ket{\psi^\sigma_u}\tensor \ket{u}.\nonumber
    \end{align}
\end{lemma}
\begin{proof}
    We provide a proof in Appendix \ref{proof:lem:coveringSuperposition}.
\end{proof}

Now, with the intention of employing the above lemma
we perform the following identification. Identify $ \calU$ with $\calM\times\calK$, $\rho_u$ with $\hat{\rho}_{m,k}^{\bB}$,  $\sigma_u$ with $\tilde{\rho}_{m,k}^{\bB}$, $\ket{\psi_u^\rho}$ with $\frac{(I\tensor M_{m,k})}{\sqrt{\lambda_{m,k}}}|{\psi_{\rho}^{\tensor n}}\rangle^{\bBn B^n}$, and $\ket{\psi_u^\sigma}$ with $\frac{(I\tensor \sqrt{A_{m,k}})}{\sqrt{\delta_{m,k}}}|{\psi_{\rho}^{\tensor n}}\rangle^{\bBn B^n}$. Note that the last two identifications are, in fact, the purifications of $\hat{\rho}_{m,k}^{\bB}$ and $\tilde{\rho}_{m,k}^{\bB}/\Tr (\tilde{\rho}_{m,k}^{\bB})$, respectively as
\begin{align*}
    \Tr_{E}\left(\frac{(I\tensor M_{m,k})}{\sqrt{\lambda_{m,k}}}{\Psi_{\rho_B}^{\tensor n}}\frac{(I\tensor M^\dagger_{m,k})}{\sqrt{\lambda_{m,k}}}\right) = \hat{\rho}_{m,k}^{\bB} , \;\; \Tr_{B}\left(\frac{(I\tensor \sqrt{A_{m,k}})}{\sqrt{\delta_{m,k}}}{\Psi_{\rho_B}^{\tensor n}}\frac{(I\tensor \sqrt{A_{m,k}})}{\sqrt{\delta_{m,k}}}\right) = \frac{\tilde{\rho}_{m,k}^{\bB}}{\Tr(\tilde{\rho}_{m,k}^{\bB})}.
\end{align*}
Using Lemma \ref{lem:coveringSuperposition}, we obtain
\begin{align}
    F&(|{\hat{\sigma}}\rangle^{{\bB}EMK}, (I_{\bB} \tensor U_{\calR} ) |{\tilde{\sigma}}\rangle^{\bB BMK}) \nonumber \\
    & \geq 1 - \frac{1}{(1-\sqrt{\epsilon})|\calM||\calK|}\sum_{m,k}\bigg\|\hat{\rho}_{m,k}^{\bB}-\frac{\tilde{\rho}_{m,k}^{\bB}}{\Tr(\tilde{\rho}_{m,k}^{\bB})}\bigg\|_1 \nonumber\\
    & \geq \frac{1}{(1-\sqrt{\epsilon})|\calM||\calK|}\sum_{m,k}{\Tr{\tilde{\rho}_{m,k}^{\bB}}}  - \frac{1}{(1-\sqrt{\epsilon})|\calM||\calK|}\sum_{m,k}\bigg\|\hat{\rho}_{m,k}^{\bB}-{\tilde{\rho}_{m,k}^{\bB}}\bigg\|_1 \geq 1-4\epsFourRoot,
\end{align}
where the last inequality follows from using the bounds in \eqref{eq:E1_prime} and \eqref{eq:E2_prime}.
\subsection{Proof of Lemma \ref{lem:coveringSuperposition}}
\label{proof:lem:coveringSuperposition}
    Consider the following:
    \begin{align*}
        F(\ket{\tau_\rho}, (U_R\tensor I)^{\dagger}\ket{\tau_\sigma}) 
        &{=}
        \left|\sum_{u\in \calU} \frac{1}{|\calU|} e^{-i\delta_u}\bra{\psi_u^\rho}U_r^{\dagger}(u)\ket{\psi_u^\sigma}\right|^2\\
        &\overset{a}{=}
        \left(\sum_{u\in \calU} \frac{1}{|\calU|} \left|\bra{\psi_u^\rho}U_r^\dagger(u)\ket{\psi_u^\sigma}\right|\right)^2\\
        &\overset{b}{=}
        \left(\sum_{u\in \calU} \frac{1}{|\calU|} \sqrt{F(\rho_u,\sigma_u)}\right)^2\\
        &\overset{c}{\geq}
        \left(1-\frac{1}{2}\sum_{u \in \calU} \frac{1}{|\calU|} \norm{\rho_u -\sigma_u}_1\right)^2
        \geq
        1-\sum_{u \in \calU} \frac{1}{|\calU|} \norm{\rho_u -\sigma_u}_1,
    \end{align*}
    where 
    $(a)$ follows by choosing $\delta_u$ such that $e^{-i\delta_u} \bra{\psi_u^\rho}U_r^\dagger(u)\ket{\psi_u^\sigma} = |\bra{\psi_u^\rho}U_r^{\dagger}(u)\ket{\psi_u^\sigma}|$, $(b)$ follows from Uhlmann's theorem \cite[Theorem 9.2.1]{wilde_arxivBook}, i.e., there exists  some isometry $U_r(u)$ such that $F(\rho_u,\sigma_u) = F(U_r(u)\ket{\psi^{\rho}},\ket{\psi^{\sigma}})$, and $(c)$ follows from Lemma \ref{lem:relationshipTraceFidelity}.

\subsection{Proof of Proposition \ref{prop:ExpectedFourier}} \label{proof:prop:ExpectedFourier}

We begin by defining indexing functions $f^{(m)}:[0,K'_m-1] \rightarrow \calI_\calE^{(m)}$, for each $m\in \calM'$, that uniquely map each element of  the $[0,K'_m-1]$ to the set $\calI_\calE^{(m)}$ in a monotonic fashion. Let $g^{(m)}: \calI_\calE^{(m)}\rightarrow [0,K'_m-1]$ be the inverse of $f^{(m)}$, for each $m\in\calM'$.
Define the transformed vectors corresponding to the collections $\curlys{\chi_k^{(m)}}$ and  $\curlys{\phi_k^{(m)}}$ as
\begin{align}
    |{\hat{\chi}_{s}^{(m)}}\rangle \deq c\sum_{j = 0}^{K_m'-1}e^{\frac{2\pi i js}{K'_m}}|\chi_{f^{(m)}(j)}^{(m)}\rangle \qand     |{\hat{\phi}_s^{(m)}}\rangle \deq c\sum_{j = 0}^{K_m'-1}e^{\frac{2\pi i js}{K'_m}}|\phi_{f^{(m)}(j)}^{(m)}\rangle, \nonumber
\end{align}
for $s \in [0,K_m'-1]$.
It follows from basic algebra that, for all $m \in \calM'$,
\begin{align} 
    \frac{1}{K'_m}\sum_{s= 0}^{K_m'-1} \langle{\hat{\phi}_s^{(m)}}|{\hat{\chi}_s^{(m)}}\rangle  = c^2\sum_{j = 0}^{K_m'-1}\langle\phi_{f^{(m)}(j)}^{(m)}|\chi_{f^{(m)}(j)}^{(m)}\rangle = 
    c^2\sum_{ k\in \calI_\calE^{(m)}} \langle{{\phi}_k^{(m)}}|{{\chi}_k^{(m)}}\rangle.
\end{align}
This implies, for all $m\in \calM'$, there exists at least one value of $s_m \in [0,K_m'-1]$ that follows the inequality: 
\[
e^{i\hat\theta_m} \langle{\hat{\phi}_{s_m}^{(m)}}|{\hat{\chi}_{s_m}^{(m)}}\rangle \geq c^2\sum_{ k\in \calI_\calE^{(m)}} \langle{{\phi}_k^{(m)}}|{{\chi}_k^{(m)}}\rangle, \quad \mbox{for some phase } \hat{\theta}_m.
\]
Observe that,
\begin{align*}
     \langle{\hat{\phi}_{s_m}^{(m)}}|{\hat{\chi}_{s_m}^{(m)}}\rangle = c^2 \sum_{ k\in \calI_\calE^{(m)}} \sum_{ k'\in \calI_\calE^{(m)}} e^{\frac{2\pi i (g^{(m)}(k)-g^{(m)}(k'))s_m}{K_m'}}\langle {\phi}_{k'}^{(m)}| {\chi}_k^{(m)}\rangle,
\end{align*}
for all $m \in \calM'$.
Choosing $\alpha_k^{(m)} = \frac{2\pi g^{(m)}(k) s_m}{K_m'}$ and $\beta_k^{(m)} = \frac{2 \pi g^{(m)}(k) s_m}{K_m'} + \hat\theta_m$, we obtain
\begin{align}
    \frac{1}{M'}\sum_{m\in \calM'} \langle\phi_m|\chi_m\rangle \geq \frac{c^2}{M'}\sum_{m\in \calM'}\sum_{ k\in \calI_\calE^{(m)}} \langle{{\phi}_k^{(m)}}|{{\chi}_k^{(m)}}\rangle = \frac{c^2}{M'}\sum_{m\in \calM'}\sum_{ k\in \calI_\calE^{(m)}} \Tr{\Xi^{(m)}_k \tau_{k}^{(m)}} 
    \geq 1-2\epsFourRoot,\nonumber
\end{align}  
where the equality uses \eqref{eq:avg_packing_error}, and the last inequality uses \eqref{eq:E3_prime} and substitutes the value of $c$.
This means  
\[
 \frac{1}{M'} \mbox{Re} \round{\sum_{m\in \calM'} \langle\phi_m|\chi_m\rangle}\geq 1-2\epsFourRoot,
\]
and consequently,
\begin{align}
    \bigg| \frac{1}{M'}\sum_{m\in \calM'} \langle\phi_m|\chi_m\rangle \bigg| & 
     \geq \bigg[\frac{1}{M'}\mbox{Re}\left( \sum_{m\in \calM'}\langle\phi_m|\chi_m\rangle\right) \bigg]  
     \geq 1-2\epsFourRoot.
\end{align}
This completes the proof.


\subsection{Proof of Lemma \ref{lem:intersectionSets}} \label{proof:lem:intersectionSets}
Here we follow arguments similar to the proof of \cite[Lemma VI.5]{cuff2013distributed}. We begin by defining $\calI'_{\epsilon}$ (removing the relaxation in the rate) as, for all $\epsilon \geq 0$,
\begin{align}
    \calI'_\epsilon \deq  \curly{R: \exists \; \rho^{{A}_R} \in \calS_\epsilon(\rho^{B},\calN_W) \mbox{ such that } R \geq I_c(\calN_W,\rho^{{A}_R})},
\end{align}
and note from \cite{cuff2013distributed} that
\[\bigcap_{\epsilon > 0} \calI_\epsilon \subseteq \mbox{Closure} \left(\bigcap_{\epsilon > 0} \calI'_\epsilon\right). \]
Now we prove the following: 
\begin{align}
     \calS_0(\rho^{B},\calN_W) = \bigcap_{\epsilon > 0} \calS_\epsilon(\rho^{B},\calN_W). \label{eq:defI_equal}
\end{align}
$\calS_0(\rho^{B},\calN_W) \subseteq \bigcap_{\epsilon > 0} \calS_\epsilon(\rho^{B},\calN_W)$ is straightforward. To show the other direction, consider any $\rho_1^{A_R} \in \bigcap_{\epsilon > 0} \calS_\epsilon(\rho^{B},\calN_W)$. This means, for all $\epsilon > 0$,
    \begin{align}
        \| \mathcal{N}_W(\rho_1^{\bA})-\rho^{\bB} \|_1 \leq \epsilon \implies \| \mathcal{N}_W(\rho_1^{\bA})-\rho^{\bB} \|_1 = 0 \implies \mathcal{N}_W(\rho_1^{\bA}) = \rho^{\bB},
    \end{align}
    where the second implication follows from the definition of a metric, and hence $\rho_1^{\bA} \in \calS_0(\rho^{B},\calN_W)$ and \eqref{eq:defI_equal} is true. Observe that, since the intersection of  decreasing sequence of non-empty closed and bounded sets of a compact (finite-dimensional) metric space is non-empty, $\calS_0(\rho^{B},\calN_W)$ is non-empty.
    Therefore, using the continuity of $f(\rho^{\bA}) = I_c(\calN_W,\rho^{\bA})$, and the fact that $\calS_\epsilon$ are decreasing non-empty closed and bounded subsets of a compact (finite-dimensional) metric space gives
    \[ f(\calS_0(\rho^{B},\calN_W)) = \bigcap_{\epsilon>0}f(\calS_\epsilon(\rho^{B},\calN_W)).\]
    Noting that the images $f(\calS_\epsilon(\rho^{B},\calN_W))$ and  $f(\calS_0(\rho^{B},\calN_W))$ characterize the rate regions $\calI'_\epsilon$ and $\calI_0$, respectively, and the fact that $\calI_0$ is closed 
    completes the proof.

\subsection{Proof of Proposition \ref{prop:qcerrornotcovering}}
\label{app:prop:proof:qcnotcoveringerror}
Consider the following inequalities:
    \begin{align}
    \EE_{}\sq{\I_{\curly{\mbox{\normalfont sP}}}\Tilde{\zeta}} &\leq \EE_{}\sq{\Tr{\round{I-\sum_{\xn \in\Txqc}A_{\xn}}\sourcedo^{\tensor n}}} = 1 - \EE_{}\sq{ \ \sum_{\xn \in \Txqc} \Tr{A_{\xn}\sourcedo^{\tensor n}}},\nonumber \\
      &= 1 -   \frac{(1-\varepsilon)}{(1+\eta)} \sum_{\xn \in \Txqc} \frac{\targetpx^n(\xn)}{(1-\varepsilon)}\Tr{\rhotilde_{\xn}},\nonumber \\
      &\overset{a}{=} 1 -   \frac{(1-\varepsilon)}{(1+\eta)} \Tr{\rhotilde} \overset{b}{\leq} 1-\frac{(1-\varepsilon)}{(1+\eta)}(1-2\varepsilon -2\sqrt{\varepsilon}) <\epsilon, \nonumber 
\end{align}
for all sufficiently large $n$ and all sufficiently small $\eta,\delta >0 $, where (a) follows from \eqref{eqn:rhotilde} and (b) follows from \cite[Eq. 28]{wilde_e}. 
This completes the proof of Proposition \ref{prop:qcerrornotcovering}.

\subsection{Proof of Proposition \ref{prop:qcerrorcovering}}
\label{app:prop:proof:qcerrorcovering}
We begin the proof by splitting the term $\zeta$ under the condition $\I_{\curly{\mbox{\normalfont sP}}} = 1$; using triangle inequality, we get $\zeta \leq \zeta_1+\zeta_2$, where 
\begin{align}\label{eqn:qcerrorevent}
    \zeta_1 \deq  \sum_{\xn \in \Txqc}\gamma_{\xn} \Tr{\rhotilde_{\xn}}
\norm{\frac{\rhotilde_{\xn}}{\Tr{\rhotilde_{\xn}}} - \rhotilde_{\xn}}_1, \;  
\zeta_2 \deq \sum_{\xn \in \Txqc} \gamma_{\xn} \Tr{\rhotilde_{\xn}}
\norm{\rhotilde_{\xn} - \calW_{\xn}}_1.
\end{align}
Consider the following inequalities:
    \begin{align}
        \EE_{}\sq{\I_{\curly{\mbox{\normalfont sP}}}\zeta_1} 
&\leq \sum_{m=1}^{2^{nR}} \sum_{\xn \in \Txqc}  \frac{1}{2^{nR}} \frac{(1-\varepsilon)}{(1+\eta)} \EE_{}\sq{\I_{\curly{\Xn(m) = \xn}}} \Tr{\rhotilde_{\xn}}  
\norm{\frac{\rhotilde_{\xn}}{\Tr{\rhotilde_{\xn}}} - \rhotilde_{\xn}}_1\nonumber\\
&\leq \sum_{m=1}^{2^{nR}} \sum_{\xn \in \Txqc} \frac{1}{2^{nR}} \frac{(1-\varepsilon)}{(1+\eta)} \frac{P_{X}^n(\xn)}{(1-\varepsilon)}\round{1-\Tr{\rhotilde_{\xn}}}
\nonumber\\
&= \frac{(1-\varepsilon)}{(1+\eta)} \sq{1 - \sum_{\xn \in \Txqc} \frac{P_{X}^n(\xn)}{(1-\varepsilon)}\Tr{\rhotilde_{\xn}}}\nonumber \\
&\overset{a}{=} \frac{(1-\varepsilon)}{(1+\eta)} \sq{1-\Tr{\rhotilde}} \overset{b}{\leq} \frac{(1-\varepsilon)}{(1+\eta)} (2\varepsilon+2\sqrt{\varepsilon})\leq  \epsilon, \nonumber
    \end{align}
for all sufficiently large $n$ and all sufficiently small $\eta,\delta >0 $, where $(a)$ follows from definition \eqref{eqn:rhotilde}, and $(b)$  follows from \cite[Eq. 28]{wilde_e}. 
Similarly, we now compute
    \begin{align}   \EE_{}\sq{\I_{\curly{\mbox{\normalfont sP}}}\zeta_2} 
&\overset{}{\leq} \sum_{m=1}^{2^{nR}} \sum_{\xn \in \Txqc} \frac{(1-\varepsilon)}{(1+\eta)} \frac{1}{2^{nR}} \frac{P_{X}^n(\xn)}{(1-\varepsilon)} 
\norm{\rhotilde_{\xn} - \calW^{\tensor n}_{\xn}}_1 \nonumber \\
&\overset{}{=}  \frac{(1-\varepsilon)}{(1+\eta)} \sum_{\xn \in \Txqc}  \frac{P_{X}^n(\xn)}{(1-\varepsilon)} 
\norm{\rhotilde_{\xn} - \calW^{\tensor n}_{\xn}}_1 \leq \epsilon \nonumber.
\end{align}
for all sufficiently large $n$ and all sufficiently small $\eta,\delta >0 $, where the last inequality follows from \eqref{eq:differencerhotildeandrhohat}. 
This completes the proof of Proposition \ref{prop:qcerrorcovering}.

\subsection{Proof of Proposition \ref{prop:clssubPMF}}\label{app:proof:prop:clssubPMF}
Recall, the definition of $E_{M|\Xn}(m|x^n)$,\text{ for } $\xn \in \Tx$:
\begin{align*}
   E_{M|X^n}(m|x^n) =  &\sum_{\hat{x}^n \in \Txhat}  \frac{1}{2^{nR}}\frac{(1-\varepsilon)}{(1+\eta)}\frac{\prevTC^n(x^n|\hat{x}^n)}{\px^n(x^n)} 
  \I_{\{x^n\in \Txcond\}}
  \I_{\{\hat{X}^n(m) = \hat{x}^n\}}.
\end{align*}

\noindent Let $D= 2^{n(H(X|\Xhat)-\delta_1)}$, where  $\delta_1(\delta) \searrow 0$ as $\delta \searrow 0$,
and will be specified in the sequel. Define a sequence of $2^{nR}$ IID random variables 
$\curly{Z_m(x^n)}_{m=1}^{2^{nR}}$, for all $\xn \in \Tx$,
\begin{align}\label{def:ZtZ}
    Z_m(x^n) & \deq \sum_{\xhat^n \in 
\Txhat} (1-\varepsilon) \prevTC^n(x^n|\hat{x}^n)  
\I_{\{x^n\in \Txcond\}} 
\I_{\{\hat{X}^n(m) = \hat{x}^n\}}.
\end{align}

\noindent We get the following bound on the expectation of the empirical average of $\{Z_m(x^n)\}_{ \in [2^{n{R}}]}$, for all sufficiently large $n$:
\begin{align}
    \EE\bigg[\frac{1}{2^{nR}}\sum_{m=1}^{2^{nR}} D Z_m(x^n) \bigg]
    &= \frac{D}{2^{nR}}\sum_{m=1}^{2^{nR}}  \sum_{\xhat^n \in 
\Txhat} (1-\varepsilon) \prevTC^n(x^n|\hat{x}^n)  
    \I_{\{x^n\in \Txcond\}} 
    \EE[\I_{\{\hat{X}^n(m) = \hat{x}^n\}}], \nonumber \\
    &\overset{}{=}  \frac{D}{2^{nR}}\sum_{m=1}^{2^{nR}}  \sum_{\xhat^n \in 
\Txhat} P^n_{\Xhat}(\xhat^n) \prevTC^n(x^n|\hat{x}^n)  
    \I_{\{x^n\in \Txcond\}},   
    \nonumber \\
    &\overset{a}{=} D  \sum_{\xhat^n \in \Txhat} P^n_{\Xhat}(\xhat^n) \prevTC^n(x^n|\hat{x}^n)  
    \I_{\{(x^n,\xhat^n) \in \Txxhat\}}
    \I_{\{x^n\in \Txcond\}}, 
    \nonumber \\
    &\overset{b}{\geq} 2^{n(H(X|\Xhat)-\delta_1)} 2^{-n(H(X,\Xhat) + 2\delta_1)} 2^{n(H(\Xhat|X)-\delta_1)} 
    = 2^{-n(I(X;\Xhat)+4\delta_1)}, \label{eq:app:boundexpecationZ}
\end{align}
 where in 
$(a)$ follows from the fact that if $x^n \in \Txcond$ and 
$\xhat^n \in \Txhat$, then 
$(x^n,\xhat^n) \in \Txxhat$, where $\hat{\delta} = \delta(|\calX|+|\hat{\calX}|)$, and $(b)$ follows from the properties of joint typical and conditional typical sequences and $\delta_1(\delta)$ is a function that follows from the characterization of the size of the typical set \cite{cover2006elements}.

\noindent Furthermore, observe that, for all sufficiently large $n$, we have
\begin{align} \label{eq:app:boundonZ}
    D Z_m(x^n) & \leq  2^{n(H(X|\Xhat)-\delta_1)} 2^{-n(H(X|\Xhat)-\delta_1)}(1-\varepsilon) \Bigg(\; \sum_{\xhat^n \in \Txhat}
    \I_{\{\hat{X}^n(m) = \hat{x}^n\}}\Bigg)  \leq   1, 
\end{align}
where the first inequality follows from the properties of joint typical sequences, i.e., if $(\xn,\xhat^n) \in \Txxhat$, then
$\prevTC(\xn|\xhat^n) \leq 2^{-n(H(X|\Xhat)-\delta_1)}$.
From \eqref{eq:app:boundexpecationZ} and \eqref{eq:app:boundonZ}, observe that $\{DZ_m(x^n)\}_{m}$ satisfies the  constraints of Lemma \ref{lem:chernoff} (stated below).
Thus, after applying Lemma (\ref{lem:chernoff}) to $\{DZ_m(x^n)\}_{m}$ for all $\eta \in (0,1)$ and $x^n \in \Tx$, we get
\begin{align}
    \!\!\Pr\left(  Z(x^n) \in 
 \Big[(1-\eta)\EE[ Z(x^n)],(1+\eta)\EE[ Z(x^n)]\Big]\right) \geq 1 \!\! - \!2\exp{\!\!-\frac{\eta^2 2^{n(R - I(X;\Xhat)-4\delta_1)} }{4 }\!\!}, \label{eq:app:chernoofZavg}
\end{align}
where $Z(x^n) \deq \frac{1}{2^{nR}}\sum_{m=1}^{2^{n{R}}} Z_m(x^n) $.
Moreover, using the definition of $Z_{m}(\xn)$ and $E_{M|\Xn}(m|\xn)$, we can simplify the above inequality as follows:
\begin{align}
    \Pr\left(\!\!(1+\eta)P_{X}^n(x^n) \sum_{m=1}^{2^{nR}}  E_{M|X^n}(m|x^n) \!\leq \!(1+\eta)\EE{[Z(x^n)]}\right) \!\! \geq 1 \!\! - \!2\exp{\!\!-\frac{\eta^2 2^{n(R - I(X;\Xhat)-4\delta_1)} }{4 }\!\!}. \label{eqn:app:Zchernoof} 
\end{align}
Now, observe the following bound on  $\EE[Z(x^n)]$
\begin{align}
    \frac{1}{P^n_X(x^n)} \EE[Z(x^n)] &\leq \frac{1}{P^n_X(x^n)} \sum_{\xhat^n} P^n_{\Xhat}(\xhat^n) \prevTC^n(x^n|\hat{x}^n)
    = 1. \nonumber
\end{align}
This simplifies the inequality $\eqref{eqn:app:Zchernoof}$ as, for all $\xn \in \Tx$,
\begin{align}
    \PP\left(\sum_{m =1}^{2^{nR}}  E_{M|X^n}(m|x^n) \leq 1\right)  \geq 1 - 2\exp{\bigg(-\frac{\eta^2 2^{n({R} - I(X,\Xhat)-4\delta_1)} }{4 }\bigg)}. \nonumber
\end{align}
Eventually, using the union bound, we get
\begin{align}
\Pr\left[\; \bigcap_{x^n\in \Tx}\round{\sum_{m =1}^{2^{n{R}}}  E_{M|X^n}(m|x^n) \leq 1}\right] &\geq 1-\!\!\!\sum_{x^n \in \Tx} \Pr\left(\sum_{m =1}^{2^{n{R}}}  E_{M|X^n}(m|x^n) > 1\right),\nonumber \\&
\geq 1 - 2|\Tx|\exp{\!\!-\frac{\eta^2 2^{n(R - I(X;\Xhat)-4\delta_1)} }{4 }\!\!} 
\label{eq:unionProb}.
\end{align}
Thus, if ${R} > I(X;\Xhat) + 4\delta_1$, the second term in the right hand side of (\ref{eq:unionProb}) decays exponentially to zero, and as a result, the probability of the above intersections goes to 1. 
This completes the proof of Proposition \ref{prop:clssubPMF}.

\begin{lemma} \label{lem:chernoff}
Let $\{Z_n\}_{n=1}^N$ be a sequence of N IID random variables bounded between zero and one, i.e., $ Z_n \in [0,1] \quad \forall n \in [N]$, and suppose $\EE\bigg[\frac{1}{N}\sum_{n=1}^{N}Z_n\bigg] = \mu$ be bounded below by a positive constant $\theta$ as $ \mu \geq \theta$ where $\theta \in (0,1)$, then for every $\eta \in (0,1/2)$ and $(1+\eta)\theta < 1$, we can bound the probability that the ensemble average of the sequence $\{Z_n\}_{n=1}^{N}$ lies in $(1\pm \eta)\mu$ as
\begin{align}
    \PP\bigg(\frac{1}{N}\sum_{n=1}^{N}Z_n \in & [(1-\eta)\mu,(1+\eta)\mu] \bigg)  \geq 1 - 2\exp{\bigg(-\frac{N\eta^2 \theta}{4}\bigg)}.
\end{align}
\end{lemma}
\begin{proof}
Follows from the Operator Chernoff Bound \cite{ahlswede2002strong}.
\end{proof}
\subsection{Proof of Proposition \ref{prop:clserrornotcovering}}\label{app:proof:prop:clserrornotcovering}
Fix an $\epsilon>0$.
Recall, the definition of $Z_m(\xn)$ for $\xn \in \Tx$ from Appendix \ref{app:proof:prop:clssubPMF},
\begin{align}
    Z_m(x^n) & = \sum_{\xhat^n \in 
\Txhat} (1-\varepsilon) \prevTC^n(x^n|\hat{x}^n)  
\I_{\{x^n\in \Txcond\}} 
\I_{\{\hat{X}^n(m) = \hat{x}^n\}}. \nonumber
\end{align}

\noindent We begin by using the lower bound from (\ref{eq:app:chernoofZavg}) given in Appendix \ref{app:proof:prop:clssubPMF}. 
We have for all sufficiently large $n$:
\begin{align}
\sum_{m =1}^{2^{n{R}}}  E_{M|X^n}(m|x^n)
& =  \left(\frac{1}{1+\eta}\right)\frac{1}{P^{n}_X(x^n)}  \frac{1}{2^{n{R}}}\sum_{m=1}^{2^{n{R}}} Z_m(x^n)\nonumber \\
& \stackrel{w.h.p} \geq  \left(\frac{1-\eta}{1+\eta}\right)\frac{1}{P^{n}_X(x^n)}\EE[Z(x^n)]  \geq \left(\frac{1-\eta}{1+\eta}\right)(1-\epsilon),\label{eq:appxDA_result1}
\end{align}
where the first inequality uses the lower bound  from \eqref{eq:app:chernoofZavg}, which holds true with probability greater than $1-\tau$, where $ \tau \deq 2\exp{\bigg(-\frac{\eta^2 2^{n({R} - I(X,W)-4\delta_1)} }{4}\bigg)} $, and the second inequality uses the following bound on  $\EE[Z(x^n)]$, for all sufficiently large $n$: 
\begin{align}\label{eq:app:lowerboundonZ}
   \frac{1}{\px^n(\xn)} \EE[Z(x^n)] &=  \frac{1}{\px^n(\xn)} \sum_{\xhat^n \in \Txhat} P^n_{\Xhat}(\xhat^n) \prevTC^n(x^n|\hat{x}^n)  
    \I_{\{x^n\in \Txcond\}} \geq (1-\epsilon).
\end{align} 
Using \eqref{eq:app:lowerboundonZ}, we get, with probability greater than $(1-|\Tx|\tau)$ and for all sufficiently large $n$,
\begin{align}
   \widetilde{\zeta} & \leq \!\! \sum_{\xn \in \Tx} 
\px^n(\xn) \Bigg(1 - \round{\frac{1-\eta}{1+\eta}}(1-\epsilon)\Bigg)  \leq 
\frac{2\eta + {\epsilon}(1-\eta)}{1+\eta}.  \nonumber
\end{align}

\noindent Noting that  $\widetilde{\zeta}\cdot\Ipmf\leq 1,$ and using the above result, we have, for all sufficiently large $n$,
\begin{align}\label{eq:SI_ExpecSTilde_Small}
\EE\left[\widetilde{\zeta}\cdot\Ipmf \right] \leq \frac{2\eta + \epsilon(1-\eta)}{1+\eta} + \tau.
\end{align}
Therefore, if $R>I(X;\Xhat) + 4\delta_1$, then $\tau \rightarrow 0$
for all sufficiently large $n$. 
Hence, $\EE\left[\widetilde{\zeta}\cdot\Ipmf\right]$ can be made smaller than $\epsilon$ for all sufficiently large $n$ and sufficiently small $\eta$.
This completes the proof of Proposition \ref{prop:clserrornotcovering}.
\subsection{Proof of Proposition \ref{prop:clserrorcovering}}\label{app:proof:prop:clserrorcovering}
 We begin by writing the PMF $P_M(m)$ for $m \in[2^{nR}]$ and under the condition that $\I_{\curly{\mbox{\normalfont sPMF}}} = 1$.
\begin{align}
    P_M(m) &= \sum_{x^n \in \Tx}\px^n(x^n) E_{M|X^n}(m|x^n),\nonumber \\ 
    &= \sum_{x^n \in \Tx} \sum_{\hat{x}^n \in \Txhat}\!\! \frac{1}{2^{nR}}\frac{(1-\varepsilon)}{(1+\eta)}\prevTC^n(x^n|\hat{x}^n)
  \I_{\{x^n\in \Txcond\}}
  \I_{\{\hat{X}^n(m) = \hat{x}^n\}}.
\end{align}
We now argue that the error term $\zeta$, averaged over the random codebook, can be made arbitrarily small by the following inequalities:
\begin{align*}
&\EE_{\codebook}\sq{\I_{\curly{\mbox{\normalfont sPMF}}} \zeta} \\
    &\overset{}{\leq}  \EE\sq{ \I_{\curly{\mbox{\normalfont sPMF}}} \sum_{ m\in [2^{nR}]}\sum_{\substack{\xhat^n \\ \xn\in\Tx }}\absolute{\px^n(x^n) E_{M|X^n}(m|x^n) -  P_{M}(m)\prevTC^n(x^n|\hat{x}^n)} \I_{\{\hat{X}^n(m) = \hat{x}^n\}}} \\
    &\leq \sum_{m \in [2^{nR}]}\sum_{\substack{\xn\in\Tx \\ \xhat^n }}\codeDistribution(\Xhat^n(m)=\xhat^n)
    \Bigg|\frac{1}{2^{nR}}\frac{(1-\varepsilon)}{(1+\eta)}\prevTC^n(x^n|\hat{x}^n) 
  \I_{\{\hat{x}^n\in \Txhat\}} 
  \I_{\{x^n\in \Txcond\}} \\
  &\hspace{0.7in} -  \round{\sum_{\texttt{x}^n \in \Tx}  \frac{1}{2^{nR}}\frac{(1-\varepsilon)}{(1+\eta)}\prevTC^n(\texttt{x}^n|\hat{x}^n)
  \I_{\{\hat{x}^n\in \Txhat\}} 
  \I_{\{\texttt{x}^n\in \Txcond\}}}\prevTC^n(x^n|\hat{x}^n) \Bigg| \\
  &=\!\!\!\sum_{m \in [2^{nR}]}\sum_{\substack{x^n \in \Tx \\ \hat{x}^n\in \Txhat}} \!\!\!
  \frac{1}{2^{nR}}\frac{1}{(1+\eta)}
    P^n_{\Xhat}(\xhat^n)
  \prevTC^n(x^n|\hat{x}^n) 
  \Bigg|
  \I_{\{x^n\in \Txcond\}} - \!\!\!\!\!\!\sum_{\texttt{x}^n \in \Txcond} \!\!\!\prevTC^n(\texttt{x}^n|\hat{x}^n)\Bigg|\\
  \end{align*}
  \begin{align*}
  &\overset{a} {\leq} \sum_{m \in [2^{nR}]}\sum_{\substack{\hat{x}^n\in \Txhat}} \!\!\!
  \frac{1}{2^{nR}}\frac{1}{(1+\eta)}
    P^n_{\Xhat}(\xhat^n)
   \;  2 \!\!\!\! \sum_{\substack{x^n  \in \Txcond  \\ \texttt{x}^n \not \in \Txcond }}
\prevTC^n(x^n|\hat{x}^n)\prevTC^n(\texttt{x}^n|\hat{x}^n)\\
  &\overset{}{\leq}\sum_{\substack{m \in [2^{nR}]\\ \hat{x}^n\in \Txhat}} \!\!\!
  \frac{1}{2^{nR}}\frac{1}{(1+\eta)}
  P^n_{\Xhat}(\xhat^n)
  \; 2\!\!\!\!\sum_{\substack{\texttt{x}^n \not \in \Txcond}}\prevTC^n(\texttt{x}^n|\hat{x}^n)\\
  &\overset{b}{\leq}\sum_{m \in [2^{nR}]} \frac{1}{2^{nR}} \frac{(1-\varepsilon)}{(1+\eta)} 
 \round{\; \sum_{\hat{x}^n\in \Txhat}  \frac{1}{(1-\varepsilon)} P^n_{\Xhat}(\xhat^n)}
   2\epsilon
\leq  2\epsilon,
\end{align*}
for all sufficiently large $n$ and all $\eta,\delta>0$, where $(a)$ follows by splitting the summation over $x^n \in \Tx$ as summation over $\{x^n \in \Txcond\} \cap \{\xn \in \Tx\}$ and $\{x^n \not \in \Txcond\} \cap \{\xn \in \Tx\} $ 
and $(b)$ follows from the standard conditional typicality argument.
This completes the proof of Proposition \ref{prop:clserrorcovering}.
\subsection{Proof of Lemma \ref{lem:averageSingleletter}}\label{app:lem:proof:nlettergeqavg}
 Consider the following inequalities:
\begin{align}
    \|P_{X^n\hat{X}^n}(x^n,\hat{x}^n) &-
    P_{\hat{X}^n}(\hat{x}^n) \prevTC^n(x^n|\hat{x}^n)
    \|_{\normalfont\text{TV}} \nonumber \\
    &= \frac{1}{2}\sum_{x^n,  \hat{x}^n} \Big(\sum_i \frac{1}{n}\Big) \Big|P_{X^n\hat{X}^n}(x^n,\hat{x}^n) -
    P_{\hat{X}^n}(\hat{x}^n)\prevTC^n(x^n|\hat{x}^n)\Big|\nonumber \\
    &= \frac{1}{2}\sum_{i , x^n,  \hat{x}^n} \frac{1}{n} \Big|P_{X_i\hat{X}_i}(x_i,\hat{x}_i) P_{X^{n\backslash i}\hat{X}^{n\backslash i}|X_i\hat{X}_i}(x^{n\backslash i},\hat{x}^{n\backslash i}|x_i,\hat{x}_i) \nonumber\\
    & \hspace{1.5in}- 
    P_{\hat{X}_i}(\hat{x}_i)\prevTC(x_i|\hat{x}_i) 
    P_{\hat{X}^{n\backslash i}|\hat{X}_i}(\hat{x}^{n\backslash i}|\hat{x}_i) \Pi_{j\neq i} \prevTC(x_j|\hat{x}_j)\Big| \nonumber \\ 
    &\overset{}{\geq}
    \frac{1}{2} \sum_{i, x_i, \hat{x}_i} \frac{1}{n} \Big|P_{X_i\hat{X}_i}(x_i,\hat{x}_i) - 
    P_{\hat{X}_i}(\hat{x}_i)\prevTC(x_i|\hat{x}_i)
    \Big|
    \nonumber \\ 
    &\overset{}{\geq}
    \frac{1}{2} \sum_{x, \hat{x}} \Big|\sum_i  \frac{1}{n} P_{X_i\hat{X}_i}(x,\hat{x}) - \sum_i  \frac{1}{n}
    P_{\hat{X}_i}(\hat{x})\prevTC(x|\hat{x}) 
    \Big|
    \nonumber \\ 
    &= \|P_{X_Q\hat{X}_Q} - P_{\hat{X}_Q}\prevTC\|_{TV} \geq\|P_X - \sum_{\xhat} P_{\hat{X}_Q}(\xhat)\prevTC(\cdot|\xhat) \|_{\normalfont \text{TV}} , \nonumber
\end{align}
where the above inequalities are the consequence of triangle inequality and the fact that $P_{X_Q} = P_X$.

\bibliographystyle{IEEEtran}
\bibliography{references}

\begin{thebibliography}{10}
\providecommand{\url}[1]{#1}
\csname url@rmstyle\endcsname
\providecommand{\newblock}{\relax}
\providecommand{\bibinfo}[2]{#2}
\providecommand\BIBentrySTDinterwordspacing{\spaceskip=0pt\relax}
\providecommand\BIBentryALTinterwordstretchfactor{4}
\providecommand\BIBentryALTinterwordspacing{\spaceskip=\fontdimen2\font plus
\BIBentryALTinterwordstretchfactor\fontdimen3\font minus
  \fontdimen4\font\relax}
\providecommand\BIBforeignlanguage[2]{{%
\expandafter\ifx\csname l@#1\endcsname\relax
\typeout{** WARNING: IEEEtran.bst: No hyphenation pattern has been}%
\typeout{** loaded for the language `#1'. Using the pattern for}%
\typeout{** the default language instead.}%
\else
\language=\csname l@#1\endcsname
\fi
#2}}

\bibitem{schumacher1995quantum}
B.~Schumacher, ``Quantum coding,'' \emph{Physical Review A}, vol.~51, no.~4, p.
  2738, 1995.

\bibitem{jozsa1994new}
R.~Jozsa and B.~Schumacher, ``A new proof of the quantum noiseless coding
  theorem,'' \emph{Journal of Modern Optics}, vol.~41, no.~12, pp. 2343--2349,
  1994.

\bibitem{winter1999coding}
A.~Winter, ``Coding theorems of quantum information theory,'' \emph{arXiv
  preprint quant-ph/9907077}, 1999.

\bibitem{baghali2022strong}
Z.~Baghali~Khanian, ``Strong converse bounds for compression of mixed states,''
  \emph{arXiv e-prints}, pp. arXiv--2206, 2022.

\bibitem{barnum2000quantum}
H.~Barnum, ``Quantum rate-distortion coding,'' \emph{Physical Review A},
  vol.~62, no.~4, p. 042309, 2000.

\bibitem{shannon1959coding}
C.~E. Shannon \emph{et~al.}, ``Coding theorems for a discrete source with a
  fidelity criterion,'' \emph{IRE Nat. Conv. Rec}, vol.~4, no. 142-163, p.~1,
  1959.

\bibitem{datta2012quantum}
N.~Datta, M.-H. Hsieh, and M.~M. Wilde, ``Quantum rate distortion, reverse
  shannon theorems, and source-channel separation,'' \emph{IEEE Transactions on
  Information Theory}, vol.~59, no.~1, pp. 615--630, 2012.

\bibitem{wilde2013auxiliary}
M.~M. Wilde, N.~Datta, M.-H. Hsieh, and A.~Winter, ``Quantum rate-distortion
  coding with auxiliary resources,'' \emph{IEEE Transactions on Information
  Theory}, vol.~59, no.~10, pp. 6755--6773, 2013.

\bibitem{berta2011quantum}
M.~Berta, M.~Christandl, and R.~Renner, ``The quantum reverse shannon theorem
  based on one-shot information theory,'' \emph{Communications in Mathematical
  Physics}, vol. 306, no.~3, pp. 579--615, 2011.

\bibitem{bennett2014quantum}
C.~H. Bennett, I.~Devetak, A.~W. Harrow, P.~W. Shor, and A.~Winter, ``The
  quantum reverse shannon theorem and resource tradeoffs for simulating quantum
  channels,'' \emph{IEEE Transactions on Information Theory}, vol.~60, no.~5,
  pp. 2926--2959, 2014.

\bibitem{datta2013quantum}
N.~Datta, M.-H. Hsieh, M.~M. Wilde, and A.~Winter, ``Quantum-to-classical rate
  distortion coding,'' \emph{Journal of Mathematical Physics}, vol.~54, no.~4,
  p. 042201, 2013.

\bibitem{devetak2008exact}
I.~Devetak and J.~Yard, ``Exact cost of redistributing multipartite quantum
  states,'' \emph{Physical Review Letters}, vol. 100, no.~23, p. 230501, 2008.

\bibitem{luo2009channel}
Z.~Luo and I.~Devetak, ``Channel simulation with quantum side information,''
  \emph{IEEE Transactions on Information Theory}, vol.~55, no.~3, pp.
  1331--1342, 2009.

\bibitem{khanian2021rate}
Z.~B. Khanian and A.~Winter, ``A rate-distortion perspective on quantum state
  redistribution,'' \emph{arXiv preprint arXiv:2112.11952}, 2021.

\bibitem{khanian2022general}
------, ``General mixed-state quantum data compression with and without
  entanglement assistance,'' \emph{IEEE Transactions on Information Theory},
  vol.~68, no.~5, pp. 3130--3138, 2022.

\bibitem{baghali2022rate}
Z.~Baghali~Khanian, K.~Kuroiwa, and D.~Leung, ``Rate-distortion theory for
  mixed states,'' \emph{arXiv e-prints}, pp. arXiv--2208, 2022.

\bibitem{koashi2001compressibility}
M.~Koashi and N.~Imoto, ``Compressibility of quantum mixed-state signals,''
  \emph{Physical Review Letters}, vol.~87, no.~1, p. 017902, 2001.

\bibitem{devetak2002quantum}
I.~Devetak and T.~Berger, ``Quantum rate-distortion theory for memoryless
  sources,'' \emph{IEEE Transactions on Information Theory}, vol.~48, no.~6,
  pp. 1580--1589, 2002.

\bibitem{winter2002compression}
A.~Winter, ``Compression of sources of probability distributions and density
  operators,'' \emph{arXiv preprint quant-ph/0208131}, 2002.

\bibitem{datta2013one}
N.~Datta, J.~M. Renes, R.~Renner, and M.~M. Wilde, ``One-shot lossy quantum
  data compression,'' \emph{IEEE Transactions on Information Theory}, vol.~59,
  no.~12, pp. 8057--8076, 2013.

\bibitem{hsieh2016channel}
M.-H. Hsieh and S.~Watanabe, ``Channel simulation and coded source
  compression,'' \emph{IEEE Transactions on Information Theory}, vol.~62,
  no.~11, pp. 6609--6619, 2016.

\bibitem{salek2018quantum}
S.~Salek, D.~Cadamuro, P.~Kammerlander, and K.~Wiesner, ``Quantum
  rate-distortion coding of relevant information,'' \emph{IEEE Transactions on
  Information Theory}, vol.~65, no.~4, pp. 2603--2613, 2018.

\bibitem{anshu2019convex}
A.~Anshu, R.~Jain, and N.~A. Warsi, ``Convex-split and hypothesis testing
  approach to one-shot quantum measurement compression and randomness
  extraction,'' \emph{IEEE Transactions on Information Theory}, vol.~65, no.~9,
  pp. 5905--5924, 2019.

\bibitem{csiszar2011information}
I.~Csisz{\'a}r and J.~K{\"o}rner, \emph{Information theory: coding theorems for
  discrete memoryless systems}.\hskip 1em plus 0.5em minus 0.4em\relax
  Cambridge University Press, 2011.

\bibitem{holevo1998capacity}
A.~S. Holevo, ``The capacity of the quantum channel with general signal
  states,'' \emph{IEEE Transactions on Information Theory}, vol.~44, no.~1, pp.
  269--273, 1998.

\bibitem{schumacher1997sending}
B.~Schumacher and M.~D. Westmoreland, ``Sending classical information via noisy
  quantum channels,'' \emph{Physical Review A}, vol.~56, no.~1, p. 131, 1997.

\bibitem{winter}
A.~Winter, ``{''Extrinsic'' and ''intrinsic'' data in quantum measurements:
  asymptotic convex decomposition of positive operator valued measures},''
  \emph{Communication in Mathematical Physics}, vol. 244, no.~1, pp. 157--185,
  2004.

\bibitem{holevo2019quantum}
A.~S. Holevo, ``Quantum systems, channels, information,'' in \emph{Quantum
  Systems, Channels, Information}.\hskip 1em plus 0.5em minus 0.4em\relax de
  Gruyter, 2019.

\bibitem{cheng2019duality}
H.-C. Cheng, E.~P. Hanson, N.~Datta, and M.-H. Hsieh, ``Duality between source
  coding with quantum side information and cq channel coding,'' in \emph{2019
  IEEE International Symposium on Information Theory (ISIT)}.\hskip 1em plus
  0.5em minus 0.4em\relax IEEE, 2019, pp. 1142--1146.

\bibitem{lloyd1997capacity}
S.~Lloyd, ``Capacity of the noisy quantum channel,'' \emph{Physical Review A},
  vol.~55, no.~3, p. 1613, 1997.

\bibitem{shor2002quantum}
P.~W. Shor, ``The quantum channel capacity and coherent information,'' in
  \emph{lecture notes, MSRI Workshop on Quantum Computation}, 2002.

\bibitem{devetak2005private}
I.~Devetak, ``The private classical capacity and quantum capacity of a quantum
  channel,'' \emph{IEEE Transactions on Information Theory}, vol.~51, no.~1,
  pp. 44--55, 2005.

\bibitem{hayden2008decoupling}
P.~Hayden, M.~Horodecki, A.~Winter, and J.~Yard, ``A decoupling approach to the
  quantum capacity,'' \emph{Open Systems \& Information Dynamics}, vol.~15,
  no.~01, pp. 7--19, 2008.

\bibitem{nielsen2002quantum}
M.~A. Nielsen and I.~Chuang, ``Quantum computation and quantum information,''
  2002.

\bibitem{berger1975rate}
T.~Berger, ``Rate distortion theory and data compression,'' in \emph{Advances
  in Source Coding}.\hskip 1em plus 0.5em minus 0.4em\relax Springer, 1975, pp.
  1--39.

\bibitem{gallager1968information}
R.~G. Gallager, \emph{Information theory and reliable communication}.\hskip 1em
  plus 0.5em minus 0.4em\relax Springer, 1968, vol. 588.

\bibitem{berger1971rate}
T.~Berger, \emph{Rate Distortion Theory: A Mathematical Basis for Data
  Compression}, ser. Prentice-Hall electrical engineering series.\hskip 1em
  plus 0.5em minus 0.4em\relax Prentice-Hall, 1971.

\bibitem{gerrish1963estimation}
A.~M. Gerrish, ``Estimation of information rates,'' Ph.D. dissertation, Yale
  University, 1963.

\bibitem{shamai1997empirical}
S.~Shamai and S.~Verd{\'u}, ``The empirical distribution of good codes,''
  \emph{IEEE Transactions on Information Theory}, vol.~43, no.~3, pp. 836--846,
  1997.

\bibitem{pradhan2004approximation}
S.~S. Pradhan, ``Approximation of test channels in source coding,'' in
  \emph{Proc. Conf. Inform. Syst. Sci.(CISS)}, 2004.

\bibitem{weissman2005empirical}
T.~Weissman and E.~Ordentlich, ``The empirical distribution of rate-constrained
  source codes,'' \emph{IEEE transactions on information theory}, vol.~51,
  no.~11, pp. 3718--3733, 2005.

\bibitem{cuff2010coordination}
P.~W. Cuff, H.~H. Permuter, and T.~M. Cover, ``Coordination capacity,''
  \emph{IEEE Transactions on Information Theory}, vol.~56, no.~9, pp.
  4181--4206, 2010.

\bibitem{schieler2013connection}
C.~Schieler and P.~Cuff, ``A connection between good rate-distortion codes and
  backward dmcs,'' in \emph{2013 IEEE Information Theory Workshop (ITW)}.\hskip
  1em plus 0.5em minus 0.4em\relax IEEE, 2013, pp. 1--5.

\bibitem{kostina2015output}
V.~Kostina and S.~Verd{\'u}, ``The output distribution of good lossy source
  codes,'' in \emph{2015 Information Theory and Applications Workshop
  (ITA)}.\hskip 1em plus 0.5em minus 0.4em\relax IEEE, 2015, pp. 308--312.

\bibitem{poor1998introduction}
H.~V. Poor, \emph{An introduction to signal detection and estimation}.\hskip
  1em plus 0.5em minus 0.4em\relax Springer Science \& Business Media, 1998.

\bibitem{petz1986sufficient}
D.~Petz, ``Sufficient subalgebras and the relative entropy of states of a von
  neumann algebra,'' \emph{Communications in mathematical physics}, vol. 105,
  no.~1, pp. 123--131, 1986.

\bibitem{barnum2002reversing}
H.~Barnum and E.~Knill, ``Reversing quantum dynamics with near-optimal quantum
  and classical fidelity,'' \emph{Journal of Mathematical Physics}, vol.~43,
  no.~5, pp. 2097--2106, 2002.

\bibitem{hayden2004structure}
P.~Hayden, R.~Jozsa, D.~Petz, and A.~Winter, ``Structure of states which
  satisfy strong subadditivity of quantum entropy with equality,''
  \emph{Communications in mathematical physics}, vol. 246, no.~2, pp. 359--374,
  2004.

\bibitem{uhlmann1976transition}
A.~Uhlmann, ``The “transition probability” in the state space of
  a*-algebra,'' \emph{Reports on Mathematical Physics}, vol.~9, no.~2, pp.
  273--279, 1976.

\bibitem{wilde_arxivBook}
M.~M. Wilde, ``From classical to quantum shannon theory,'' \emph{arXiv preprint
  arXiv:1106.1445}, 2011.

\bibitem{cuff2013distributed}
P.~Cuff, ``Distributed channel synthesis,'' \emph{IEEE Transactions on
  Information Theory}, vol.~59, no.~11, pp. 7071--7096, 2013.

\bibitem{atif2022source}
T.~A. Atif, A.~Padakandla, and S.~S. Pradhan, ``Source coding for synthesizing
  correlated randomness,'' \emph{IEEE Transactions on Information Theory},
  2022.

\bibitem{fuchs1999cryptographic}
C.~A. Fuchs and J.~Van De~Graaf, ``Cryptographic distinguishability measures
  for quantum-mechanical states,'' \emph{IEEE Transactions on Information
  Theory}, vol.~45, no.~4, pp. 1216--1227, 1999.

\bibitem{neumark1940spectral}
M.~Naimark, ``Spectral functions of a symmetric operator,'' \emph{Bull. Acad.
  Sci. URSS. S{\'e}r. Math.[Izvestia Akad. Nauk SSSR]}, vol.~4, pp. 277--318,
  1940.

\bibitem{wilde2013sequential}
M.~M. Wilde, ``Sequential decoding of a general classical-quantum channel,''
  \emph{Proceedings of the Royal Society A: Mathematical, Physical and
  Engineering Sciences}, vol. 469, no. 2157, p. 20130259, 2013.

\bibitem{mattson2012upper}
H.~Mattson~Jr, ``An upper bound on covering radius,'' in \emph{Combinatorial
  Mathematics Proceedings of the International Colloquium on Graph Theory and
  Combinatorics}, vol.~75, 2012, pp. 453--458.

\bibitem{cover2006elements}
T.~M. Cover and J.~A. Thomas, ``Elements of information theory 2nd edition
  (wiley series in telecommunications and signal processing),'' \emph{Acessado
  em}, 2006.

\bibitem{wilde_e}
M.~M. Wilde, P.~Hayden, F.~Buscemi, and M.-H. Hsieh, ``The
  information-theoretic costs of simulating quantum measurements,''
  \emph{Journal of Physics A: Mathematical and Theoretical}, vol.~45, no.~45,
  p. 453001, 2012.

\bibitem{anshu2017measurement}
A.~Anshu, R.~Jain, and N.~Warsi, ``Measurement compression with quantum side
  information using shared randomness,'' \emph{arXiv preprint
  arXiv:1703.02342}, 2017.

\bibitem{fourier-motzkin}
G.~M. Ziegler, \emph{Lectures on polytopes}.\hskip 1em plus 0.5em minus
  0.4em\relax Springer Science \& Business Media, 2012, vol. 152.

\bibitem{dueck1981strong}
G.~Dueck, ``The strong converse to the coding theorem for the multiple--access
  channel,'' \emph{J. Comb. Inform. Syst. Sci}, vol.~6, no.~3, pp. 187--196,
  1981.

\bibitem{ahlswede2006strong}
R.~Ahlswede and N.~Cai, ``A strong converse theorem for quantum multiple access
  channels,'' in \emph{General Theory of Information Transfer and
  Combinatorics}.\hskip 1em plus 0.5em minus 0.4em\relax Springer, 2006, pp.
  460--485.

\bibitem{powers1970free}
R.~T. Powers and E.~St{\o}rmer, ``Free states of the canonical anticommutation
  relations,'' \emph{Communications in Mathematical Physics}, vol.~16, no.~1,
  pp. 1--33, 1970.

\bibitem{ahlswede2002strong}
R.~Ahlswede and A.~Winter, ``Strong converse for identification via quantum
  channels,'' \emph{IEEE Transactions on Information Theory}, vol.~48, no.~3,
  pp. 569--579, 2002.

\end{thebibliography}

\end{document}